\documentclass{article}

\usepackage{natbib}
\usepackage[htt]{hyphenat}
\usepackage{enumitem}
\usepackage{makecell}
\usepackage[below]{placeins}
\usepackage{tikz}
\usetikzlibrary{positioning}

\usepackage{mathtools}

\usepackage{soul}

\usepackage{minitoc}

\usepackage[utf8]{inputenc}
\usepackage[english]{babel}

\usepackage{verbatim}

\usepackage{geometry}

\usepackage[export]{adjustbox}

\usepackage[htt]{hyphenat}

\usepackage{setspace}

\usepackage{caption}
\usepackage[labelformat=simple]{subcaption}

\usepackage{tabularx}
\usepackage{multirow}
\usepackage{multicol}
\usepackage{booktabs}

\usepackage{amsmath}
\usepackage{amsthm}
\usepackage{amssymb}
\usepackage{bm}

\usepackage{xcolor}

\usepackage{tikz}
\usetikzlibrary{graphs}

\usepackage{graphicx}
\graphicspath{{./figures/}}

\usepackage[ruled]{algorithm}
\usepackage{algpseudocode}
\algnewcommand{\lIf}[1]{\State\algorithmicif\ #1\ \algorithmicthen}
\algnewcommand{\EndlIf}{\unskip\ \algorithmicend\ \algorithmicif}

\usepackage[title]{appendix}

\usepackage{fancyhdr}

\usepackage[pages=some,placement=top]{background}

\usepackage{pdfpages}

\usepackage[bookmarks=false]{hyperref}

\DeclareMathOperator*{\argmax}{argmax}

\newcommand{\R}{\mathbb{R}}

\newcommand{\N}{\mathbb{N}}

\newcommand{\indep}{\mathrel{\perp\!\!\!\perp}}

\newcommand{\sctlabel}[1]{\label{sct:#1}}
\newcommand{\sctref}[1]{Section \ref{sct:#1}}

\newcommand{\figlabel}[1]{\label{fig:#1}}
\newcommand{\figref}[1]{Figure \ref{fig:#1}}
\newcommand{\figrefmult}[2]{Figures \ref{fig:#1} and \ref{fig:#2}}

\newcommand{\apxlabel}[1]{\label{apx:#1}}
\newcommand{\apxref}[1]{Appendix \ref{apx:#1}}

\newcommand{\tbllabel}[1]{\label{tbl:#1}}
\newcommand{\tblref}[1]{Table \ref{tbl:#1}}

\newcommand{\eqnlabel}[1]{\label{eqt:#1}}
\newcommand{\eqnref}[1]{(\ref{eqt:#1})}

\newcommand{\say}[1]{``#1''}

\newtheorem{theorem}{Theorem}[section]

\newtheorem{prop}[theorem]{Proposition}

\theoremstyle{definition}

\NewCommandCopy{\proofqedsymbol}{\qedsymbol}

\AtBeginEnvironment{proof}{\renewcommand{\qedsymbol}{\proofqedsymbol}}

\theoremstyle{definition}
\newtheorem{example}{Example}[section]
\AtBeginEnvironment{example}{%
  \pushQED{\qed}\renewcommand{\qedsymbol}{$\triangle$}%
}
\AtEndEnvironment{example}{\popQED\endexample}

\newcommand{\howtobeginartappendix}{\begin{appendices}}
\newcommand{\howtoendartappendix}{\end{appendices}}
\newenvironment{artappendix}[1][A]{
    \howtobeginartappendix
    \setcounter{equation}{0}
    \setcounter{figure}{0}
    \setcounter{table}{0}
    \setcounter{theorem}{0}
    \setcounter{example}{0}
    \setcounter{definition}{0}
    \setcounter{algorithm}{0}
    \renewcommand{\theequation}{{#1}\arabic{equation}}
    \renewcommand{\thefigure}{{#1}\arabic{figure}}
    \renewcommand{\thetable}{{#1}\arabic{table}}
    \renewcommand{\thelemma}{{#1}\arabic{lemma}}
    \renewcommand{\theprop}{{#1}\arabic{prop}}
    \renewcommand{\thetheorem}{{#1}\arabic{theorem}}
    \renewcommand{\thecorollary}{{#1}\arabic{corollary}}
    \renewcommand{\theexample}{{#1}\arabic{example}}
    \renewcommand{\thedefinition}{{#1}\arabic{definition}}
    \renewcommand{\thealgorithm}{{#1}\arabic{algorithm}}
}{
    \howtoendartappendix
}

\fancypagestyle{plain}{%
\fancyhf{}
\fancyfoot[C]{\iffloatpage{\raisebox{-3cm}{\thepage}}{\thepage}}

}
\renewcommand{\howtobeginartappendix}{\appendix}
\renewcommand{\howtoendartappendix}{}

\makeatletter
\renewcommand{\p@subfigure}{\thefigure, panel }
\makeatother

\setlistdepth{9}

\newlist{denumerate}{enumerate}{6}
\setlist[denumerate,1]{label=\arabic*.}
\setlist[denumerate,2]{label=(\alph*)}
\setlist[denumerate,3]{label=\roman*.}
\setlist[denumerate,4]{label=(\Alph*)}
\setlist[denumerate,5]{label=\Roman*.}
\setlist[denumerate,6]{label=\alph*.}

\newcommand{\nrasem}{no\--repeated\--anything graph pattern evaluation semantics}

\title{tmQM-RDF Dataset: a Knowledge Graph Representing Transition Metal Complexes}

\newcommand{\affil}[1]{\\\footnotesize #1}
\newcommand{\affilIntegreat}{\affil{Integreat -- Norwegian Centre for Knowledge-Driven Machine Learning}}
\newcommand{\affilUIOMI}{\affil{University of Oslo -- Department of Mathematics}}
\newcommand{\affilBU}{\affil{Bielefeld University -- Center for Cognitive Interaction Technology}}
\newcommand{\affilUIOIFI}{\affil{University of Oslo -- Department of Informatics}}
\newcommand{\affilUIOKI}{\affil{University of Oslo -- Department of Chemistry}}
\newcommand{\affilIcredd}{\affil{Hokkaido University -- Institute for Chemical Reaction Design and Discovery (WPI-ICReDD)}}
\newcommand{\corrEmail}{\\{lucaci@math.uio.no}}
\author{
    Luca Cibinel\corrEmail\affilIntegreat\affilUIOMI\and
    Trond Linjordet\affilIcredd\and
    Johan Pensar\affilIntegreat\affilUIOMI\and
    David Balcells\affilUIOKI\and
    Riccardo De Bin\affilIntegreat\affilUIOMI\and
    Basil Ell\affilIntegreat\affilBU\affilUIOIFI
}

\date{February 2026}

\mathchardef\breakingcomma\mathcode`\,
{\catcode`,=\active
  \gdef,{\breakingcomma\discretionary{}{}{}}
}

\begin{document}

\doparttoc
\faketableofcontents

\maketitle

\begin{abstract}
Transition Metal Complexes (TMCs) have wide-ranging practical utility in chemistry, with possible applications that range from catalysis to medicinal chemistry. The study of TMCs and their properties is thus a field rich with potential, one in which machine learning and computational approaches can offer a substantial aid. For this reason, appropriate and accessible datasets, collecting a wide range of information, are required in order to facilitate the effective analysis and investigation of such compounds. This paper contributes to the data modelling effort via the introduction of the \emph{transition metal quantum mechanics RDF} (tmQM-RDF) dataset, a knowledge graph constructed using the Resource Description Framework (RDF) vocabulary which collects rich and detailed descriptions of approximately 50k TMCs. These descriptions are both qualitative and quantitative in nature, encompassing the compositional nature of TMCs in terms of their constituting ligands, as well as the entirety of their molecular graphs. An example of the power of the proposed representation is presented, showcasing how the information available in tmQM-RDF can be exploited for TMC manipulation tasks, achieving promising performance even with relatively simple probabilistic models.
\end{abstract}

\section{Introduction}
Transition metal complexes (TMCs) are chemical compounds with a widespread range of applications. Fields such as catalysis, medicine and material science can gain significant benefits from advancements in our understanding of TMCs \citep{abu2015tmcreview, liang2023transition}.

What makes this chemical family so broadly relevant is the possibility of adjusting their properties by appropriate manipulation of the bond between the transition metal centre and the ligands \citep{nandy2021computational}. At the same time,  the involvement of d orbitals in these bonds make TMCs more difficult to represent as molecular graphs, compared to other classes of molecules \citep{gee2025enriched,kneiding2023learning}. This is one of the key challenges that arise in the study of TMCs. Another fundamental obstacle is represented by the combinatoric explosion of the number of possible chemical structures, due to the numerous possible assemblies of the central transition metal atom with the surrounding ligands \citep{nandy2021computational}.

In order to tackle and overcome these issues, computational techniques have become increasingly predominant in the last decades. Density functional theory (DFT) is one of the fundamental pillars of modern computational chemistry, as it enables the preliminary screening of compounds of interests on the basis of their electronic properties \citep{griego2021acceleration}. More recently, machine learning (ML) approaches have attracted considerable interest as surrogates for DFT simulations, which can easily become prohibitively expensive in computational costs \citep{garrison2023applying, gee2025enriched}. ML models , on the other hand, have considerably lower inference times, making them an attractive solution for tasks like property prediction, electronic structure calculation or even structure generation \citep{goh2017deepchem, jin2019junction, gallegos2024explainable}.

This new paradigm, however, comes with its own set of caveats. First, popular ML algorithms, like deep neural networks (DNNs), are notorious black boxes, and do not allow for their parameters to be intuitively interpreted \citep{goh2017deepchem, gallegos2024explainable}. Second, acceptable performances are typically achieved only when an adequately large amount of data is available for the training of the model. This is perhaps one of the most crucial aspects of ML and DNNs, as the quantity, quality and chemical diversity of the training data can have a dramatic impact on the results \citep{lavravc2021representation, garrison2023applying, heid2023uncertainty}. 

Notable efforts have been dedicated, then, to the curation of suitable datasets which can provide a solid foundation for ML training. In the case of TMCs, the most prominent representative of this effort is the tmQM dataset series \citep{balcells2020tmqm, kneiding2023learning, kneiding2024optmization, kneiding2025tmqmgstar, kneiding2026delta}, which describes a subset of the TMCs reported in the Cambridge Structural Database \citep[CSD;][]{groom2016cambridge}. Among its many contributions, we primarily highlight the merit of having provided, with tmQM, one of the first large-scale datasets on the geometric and electronic properties of TMCs \citep{balcells2020tmqm}, while also tackling the problem of molecular graph representation in tmQMg \citep{kneiding2023learning} and that of extensive ligand description in tmQMg-L \citep{kneiding2024optmization}. The intrinsic power and utility of this dataset series have been demonstrated by the numerous applications and extensions found in the literature \citep{garrison2023applying, kevlishvili2024language, gee2025enriched, meggio2026path}. 

Nevertheless, diversity and size are not the only desirable properties of a dataset. Accessibility, ease of manipulation and compatibility with other resources are also important factors that can affect the overall impact of a dataset on the research process \citep{kanza2019new}. Data integration and machine readability are the main focus of the Semantic Web Standards, centred around the objective of creating a \say{Web of Data} that is easily accessible, and most importantly interpretable, to both humans and machines \citep{bratt2005semanticweb}. This kind of methodology prompted the development of a long-standing series of computational tools for chemical research \citep{gordon1988chemical, borkum2014usage, kanza2019new}. In addition, several chemical datasets and ontologies, i.e., collections of large-scale standardised representations, have been developed over the years, contributing to the development of computational chemistry under this  framework \citep{hastings2011chemical,djoumbou2016classyfire,pascazio2023chemical, he2024reaction}. 

This work participates in the data representation effort by building on top of the tmQM dataset series, specifically on the three datasets tmQM, tmQMg and tmQMg-L. Many of the key challenges in modern data collection, such as chemical variety, ML-readiness and appropriate chemical representation, are already addressed \citep{balcells2020tmqm, kneiding2023learning, kneiding2024optmization}, and we contribute to the existing resources by organising the content of the three datasets into a single coherent representation that can facilitate data access. Inspired by the aforementioned principles of data availability, readability and operability, we present the tmQM-RDF dataset, an integrated description of a large (approximately 50k entries) and diverse population of TMCs, encoded using the Resource Description Framework (RDF) vocabulary \citep{klyne2004rdfsyntax} and its RDF Schema (RDFS) semantic extension \citep{hayes2004rdfsemantics, brickley2014rdfs}.

By leveraging the RDF and RDFS vocabulary, tmQM-RDF falls into the category of knowledge graphs \citep[KGs;][]{hogan2021knowledge}, where the fundamental units of information are triples of the form $(t_s, t_p, t_o) \in \mathcal{T}^3$, where $\mathcal{T}$ is a (possibly infinite) set of terms. Each triple represents a subject-predicate-object structure, but indeed it can also be interpreted as a directed labelled edge, where the label is the \emph{predicate} $t_p$, between two nodes, namely the \emph{subject} $t_s$ and the \emph{object} $t_o$. In tmQM-RDF, then, knowledge about TMCs is easily accessible using this dual interpretation.

We furthermore contribute an example application of the novel integrated data representation, and define the task of \emph{plausible TMC completion}. By extracting frequent structural motifs from the dataset and then estimating their joint distribution using statistical techniques, we show how to quantitatively assess the appropriateness of a candidate TMC structure considered as a completion of an initial molecular scaffold.

The remainder of this paper is organised as follows. \sctref{methods} introduces the vocabulary used to build the tmQM-RDF KG. \sctref{tmQM-RDF} presents a detailed account of how tmQM-RDF is built using  the information from the tmQM series. In \sctref{experiments} we present and discuss our example on plausible TMC reconstruction powered by the structural analysis of the content of tmQM-RDF. Finally, in \sctref{conclusions} we summarise our main contributions. Additional content, such as notable subsets of tmQM-RDF, and technical details about the creation of the dataset and the example application, can be found in the Appendix.  

\section{Methods}\sctlabel{methods}
The act of converting existing chemical datasets into RDF format is not unprecedented and several notable examples can be found in the literature \citep{belleau2008biotordf, willighagen2013chembl, fu2015pubchemrdf}. The creation of tmQM-RDF falls within this category, as it requires a careful examination of the data contained in the tmQM dataset series and the elaboration of an appropriate \emph{terminology component}, or \emph{TBox} \citep{degiacomo1996tbox}. This section focuses on this aspect of data modelling, by introducing the challenges involved in designing the semantic representation of TMCs and the proposed solution, intended as an RDF vocabulary.

In accordance with the RDF syntax, we use Uniform Resource Identifiers (URIs) to identify entities and resources, blank nodes to state the existence of entities for which it is not necessary to state their precise identity \footnote{In other words, they behave as \emph{existentially quanitified variables} \citep{hayes2004rdfsemantics, bergmann2014logic}.}, and literals to report quantitative data.

For the sake of brevity, we adopt the common convention of shortening URIs by employing prefix names. An excerpt of the prefixes we use is available in \tblref{prefixes_short}. For a full list, the reader is referred to \apxref{comprehensive}. In general, the newly introduced prefixes adhere to the following scheme: the first two lowercase letters denote the general topic (the whole \underline{\textbf{c}}o\underline{\textbf{m}}plex, \underline{\textbf{l}}i\underline{\textbf{g}}and-like structures, \underline{\textbf{a}}to\underline{\textbf{m}}ic structures and \underline{\textbf{d}}ata\underline{\textbf{s}}ets), the following uppercase letter specifies the first subtopic (\underline{\textbf{T}}MC, \underline{\textbf{B}}onds, metal \underline{\textbf{C}}entre, \underline{\textbf{L}}igand, \underline{\textbf{S}}tructural connectivity, \underline{\textbf{A}}tom). A variable number of lowercase letters can follow, indicating further specification.

\begin{table}[!t]
    \footnotesize
    \centering
    \begin{tabular}{l l l}
         \toprule
         Prefix & Namespace & \\
         \midrule
         \midrule
         rdf & http://www.w3.org/1999/02/22-rdf-syntax-ns\# & \\
         rdfs & http://www.w3.org/2000/01/rdf-schema\# & \\
         xmls & http://www.w3.org/2001/XMLSchema\# & \\
         inp5 & resource://integreat/p5/ & \\
         nm & resource://integreat/p5/numerical/ & \\[0.15cm]
         \multicolumn{3}{c}{\emph{Families of namespaces}}\\
         \midrule
          & & Relevant prefixes\\
         \midrule
         cm\texttt{*} & resource://integreat/p5/complex/\texttt{*}  & cmT\ \ cmTp\\
         lg\texttt{*} & resource://integreat/p5/ligand/\texttt{*}  & lgB\ \ lgC\ \ lgCr\ \ lgL\ \ lgLr\ \ lgS\\
         tm\texttt{*} & resource://integreat/p5/atomic/\texttt{*} & tmA\ \ tmAr\ \ tmB\ \ tmBp\ \ tmS\\
         ds\texttt{*} & resource://integreat/p5/datasets/\texttt{*}  & ds\\
         \bottomrule
    \end{tabular}
    \caption{In RDF, entities are represented via URIs, which are often organised by of namespaces. These namespaces can employ prefixes as their alias, as a way to shorten URIs. This table summarizes the namespace prefixes used in this work. To improve readability, certain groups of prefixes are not listed entirely, but rather are summarised using the symbol \texttt{*} as a placeholder for a (possibly empty) sequence of characters. For each group, a (non-exhaustive) list of the most important prefixes is reported. See \apxref{methods_rdf_uri} and \tblref{prefixes} for a complete discussion.}
    \tbllabel{prefixes_short}
\end{table}

\subsection{The semantic integration of the tmQM series}\sctlabel{methods_semantic}
The vocabulary that we employ to semantically integrate the datasets of the tmQM series must be tailored to the data , in order to appropriately convey the information once the unified representation is assembled. The content of tmQM, tmQMg and tmQMg-L is then hereby briefly introduced, followed by an explanation of how the data is organised in a single, coherent KG.

\subsubsection{The tmQM dataset series}
 In this work we consider three of the datasets in the tmQM series, namely tmQM \citep{balcells2020tmqm}, tmQMg \citep{kneiding2023learning} and tmQMg-L \citep{kneiding2024optmization}.\footnote{For brevity, we shall still refer to the group formed by tmQM, tmQMg and tmQMg-L as the \emph{tmQM dataset series}. There is no risk of ambiguity, as no other dataset is included in this work.} These three datasets are closely related and describe the same population of TMCs, but each one provides a different perspective:

\begin{itemize}
    \item tmQM describes TMCs primarily at a whole-complex level, providing quantum properties computed at the DFT(TPSSh-D3BJ/def2-SVP) level of theory, based on geometric data optimised at the GFN2-xTB level. Geometric information, in the form of atomic Cartesian coordinates, are also available.

    \item tmQMg provides some TMC-level quantum properties as well, but unlike tmQM, this dataset focuses on the atomic composition of the TMCs. In particular, \citet{kneiding2023learning} developed two different graphical representations to analyse TMCs from a graph-theoretical point of view, namely the directed and the undirected natural quantum graph. For the purpose of this work, we only employ the undirected version, as this is the one that most closely relates to the classical molecular graph representation. Calculations of geometries, frequencies and thermochemistry properties in this dataset have been conducted at the PBE-D3BJ/def2-SVP level, whereas single-point energy and natural bond orbital (NBO) calculations have been performed at the PBE0-D3BJ/def2-TZVP level. Notice that this dataset only contains complexes with at most 85 atoms, hence this constraint will be transmitted to tmQM-RDF as well.

    \item tmQMg-L acts as a ligand library in which extensive properties, descriptors and features belonging to ligands, intended as independent entities, are reported. We shall treat the entries of this dataset as abstract reference definitions of species of ligands, in a similar way in which chemical elements in the periodic table define species of atoms \citep{iupac2025element}.
\end{itemize}

\subsubsection{A hierarchical representation}\sctlabel{semantic_hierarchical}
We decided to organise the data using a three-level hierarchical representation, in which each level describes the features of a TMC at a  different resolution. Each TMC, i.e., each corresponding subgraph of the resulting KG, will be constructed using this scheme.

The top level is the \emph{complex} level, in which we only consider properties that belong to the entire TMC. Here we derive the data mostly from tmQM, but partly also from tmQMg, as it also contains whole-TMC features. Below this level we find the \emph{ligand} level, which specifies the composition of the TMC in terms of its constituting ligands and the metal centre, in the \mbox{tmQMg-L} sense. Here we report only structural information, as numerical features are treated either as pertaining to specific atomic bonds or to the general representation of the entire ligand. In the latter case, these properties are specified elsewhere, in what can be considered the ligand-level equivalent of the periodic table (see \sctref{tbox_ligand_pt}). This level draws information from \mbox{tmQMg-L}, whose entries are compared against the information in tmQMg to correctly locate the ligands within each TMC (see \sctref{tmQM-RDF}). Finally, the most detailed level is the \emph{atomic} level, which essentially corresponds to the molecular graph of the TMC, enriched with all the relevant node and edge features, as extracted from tmQMg.

This is summarised in \figref{dataset_scheme}, with a brief example of the three-level representation being shown in \figref{abox}. This example makes use of the RDF vocabulary developed below, in \sctref{methods_tbox}, and further explained therein.

\subsection{The terminology component of tmQM-RDF}\sctlabel{methods_tbox}
In order to achieve a functional RDF representation that adheres to the principles of the three-level hierarchical representation, there are several modelling decisions that have to be carefully considered.

With regards to the encoding of the structural component of TMCs, appropriate classes and predicates have to be introduced to represent:
\begin{enumerate}
    \item the \emph{TMC}, as an entity to which properties that apply to the complex as a whole can be attached;
    \item the \emph{ligands} and the \emph{metal centre}, as abstract structural components and specific instances of those abstractions;
    \item the \emph{ligand-metal centre bonds}, as entities representing the bond between structural components and its properties;
    \item the \emph{atoms}, both as abstractions and physical instances, similarly to the ligands and the metal atom;
    \item the \emph{atomic bonds}, as representations of the bond between atoms and its properties;
    \item the \emph{connections between levels}, intended as containment relationships meant to specify which structural elements belong to which higher level structure (e.g., which atoms are part of a given ligand).
\end{enumerate}
Structural information must then be enriched with quantitative and qualitative features and properties, which poses a new set of challenges:
\begin{enumerate}
    \item introducing a \emph{general property description scheme};
    \item appropriately describing \emph{non-elementary features} (intended as features which cannot be represented with a single numerical or categorical value);
    \item separating \emph{abstract ligand properties} from those of actual instances of chemical entities;
    \item appropriately describing \emph{TMC metadata};
    \item appropriately referencing \emph{information for computational reproducibility}.
\end{enumerate}

\begin{figure}
    \centering
    \makebox[\textwidth][c]{%
    \includegraphics[width=1.25\linewidth]{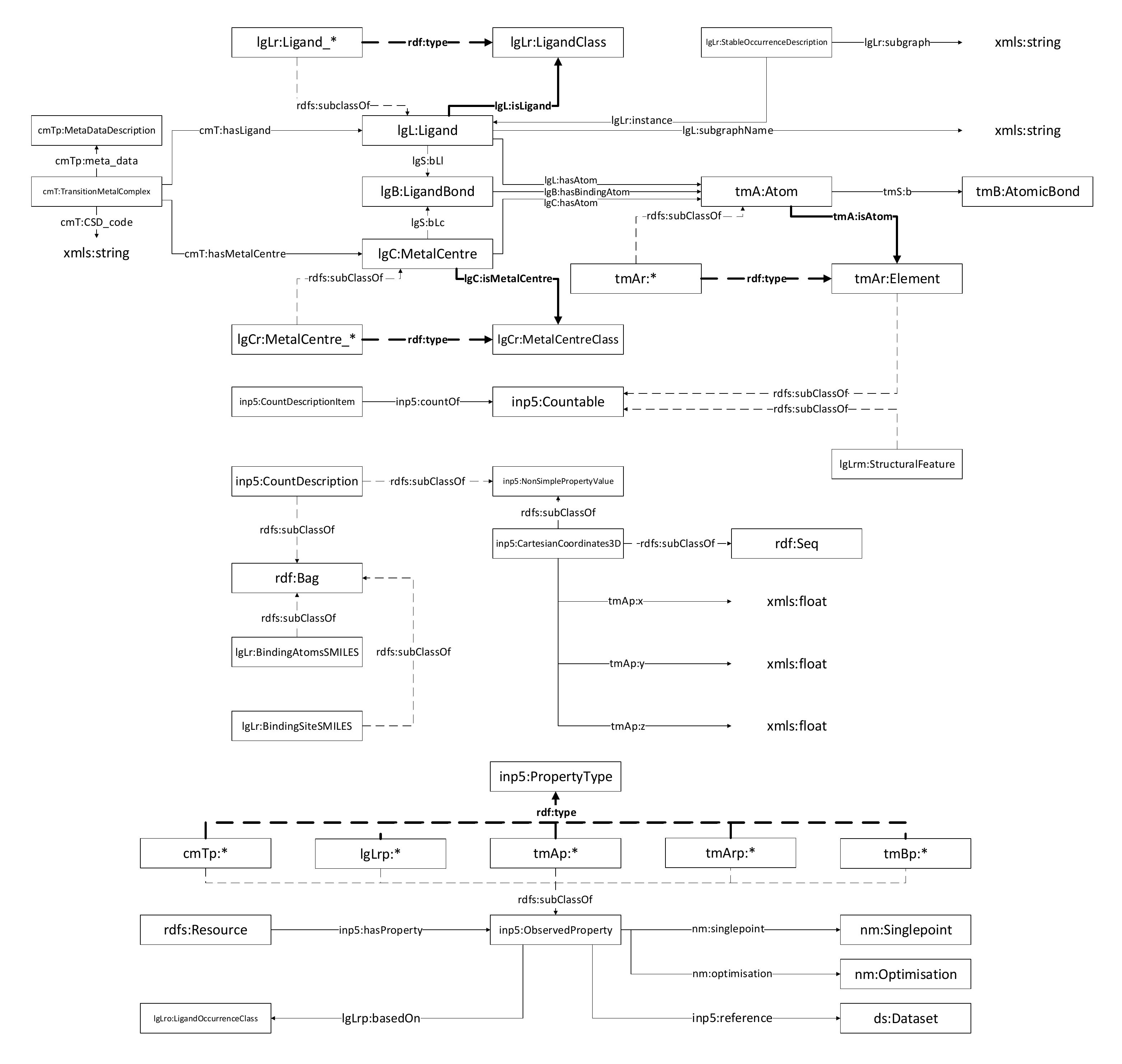}
    }
    \caption{Statements about general concepts are collected in the TBox. Here the TBox of tmQM-RDF is visually represented. Nodes represent the available classes. Solid edges represent the available predicates, where the tail and the head of the edge represent domain and range restrictions for that predicate. Dashed edges represent class-related assertions, i.e. subclass relationships between classes or class assignments. Bold edges highlight the predicate \texttt{rdf:type} and its subproperties (notice that an edge can be both solid and bold or both dashed and bold). The symbol \texttt{*} is used as a placeholder for a sequence of charachters, representing chemical elements symbols (in \texttt{tmAr:*} and \texttt{lgCr:MetalCentre\_*}), ligand ids (in \texttt{lgLr:Ligand\_*}) or property names (in \texttt{cmTp:*}, \texttt{lgLrp:*}, \texttt{tmAp:*}, \texttt{tmArp:*} and \texttt{tmBp:*}).}
    \figlabel{tbox}
\end{figure}

We will consider each of these points in detail in the rest of the section. A summary of the resulting TBox is schematised in \figref{tbox}. An example of a compliant \emph{assertion component}, or \emph{ABox}, i.e., a set of assertions on actual objects that follows the specification of the TBox \citep{degiacomo1996tbox}, is shown in \figref{abox} and will be referenced throughout the explanation below to provide illustrations for the most fundamental concepts. Contextually, this example ABox will also demonstrate another feature introduced in tmQM-RDF, namely that the TMCs and all the substructures therein are referenced using the CSD code (six alphabetic characters) that identifies the TMC. For instance, within the subgraph of the tmQM-RDF KG that describes the TMC that in the CSD is listed as KCEYPT, all the URIs that pertain to ligands and atoms of this TMC will include the KCEYPT code.

\begin{figure}[p]
    \begin{subfigure}[t]{\linewidth}
        \centering
        \includegraphics[width=0.8\linewidth]{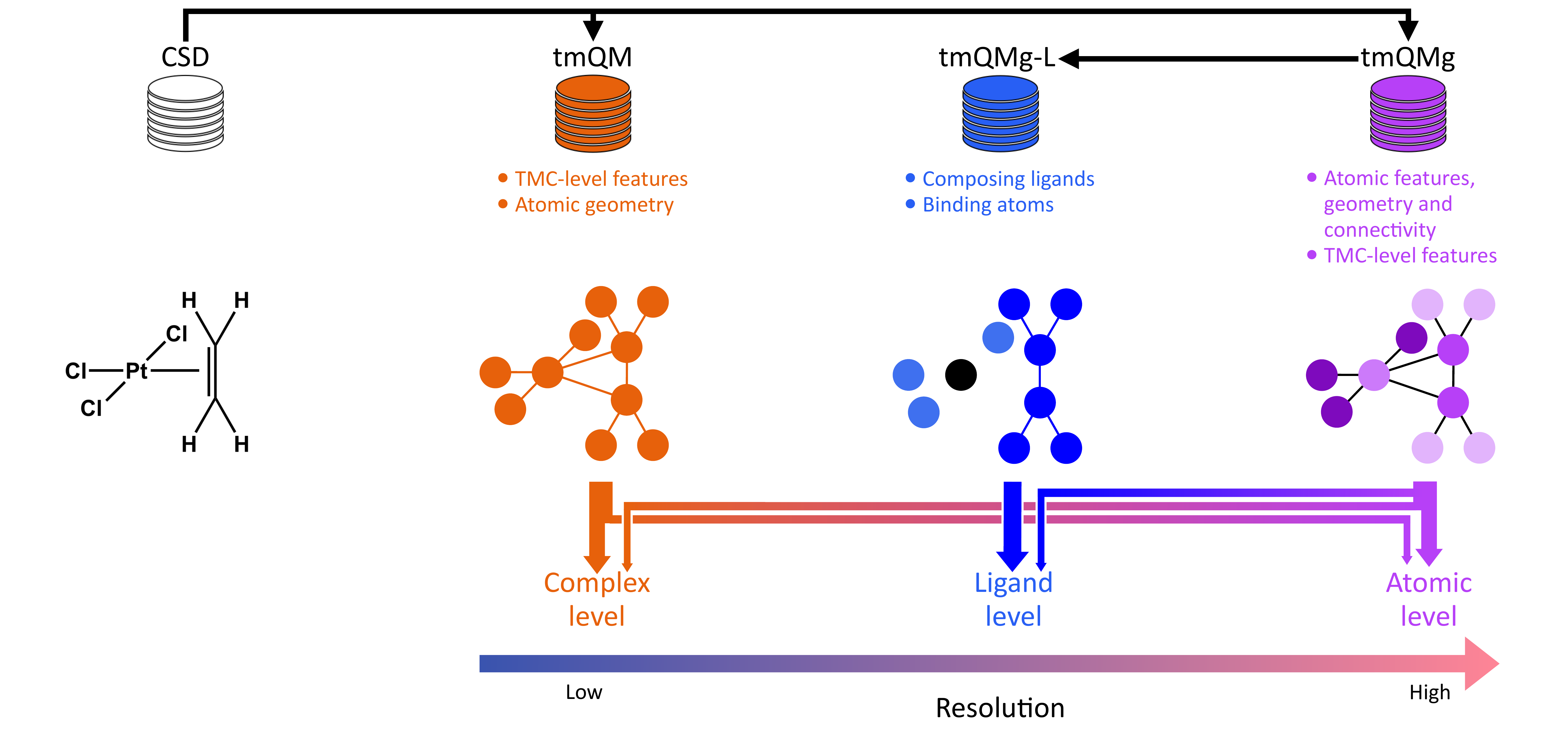}
        \caption{}
        \figlabel{dataset_scheme}
    \end{subfigure}
    \begin{subfigure}[t]{\linewidth}
        \centering
        \makebox[\textwidth][c]{%
        \includegraphics[width=\linewidth]{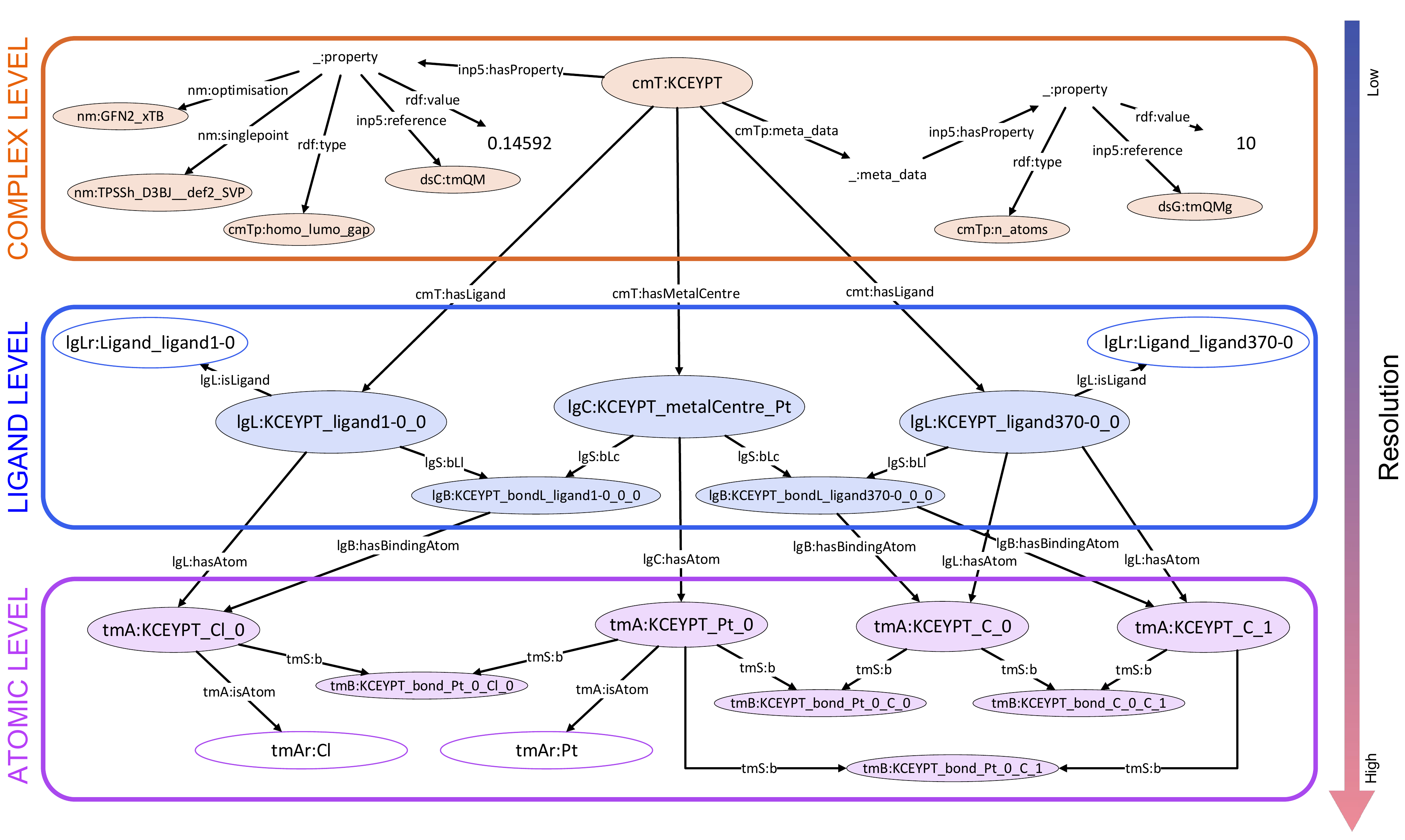}
        }
        \caption{}
        \figlabel{abox}
    \end{subfigure}
    \caption{(a) The tmQM-RDF KG is the result of the integration of tmQM, tmQMg and tmQMg-L. This scheme illustrates syntetically how exactly each dataset contributes to the final three-level representation. (b) A visual example of an ABox compliant with the TBox in \figref{tbox} showcasing how the data from the tmQM series can be represented. Nodes with neither background nor border represent either blank nodes (\texttt{\_:*}) or literals (\say{\texttt{*}}). Nodes with a white background represent classes. The remaining nodes represent instances of classes. For the sake of readability, not all the features, or literal datatypes, are represented.}
\end{figure}

\subsubsection{Structural representation}\sctlabel{methods_semantic_structure}

\paragraph{Whole TMC}
At the complex level, the TMC is the only structural feature, hence a single RDFS class is sufficient for adequate representation. For this purpose, the class \texttt{cmT:Tran\-si\-tion\-Me\-tal\-Com\-plex} has been defined and one of its instances is used to summarise the complex, interpreted as a single entity. In \figref{abox}, this instance is represented by the URI \texttt{cmT:KCEYPT}. In this case, the class assignment is implicit and performed via the domain specification of the predicates that apply to the node.

\paragraph{Ligands and metal centre}
At the ligand level, each URI ,representing either the metal centre or one of the ligands, is assigned a chemical identity which specifies its role in the structural composition of the TMC.

In the case of ligands, the predicate \texttt{lgL:isLigand}, a subproperty of \texttt{rdf:type}, assigns a subclass of \texttt{lgL:Ligand} (e.g., \texttt{lgLr:Ligand\_ligand1-0}, referring to the ligand that in \mbox{tmQMg-L} is identified as \say{ligand1-0}, a single Cl atom) to a URI. This subclass is itself an instance of \texttt{lgLr:LigandClass} and codifies the information on the chemical identity of the ligand.

An example of this scheme, as shown in \figref{abox}, is the triple
\begin{equation*}
    \texttt{lgL:KCEYPT\_ligand1-0\_0} \xrightarrow{\texttt{lgL:isLigand}} \texttt{lgLr:Ligand\_ligand1-0}.
\end{equation*}
Here, \texttt{lgL:KCEYPT\_ligand1-0\_0} is a URI that represents the first instance of ligand1-0 found within KCEYPT. Moreover, \texttt{lgLr:Ligand\_ligand1-0} is a subclass of \texttt{lgL:Ligand} and an instance of \texttt{lgLr:LigandClass}. This means that the following triples can also be found within the tmQM-RDF knowledge graph:
\begin{align*}
    \texttt{lgLr:Ligand\_ligand1-0} &\xrightarrow{\texttt{rdfs:subClassOf}} \texttt{lgL:Ligand}\\
    \texttt{lgLr:Ligand\_ligand1-0} &\xrightarrow{\texttt{rdf:type}} \texttt{lgLr:Ligand\_Class}.
\end{align*}
Moreover, since \texttt{lgL:isLigand} is a subproperty of \texttt{rdf:type}, the RDFS axioms allow us to infer\footnote{The entailment follows from the RDFS entailment rules \texttt{rdfs7} and \texttt{rdfs9} \citep{hayes2004rdfsemantics}.} the triple
\begin{equation*}
    \texttt{lgL:KCEYPT\_ligand1-0\_0} \xrightarrow{\texttt{rdf:type}} \texttt{lgL:Ligand}.
\end{equation*}

The metal centre is handled in an analogous way, using the predicate \texttt{lgC:isMetalCentre} and the classes \texttt{lgC:MetalCentre} and \texttt{lgCr:MetalCentreClass} respectively. An example of a subclass of \texttt{lgC:MetalCentre} (and instance of \texttt{lgCr:MetalCentreClass}) is \texttt{lgCr:MetalCentre\_Pt}.

\paragraph{Ligand-metal centre bonds}
A bond between a ligand and the metal centre is represented via one or more instances of the class \texttt{lgB:LigandBond}, to which the ligand and the centre are connected via the predicates \texttt{lgS:bLl} (which connects the ligand) and \texttt{lgS:bLc} (which connects the centre). The denticity of a ligand is represented by introducing one new \texttt{lgB:LigandBond} instance (and the corresponding predicates) for each binding site. Hapticity, on the other hand, is dealt with through the connection with the atomic level, described below.

In \figref{abox}, the bond between the metal centre (\texttt{lgC:KCEYPT\_metalCentre\_Pt}) and the ligand \texttt{lgL:KCEYPT\_ligand1-0\_0} introduced above is represented as
\begin{align*}
    \texttt{lgL:KCEYPT\_ligand1-0\_0} &\xrightarrow{\texttt{lgS:bLl}} \texttt{lgB:KCEYPT\_bondL\_ligand1-0\_0\_0}\\
    \texttt{lgC:KCEYPT\_metalCentre\_Pt} &\xrightarrow{\texttt{lgS:bLc}} \texttt{lgB:KCEYPT\_bondL\_ligand1-0\_0\_0}.
\end{align*}

\paragraph{Atoms}
At the atomic level, atoms are represented using the same mechanism used for the metal centre and the ligands. The predicate \texttt{tmA:isAtom} (again, a subproperty of \texttt{rdf:type}) assigns a chemical label (i.e., an element) to a URI. The label is expressed as an instance of the class \texttt{tmAr:Element} and a subclass of \texttt{tmA:Atom}, and it is meant to represent the chemical element of the atom (e.g., \texttt{tmAr:C} for carbonium).

Referencing \figref{abox} once more, the Cl atom in \texttt{lgL:KCEYPT\_ligand1-0\_0} is represented via
\begin{equation*}
    \texttt{tmA:KCEYPT\_Cl\_0} \xrightarrow{\texttt{tmA:isAtom}} \texttt{tmAr:Cl}.
\end{equation*}
The TBox statements about \texttt{tmAr:Cl} are
\begin{align*}
    \texttt{tmAr:Cl} &\xrightarrow{\texttt{rdfs:subClassOf}} \texttt{tmA:Atom}\\
    \texttt{tmAr:Cl} &\xrightarrow{\texttt{rdf:type}} \texttt{tmAr:Element},
\end{align*}
again with the possibility of inferring\footnote{Once again, this relies on the RDFS entailment rules \texttt{rdfs7} and \texttt{rdfs9} \citep{hayes2004rdfsemantics}. Recall that \texttt{tmA:isAtom} is also a subproperty of \texttt{rdf:type}.}
\begin{equation*}
    \texttt{tmA:KCEYPT\_Cl\_0} \xrightarrow{\texttt{rdf:type}} \texttt{tmA:Atom}.
\end{equation*}
\newline
\paragraph{Atomic bonds}
As with the ligand level bonds, chemical bonds between atoms are represented by instances of the class \texttt{tmB:AtomicBond}, to which the involved atoms are connected via the predicate \texttt{tmS:b}.

An example from \figref{abox}, which is the atomic level-equivalent of the ligand level bond showed above, is
\begin{align*}
    \texttt{tmA:KCEYPT\_Cl\_0} &\xrightarrow{\texttt{tmS:b}} \texttt{tmB:KCEYPT\_bond\_Pt\_0\_Cl\_0}\\
    \texttt{tmA:KCEYPT\_Pt\_0} &\xrightarrow{\texttt{tmS:b}} \texttt{tmB:KCEYPT\_bond\_Pt\_0\_Cl\_0}.
\end{align*}
Notice that, in this case, a single predicate is used for all the participants in the bond (both atoms), unlike in the ligand level, in which the two participants have different structural roles (i.e., one is a ligand, the other is the metal centre).

\paragraph{Connections between levels}
Each level communicates with the levels below and above through predicates that clarify which components are substructures of which larger structure (e.g., which atom belongs to which ligand).

At the complex level, the instance of the \texttt{cmT:TransitionMetalComplex} class is linked to the metal centre and the ligand objects at the ligand level through the predicates \texttt{cmT:hasMetalCentre} and \texttt{cmT:hasLigand}, respectively. As anticipated, the domain of both these predicates serves as the class declaration of the \texttt{cmT:TransitionMetalComplex} instance.

At the ligand level, each ligand object and the metal centre point to their composing atoms through the predicates \texttt{lgL:hasAtom} (for ligands) and \texttt{lgC:hasAtom} (for the metal centre). At this level, however, there are also objects representing the ligand-metal centre bonds, which also involve atoms. Therefore, it is necessary to introduce similar predicates for these objects as well. Since the transition metal atom always participates in these bonds, only the atoms that belong to a ligand are highlighted. This is done through the predicates \texttt{lgB:hasBindingAtom}, whose domain is the class \texttt{lgB:LigandBond}. This predicate is also used to specify the hapticity order of a ligand: if a ligand-metal centre bond involves several adjacent atoms, each of those atoms is the object of a triple whose subject is the same bond object and the predicate is \texttt{lgB:hasBindingAtom}.

In \figref{abox}, the links between the TMC object at the complex level, the ligand examined so far in the previous examples and its single composing atom read as follows:
\begin{align*}
    \texttt{cmT:KCEYPT} &\xrightarrow{\texttt{cmT:hasLigand}} \texttt{lgL:KCEYPT\_ligand1-0\_0}\\
    \texttt{lgL:KCEYPT\_ligand1-0\_0} &\xrightarrow{\texttt{lgL:hasAtom}} \texttt{tmA:KCEYPT\_Cl\_0}\\
    \texttt{lgB:KCEYPT\_bondL\_ligand1-0\_0\_0} &\xrightarrow{\texttt{lgB:hasBindingAtom}} \texttt{tmA:KCEYPT\_Cl\_0}.
\end{align*}

For an example of the representation of haptic ligands, we need to look at the other ligand-level ligand object represented in \figref{abox}, i.e., \texttt{lgL:KCEYPT\_ligand370-0\_0}. This is an haptic ligand as it binds to the metal centre via two adjacent C atoms and, therefore, its bond object, \texttt{lgB:KCEYPT\_bondL\_ligand370-0\_0\_0}, possesses two binding atoms:
\begin{align*}
    \texttt{lgB:KCEYPT\_bondL\_ligand370-0\_0\_0} &\xrightarrow{\texttt{lgB:hasBindingAtom}} \texttt{tmA:KCEYPT\_C\_0}.\\
    \texttt{lgB:KCEYPT\_bondL\_ligand370-0\_0\_0} &\xrightarrow{\texttt{lgB:hasBindingAtom}} \texttt{tmA:KCEYPT\_C\_1}.
\end{align*}

\subsubsection{Properties}\sctlabel{methods_semantic_structure_properties}
\paragraph{General property description scheme}
Within the scope of this work, any entity (in technical terms, any \texttt{rdfs:Resource} instance) can be endowed with one or more properties. In order to express this, the predicate \texttt{inp5:hasProperty} is used. As shown in \figref{tbox}, the range of \texttt{inp5:hasProperty} is the class \texttt{inp5:ObservedProperty}, which in practice is instantiated as a blank node. The reason behind this choice is that the instances of the class \texttt{inp5:ObservedProperty} are used merely to summarise the property, but they do not express any information themselves. The actual data about the property is stated by using three additional predicates that act on the blank node: \texttt{rdf:type} assigns an instance of \texttt{inp5:PropertyType} (i.e., a property name, for identifiability), \texttt{inp5:reference} points to an instance of \texttt{ds:Datataset}, representing the original dataset that provided the data and, finally, \texttt{rdf:value} states the actual value.

For example, the TMC KCEYPT has a HOMO-LUMO gap (\texttt{cmTp:homo\_lumo\_gap}) that, according to tmQM, amounts to $0.14592\,E_\mathrm{h}$. This is stated in tmQM-RDF as
\begin{align*}
    \texttt{cmT:KCEYPT} &\xrightarrow{\texttt{inp5:hasProperty}} \texttt{\_:ppp}\\
    \texttt{\_:ppp} &\xrightarrow{\texttt{rdf:type}} \texttt{cmTp:homo\_lumo\_gap}\\
    \texttt{\_:ppp} &\xrightarrow{\texttt{inp5:reference}} \texttt{dsC:tmQM}\\
    \texttt{\_:ppp} &\xrightarrow{\texttt{rdf:value}} \text{\say{\texttt{0.14592}}}.\\
\end{align*}

A list of all the properties used in this work can be found in \apxref{comprehensive}.\\

Each property may require some additional specification or some other consideration, all of which are explained in this section, but all of them employ and extend this baseline structure.
The most variable aspect of a property description is its value. Simple properties, such as the \mbox{HOMO-LUMO} gap, can be expressed using a literal. Other properties, like the NBO type of an atomic bond (\texttt{tmBp:nbo\_type}), use a single URI instead. There are, however, other properties that cannot be described using a single term, and require more complicated structures.

\paragraph{Non-elementary properties}
When a single URI or literal does not suffice to specify the value of a property, a blank node, instance of (subclasses of) \texttt{inp5:NonSimplePropertyValue} is used as the object of \texttt{rdf:value}. Several other predicates, applied to the blank node, are then used to fully describe the property value.

Multiple-term values generally follow the scheme of the \texttt{rdfs:Container} classes \citep{brickley2014rdfs}, hence they serve as \say{collectors} of the information that composes the property value. In tmQM-RDF, such \emph{non-elementary properties} can be tuples of 3D Cartesian coordinates, lists of counts (e.g., element counts or structure counts), lists of binding atoms or pointers to the most stable occurrence of a ligand.\\

In the case of count lists, which we show here as a representative example, the blank node is an instance of \texttt{inp5:CountDescription}. Via the predicates \texttt{rdf:\_nnn}, where \texttt{nnn} represents a positive integer, the entries of the list are specified. Each entry is itself a blank node and an instance of \texttt{inp5:CountDescriptionItem}, which specifies the object being counted (an instance of \texttt{inp5:Countable}), via the \texttt{inp5:countOf} predicate, and the actual count, via \texttt{rdf:value}. A compliant set of triples is:
{
\allowdisplaybreaks
\begin{align*}
    \texttt{cmT:KCEYPT} &\xrightarrow{\texttt{cmTp:meta\_data}} \texttt{\_:mmm}\\
    \texttt{\_:mmm} &\xrightarrow{\texttt{inp5:hasProperty}} \texttt{\_:ppp}\\
    \texttt{\_:ppp} &\xrightarrow{\texttt{rdf:type}} \texttt{cmTp:element\_counts}\\
    \texttt{\_:ppp} &\xrightarrow{\texttt{inp5:reference}} \texttt{dsG:tmQMg}\\
    \texttt{\_:ppp} &\xrightarrow{\texttt{rdf:value}} \texttt{\_:vvv}\\
    \texttt{\_:vvv} &\xrightarrow{\texttt{rdf:type}} \texttt{inp5:CountDescription}\\
    \texttt{\_:vvv} &\xrightarrow{\texttt{rdf:\_1}} \texttt{\_:eee}\\
    \texttt{\_:eee} &\xrightarrow{\texttt{rdf:type}} \texttt{inp5:CountDescriptionItem}\\
    \texttt{\_:eee} &\xrightarrow{\texttt{inp5:countOf}} \texttt{tmAr:Cl}\\
    \texttt{\_:eee} &\xrightarrow{\texttt{rdf:value}} \text{\say{\texttt{1}}}.
\end{align*}
}
Lists of binding atoms follow a closely related scheme, with \texttt{inp5:CountDescription} and \texttt{inp5:Cou\-nt\-Des\-crip\-tion\-I\-tem} replaced by \texttt{lgLr:BindingAtomsSMILES} and \texttt{lgLr:BindingSiteSMILES}. The only difference is that \texttt{lgLr:BindingSiteSMILES} is itself a subclass of \texttt{rdfs:Container}.

For the other available non-elementary properties, i.e., Cartesian coordinates or pointers to the most stable ligand occurrence, the scheme can be simplified. The \texttt{rdf:\_nnn} predicates are replaced by \texttt{tmAp:x}, \texttt{tmAp:y} and \texttt{tmAp:z} in the case of tuples of Cartesian coordinates (\texttt{\_:vvv} is of type \texttt{inp5:CartesianCoordinates3D}), and by \texttt{lgLr:subgraph} and \texttt{lgLr:instance} for descriptions of stable ligand occurrences (\texttt{\_:vvv} is of type \texttt{lgLr:StableOccurrenceDescription}), and their objects need not to be blank nodes, but can be literals or URIs.

\paragraph{Abstract ligand properties}\sctlabel{tbox_ligand_pt}
There are numerous quantitative and qualitative features that pertain to the entities described above. While complex-level and atomic features are attached directly to the entities of interest, contextually to the description of each TMC, properties that refer to entire ligands are reported separately. Since tmQMg-L describes each ligand using general representations that are not specific to any particular instance of the ligands, the properties thereby reported can be thought of as properties of the entire ligand class. For this reason, ligand-level properties are not specified as properties of the physical ligands within TMCs, but rather are listed once, together with the definition of the ligand classes \texttt{lgLr:Ligand\_*}.

\paragraph{TMC metadata}
Additional TMC metadata, as defined in tmQMg, is included in each TMC subgraph by introducing a blank node, instance of the class \texttt{cmT:MetaDataDescription}. As with the class \texttt{inp5:ObservedProperty}, this object serves as a summary of the collection of metadata. Actual properties are specified via the usual property description scheme as properties of the metadata blank node (as shown above), which in turn is connected to the TMC node through the \texttt{cmTp:meta\_data} predicate. Expanding on the aforementioned example, we can find:
\begin{align*}
    \texttt{cmT:KCEYPT} &\xrightarrow{\texttt{cmTp:meta\_data}} \texttt{\_:mmm}\\
    \texttt{\_:mmm} &\xrightarrow{\texttt{rdf:type}} \texttt{cmTp:MetaDataDescription}.
\end{align*}

\paragraph{Information for computational reproducibility}
For some complex-level and whole-ligand features for which this information is readily available, the quantum chemistry methods used in tmQM and tmQMg-L are specified together with their \texttt{inp5:ObservedProperty} objects. The predicates \texttt{nm:optimisation} and \texttt{nm:singlepoint}, which point to instances of \texttt{nm:Optimisation} and \texttt{nm:Singlepoint} respectively, indicate the methods used to optimise the molecular geometries and compute their electronic structure properties.

As shown in \figref{abox}, this information is available for the HOMO-LUMO gap of KCEYPT, hence it is also reported in tmQM-RDF:
\begin{align*}
    \texttt{cmT:KCEYPT} &\xrightarrow{\texttt{inp5:hasProperty}} \texttt{\_:ppp}\\
    \texttt{\_:ppp} &\xrightarrow{\texttt{rdf:type}} \texttt{cmTp:homo\_lumo\_gap}\\
    \texttt{\_:ppp} &\xrightarrow{\texttt{nm:optimisation}} \texttt{nm:GFN2\_xTB}\\
    \texttt{\_:ppp} &\xrightarrow{\texttt{nm:singlepoint}} \texttt{nm:TPSSh\_D3BJ\_\_def2\_SVP}.\\
\end{align*}
This means that the HOMO-LUMO gap, for KCEYPT, has been computed using a geometry optimised at the GFN2-xTB level and then by performing single-point lowest-energy state calculations at the TPSSh-D3BJ/def2-SVP level.
 
\section{The transition metal quantum mechanics RDF dataset}\sctlabel{tmQM-RDF}
A collection of $47,814$ TMCs, represented using the RDF vocabulary developed in \sctref{methods}, has been assembled into the \emph{transition metal quantum mechanics RDF (tmQM-RDF) dataset}. As anticipated, tmQM-RDF relies on the information extracted from the tmQM dataset series, which has been organised in the three hierarchical levels described above in \sctref{methods}. As we want to ensure that tmQM-RDF only contains TMCs for which information can be found in all the datasets of the series, we apply filters to exclude structures that do not meet this requirement.\\

The fusion of the three parent datasets into tmQM-RDF consists of three main steps, which are further divided into smaller subtasks:
\begin{enumerate}
    \item tmQMg and tmQMg-L are linked together.
    \begin{enumerate}
        \item The geometries of the ligand instances in tmQMg-L are matched against the geometries reported in tmQMg, in order to associate each instance with the correct group of atoms.
        \item Using this information, the ligand-level composition of each TMC in tmQMg is computed.
        \item A sanity check is performed, verifying that for each TMC, the reconstructed ligands extracted from tmQMg-L cover the entire molecular graph. TMCs that fail this check are removed from the dataset.\footnote{Out of the $74,636$ TMCs originally extracted from tmQMg, this step removes $26,811$ of them (approximately $36\%$), leaving $47,825$ viable TMCs.}
    \end{enumerate}
    \item tmQM and tmQMg are linked together.
    \begin{enumerate}
        \item The content of tmQM is inspected and indexed according to the viable TMCs identified in the previous step. The TMCs for which it is not possible to retrieve complete information are removed from the dataset.\footnote{This step removes $11$ more TMCs, leading to the final dataset size of $47,814$ TMCs.}
        \item The information in tmQM is reorganised in a new structure that mirrors that of tmQMg.
    \end{enumerate}
    \item tmQM, tmQMg and tmQMg-L are merged into tmQM-RDF.
    \begin{enumerate}
        \item All the processed information is used to build the RDF representation of the remaining TMCs, using the vocabulary defined in \sctref{methods_semantic}.
    \end{enumerate}
\end{enumerate}

\subsection{Dataset summary}
In total, the tmQM-RDF knowledge graph contains approximately $534$ million triples. Of these, about $12$ million are used to define the terminology component and the instances of the classes \texttt{lgCr:MetalCentreClass}, \texttt{lgLr:LigandClass} and \texttt{tmAr:Element} (together with their associated chemical properties). The remaining triples describe the assertion component, i.e., the TMCs. Each TMC, on average, requires approximately $11,000$ triples (not counting the required assertion component for ligand and atomic species) to be completely described.\\

As for the chemical expressivity of tmQM-RDf, a concise summary is provided in \figref{tmQM_RDF_summary}. 

\figref{tmQM_RDF_summary_centres} shows a bar plot of the size of the subpopulations of tmQM-RDF identified by the metal centre, and thus it offers an initial account of the diversity of the complexes represented in the dataset. In particular, it shows the wide selection of transition metals that can be encountered, while also emphasising that this representation is not uniform.

As per the abundance and the variety of the ligands that bind to the metal centres, there are currently $27,905$ different ligands in tmQM-RDF, and the average number of ligands per TMC is $3.24$ ($\mathrm{SD} = 1.16$). It is also interesting to investigate the frequency of appearance of each ligand, as adequate representativeness is essential for downstream tasks that depend on data to extract realistic correlations between different properties, substituents, and behaviours (e.g., generative tasks). With this regard, \figref{tmQM_RDF_summary_ligands} shows the $95\%$ quantile of the empirical distributions of two metrics representing two levels of ligand (absolute) frequency: the total count of occurrences of each ligand ($q_{0.95}^\mathrm{inst}$) and the number of TMCs that posses at least one copy of each ligand ($q_{0.95}^\mathrm{copy}$). These measurements are taken both for the entire dataset and for each metal centre-specific subpopulation. Both $q_{0.95}^\mathrm{inst}$ and $q_{0.95}^\mathrm{copy}$ have been designed to capture information availability with regards to ligands. Low values ($q_{0.95}^\mathrm{copy} = 5$ and $q_{0.95}^\mathrm{inst} = 6$ for the entire dataset) indicate that most ligands appear only a handful of times. Specifically, $q_{0.95}^\mathrm{inst} = 6$ means that at least $95\%$ of the ligands appear no more than $6$ times in total, whereas $q_{0.95}^\mathrm{copy} = 5$ means that at least $95\%$ of ligands appear, each, in at most $5$ different TMCs. When compared to the size of tmQM-RDF, these numbers point to a strongly sparse distribution of ligands.

We ultimately consider the $10$ most frequent ligands, in the $q_{0.95}^\mathrm{inst}$ sense, as a simple means of visually representing the most typical structures found within the dataset. We show a plot of the counts of these ligands, in the same two modalities introduced above, and their graphical representations in \figrefmult{tmQM_RDF_summary_freq_ligands_counts}{tmQM_RDF_summary_freq_ligands_structures} respectively.

\begin{figure}
    \begin{center}
    \makebox[\textwidth][c]{%
    \begin{tabular}{c c}
    \begin{subfigure}[t]{0.58\linewidth}
        \includegraphics[width=\linewidth]{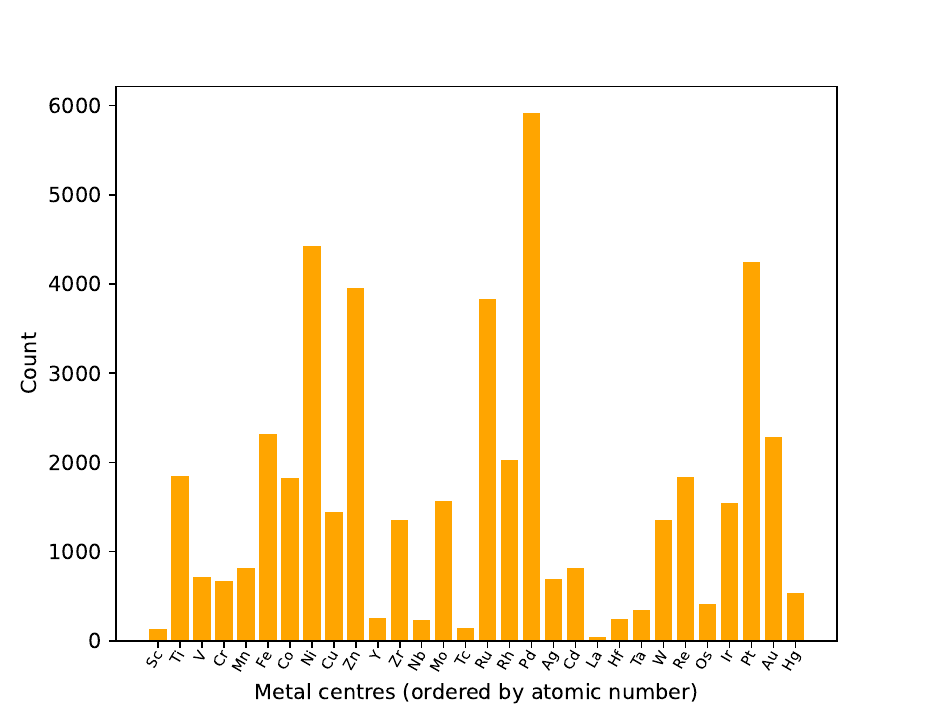}
        \caption{}\figlabel{tmQM_RDF_summary_centres}
    \end{subfigure} &
    \begin{subfigure}[t]{0.58\linewidth}
        \includegraphics[width=\linewidth]{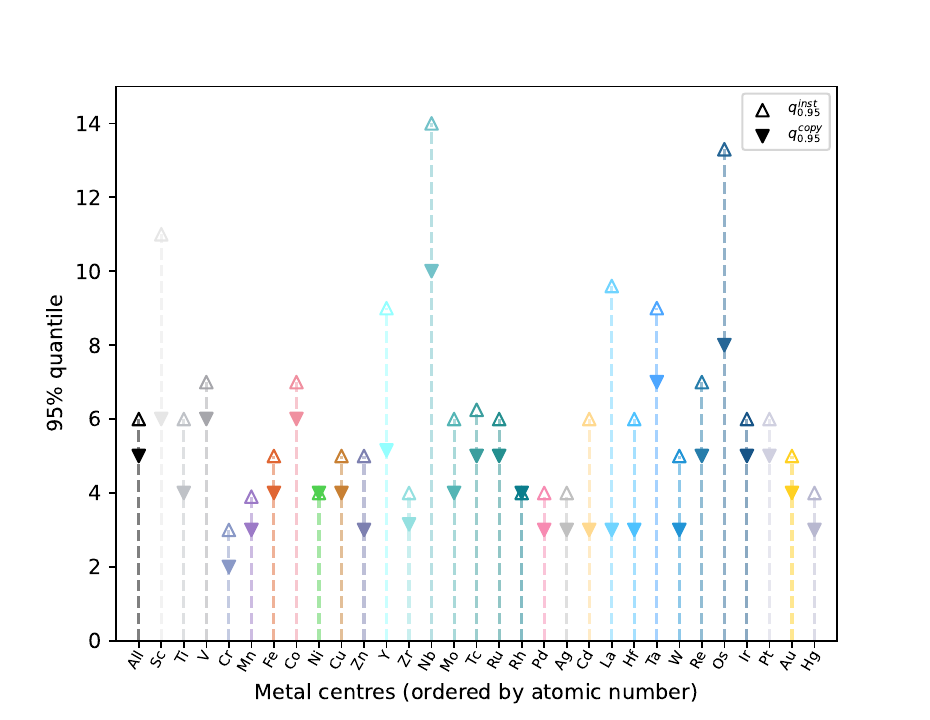}
        \caption{}\figlabel{tmQM_RDF_summary_ligands}
    \end{subfigure} \\
    \begin{subfigure}[t]{0.54\linewidth}
        \hspace*{-0.5cm}
        \includegraphics[width=\linewidth]{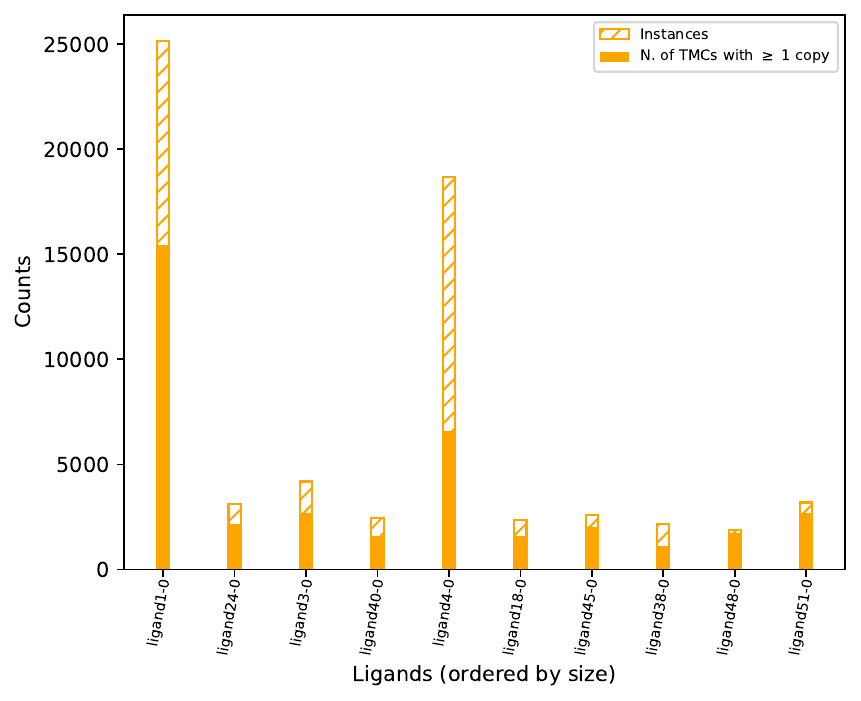}
        \caption{}\figlabel{tmQM_RDF_summary_freq_ligands_counts}
    \end{subfigure} &
    \begin{subfigure}[t]{0.54\linewidth}
        \raisebox{1.75cm}{%
        \includegraphics[width=\linewidth,margin=0pt 0cm 0pt -2cm]{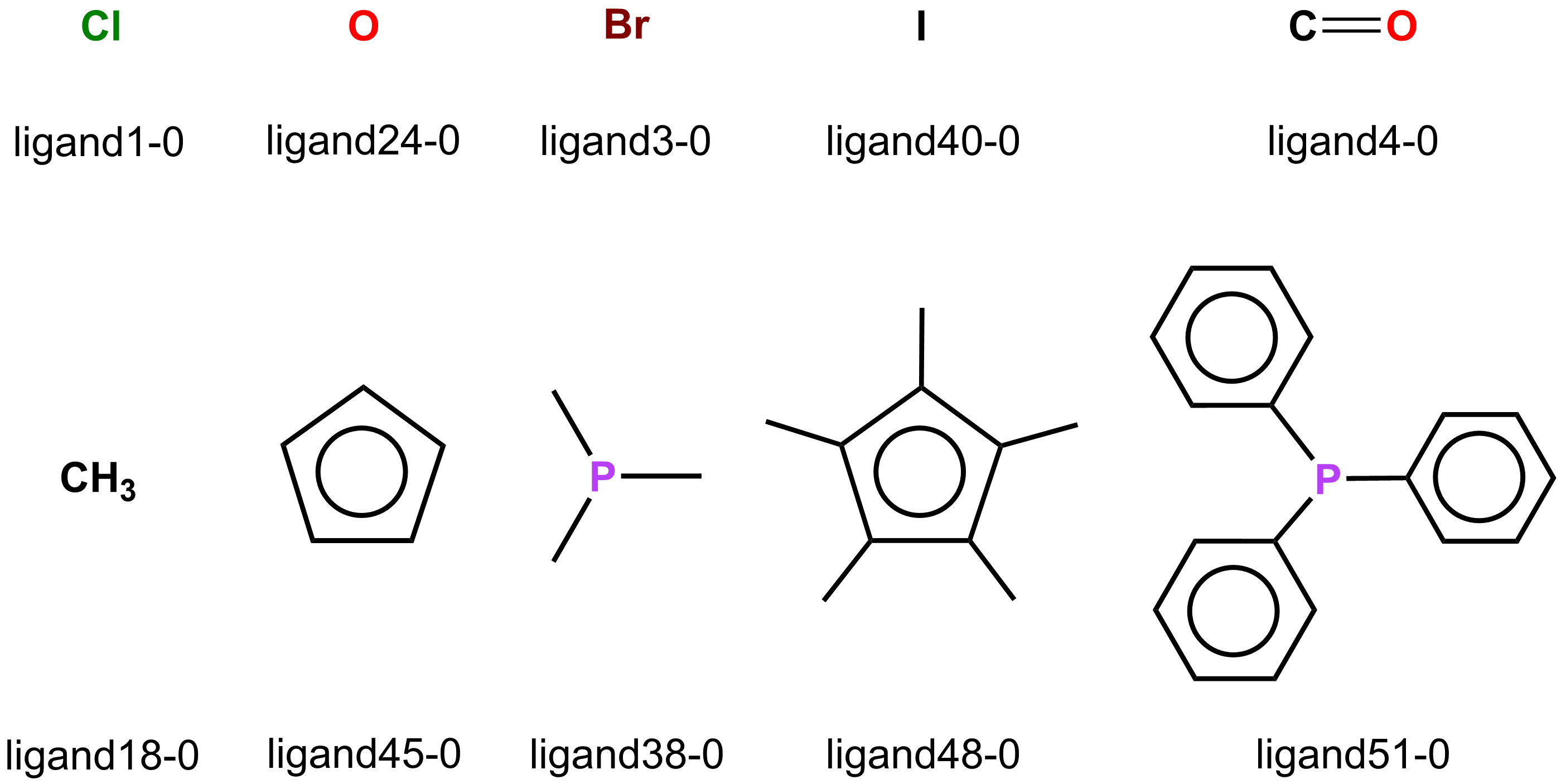}
        }
        \caption{}\figlabel{tmQM_RDF_summary_freq_ligands_structures}
    \end{subfigure}
    \end{tabular}
    }
    \end{center}
    \caption{(a) The bar plot of the counts of the appearances of the metal centres found in tmQM-RDF. (b) The $95\%$ quantiles of the empirical distributions of the count of single ligand instances ($q_{0.95}^\mathrm{inst}$) and of the count of TMCs that contain at least one copy of each ligand ($q_{0.95}^\mathrm{copy}$), in each subpopulation of TMCs identified by the metal centre. Vertical lines are added for readability. (c) The bar plot of the counts of the appearances of the 10 most frequent ligands in tmQM-RDF. (d) A visual representation of the 10 most frequent ligands in tmQM-RDF.}
    \figlabel{tmQM_RDF_summary}
\end{figure}

\subsection{Selected subpopulations}\sctlabel{dataset_selection}
With the aim of facilitating computational experiments and preliminary investigations, we select two smaller datasets sampled from the general population described within tmQM-RDF, whose distribution of chemical composition and structure does not deviate excessively from that of the parent dataset. Both datasets consist of $1600$ TMCs, split into training ($1000$ TMCs), validation ($300$ TMCs) and test ($300$ TMCs) partitions. In order to cover a wide application scope, the two subpopulations consider transition metal centres belonging to groups from separated regions of the periodic table. One dataset considers an early region (Cr, Mo, W, \emph{earlyTM} dataset selection), whereas the other considers a late region (Ni, Pd, Pt, \emph{lateTM} dataset selection). Summaries of the two datasets and additional information about the sampling procedure can be found in \apxref{selection}.\\

These subpopulations are presented in the form of lists of CSD codes. Extraction from the global KG, if needed, is greatly facilitated by the organisation of tmQM-RDF: the subgraph corresponding to each TMC is indexed by its CSD code and can be easily loaded into a triple store independently of the rest of the dataset.

\section{Experiments: plausible TMC reconstruction}\sctlabel{experiments}
In this section we propose an illustrative application of the tmQM-RDF dataset, by considering a simple TMC reconstruction and manipulation task, which we call \emph{plausible TMC completion}. Given an incomplete molecular scaffold (i.e., a TMC with a missing ligand) and a ranked set of candidate completions into full TMCs, we shall refer to the top$-k$ highest scoring reconstructions as \emph{plausible}, for an appropriate integer value of $k$.

Conceptually, in this experiment we consider a population $\mathcal{G}$ of TMCs (in practice, either one of the two selections introduced in \sctref{dataset_selection}) sampled from the set of TMCs described within tmQM-RDF. The dataset $\mathcal{G}$ is then split into two parts:\footnote{In \sctref{dataset_selection} we mention that our dataset selections are divided in three subsets, namely a training set, a validation set and a test set. Since the experiments hereby proposed do not require a validation set, we only employ the training and test partitions. We consider then, for each selection, a population of only $1300$ TMCs. On the other hand, more sophisticated frameworks, like conformal prediction \citep{balasubramanian2014conformal}, could naturally make use of the validation set as a calibration set.} a training set $\mathcal{G}_\mathrm{train}$, used to acquire relevant structural and chemical information that can be used to assign a score to a candidate TMC, and a test set $\mathcal{G}_\mathrm{test}$, which constitutes the population of target TMCs that have to be reconstructed starting from molecular scaffolds which are to be completed with the addition of a single ligand.

The former is used to implement the training phase of our experiment, the latter for the test phase. Additional information on our approach, including precise definitions and formulations, can be found in \apxref{experiments_methods}.\\

\subsection{Training phase}
The ultimate goal of the training phase is to construct a real-valued score function \mbox{$f:\mathcal{G} \to \R$} that can estimate the degree to which it is plausible that a given TMC was sampled from the population of interest. To achieve this, we proceed as follows.
\begin{enumerate}
    \item From $\mathcal{G}_\mathrm{train}$, a set of frequent (basic graph) patterns, in the SPARQL sense \citep{harris2013sparql, ali2022sparqlsurvey, ell2025graph}, is extracted. These are meant to capture recurrent substructures within the TMC population, with regards to the modalities through which ligands attach to the metal centre atom, i.e., ligand identities, denticity/hapticity order and binding atoms.
    \item Using a similarity-based agglomerative clustering algorithm, the set of mined patterns is partitioned. In this way, a set of families of typical substructures can be determined. The similarity between two patterns is evaluated as the cosine similarity between binary vectors, each representing one of the two patterns. The entries of these vectors corresponds to the TMCs in $\mathcal{G}_\mathrm{train}$ and evaluate to $1$ if and only if the associated TMC is matched by the pattern, and $0$ otherwise.
    \item\label{enum:feature_vector_def} A (binary) feature vector is now created for each TMC in $\mathcal{G}_\mathrm{train}$. Each feature corresponds to one of the families identified in the previous point and will evaluate to $1$ if and only if that family is expressed within the TMC (i.e., there exists a pattern in the family that matches the TMC). The feature will evaluate to $0$ otherwise.
    \item The joint distribution of the features is estimated.
\end{enumerate}

This last step is the most crucial of the training phase. If $G \in \mathcal{G}$ is a TMC, $\bm{Y}_G$ is the corresponding binary feature vector introduced in step \ref{enum:feature_vector_def} and $\mathbb{P}(\bm{Y}_G \mid \mathcal{G}_\mathrm{train})$ is its joint distribution, then the score of $G$ can be defined as
\begin{equation*}
    f(G) := \log{\mathbb{P}(\bm{Y}_G\mid \mathcal{G}_{\mathrm{train}})}.
\end{equation*}
If the joint probability distribution of the features is available, and the features correctly represent the fundamental structural modalities underneath the dataset, then highly representative TMCs will generate feature configurations that possess a high probability. Obviously, this distribution is not known and the high dimensionality of the feature space makes estimation generally complicated \citep{johnstone2009highdimensions}. For this reason, we model the distribution using the Bayesian Network (BN) framework \citep{koller2009probabilistic, kitson2023bnsurvey} and estimate it based on the training data. BNs are specifically designed to deal with correlated, high-dimensional features, hence they constitute a suitable solution to the challenges at hand.

\subsection{Test phase}
The test phase is meant to evaluate the goodness of the estimated score function $f$ as a tool for TMC manipulation. In particular, under the assumption that the TMCs in $\mathcal{G}$, and therefore in $\mathcal{G}_\mathrm{test}$, are themselves representative, we expect that a suitable $f$ should be able to adequately reconstruct the elements of $\mathcal{G}$ when new transition metal complexes are created by attaching one ligand to an incomplete molecular scaffold. For ease of exposition, we consider here a single test TMC $G \in \mathcal{G}_\mathrm{test}$, but the procedure below will be repeated identically for each TMC in the test set.

\begin{enumerate}
    \item The set $L_\mathrm{train}$ of all the ligands that appear in $\mathcal{G}_\mathrm{train}$ is computed.
    \item A molecular scaffold is created, by removing from $G$, at random, one of its ligand. The ligand to remove is chosen among the ligands of $G$ that also appear in $L_\mathrm{train}$.
    \item A family of candidate reconstructions is obtained through the completion of the molecular scaffold using, one by one, the ligands from different subsets of $L_\mathrm{train}$. These subsets are determined by applying different constraints on the candidate ligands. In this work we consider the following filters:
    \begin{itemize}
        \item no constraints;
        \item same hapticity and denticity order as the removed ligand;
        \item same charge as the removed ligand.
    \end{itemize}
    Notice that, by removing from $G$ a ligand from $L_\mathrm{train}$, we ensure that $G$, the original TMC, also appears among the considered reconstructions.
    \item Each reconstruction is scored via $f$ and the set of reconstructions is ranked accordingly.
\end{enumerate}

Being this experiment focused on TMC reconstruction, a natural way to interpret the resulting ranked family of reconstructions is to look for the original TMC among the plausible reconstructions. An equally natural metric to consider, then, is the top$-k$ accuracy \citep{lapin2016topk, nia2025personalized}, which can be interpreted as the fraction of test TMCs that are deemed to be plausible (recall the definition of plausible TMC completion given at the beginning of the section). We focus here on the thresholds $k \in \{1,\,5,\,10\}$. This evaluation procedure is analogous to the task of most likely sentence completion in LLM research \citep{park2020assessment, nozza2022completion}.

\subsection{Results and discussion}\sctlabel{experiments_results}
In \tblref{short_results} we summarise the results of the experiment. For each of the two dataset selections introduced in \sctref{dataset_selection}, we report the number of entries in the structural feature vector $\bm{Y}_G$, together with the total number of frequent patterns that have been mined and later clustered into the resulting families of substructures. In addition, we report the measured values of top$-k$ accuracy, under each constraint, for $k \in \{1,\, 10\}$.\\

Based on the results in \tblref{short_results}, we focus the discussion below on two key takeaways.

First, despite achieving substantially different top$-1$ accuracy, our approach is capable of achieving more that $80\%$ accuracy already for $k = 10$ for both datasets, across all three constraints. Considering that there are hundreds of possible reconstructions in each scenario for each incomplete TMC, as reported in \tblref{short_results}, this indicates that our approach is able to recognise the test TMCs as plausible members of the TMC population from which $\mathcal{G}_{\mathrm{train}}$ was sampled. Most importantly, this successfully validates the usefulness of the integrated representation offered by the tmQM-RDF knowledge graph.

Second, in the case of the \emph{lateTM} dataset, despite having access to a $\sim 2$ times larger number of patterns with respect to \emph{earlyTM} ($3710$ versus $1994$), our method is not capable of performing as well. There are multiple factors that could hypothetically contribute to this phenomenon, each deserving of further investigation which may inspire future developments. One option is that there truly is a small set of structures which alone accounts for most of the structural variability in tmQM-RDF. The high number of patterns found when analysing \emph{lateTM} could then be an indication of overfitting. Another possibility is that the mined patterns are lacking in quality, or fail to capture a sufficiently diverse set of structures. This last interpretation could also explain why such a large number of frequent patterns gave rise to a number of families which is comparable to that found in \emph{earlyTM}.

We provide additional considerations and a more extensive set of results in \apxref{extensive_results}.

\begin{table}[!t]
\footnotesize
    \centering
    \makebox[\textwidth][c]{%
    
    \begin{tabular}{c c c c c c c c c c c}
         \toprule
\multirow{2}{*}{\makecell{Dataset\\selection}} & No. of features  & \multicolumn{9}{c}{Filter}\\
\cmidrule(lr){2-2} \cmidrule(lr){3-11}
 & \emph{Clustered (Mined)} & \multicolumn{3}{c}{\emph{No filter}} & \multicolumn{3}{c}{\emph{Hapticity/denticity order}} & \multicolumn{3}{c}{\emph{Charge}}\\
\cmidrule(lr){3-5}\cmidrule(lr){6-8}\cmidrule(lr){9-11}
 &   & \multirow{2}{*}{\makecell{No. of\\candidates}} & \multicolumn{2}{c}{Top$-k$ accuracy} & \multirow{2}{*}{\makecell{No. of\\candidates}} & \multicolumn{2}{c}{Top$-k$ accuracy} & \multirow{2}{*}{\makecell{No. of\\candidates}} & \multicolumn{2}{c}{Top$-k$ accuracy}\\
\cmidrule(lr){4-5}\cmidrule(lr){7-8}\cmidrule(lr){10-11}
 &   &   & $\emph{k = 1}$ & $\emph{k = 10}$ &   & $\emph{k = 1}$ & $\emph{k = 10}$ &   & $\emph{k = 1}$ & $\emph{k = 10}$\\
\midrule
\midrule
\emph{earlyTM} & $455\ (1994)$  & $828$ & 0.619 & 0.866 &  $337.0$ ($\mathrm{SD} = 121.5$) & 0.702 & 0.967 & $321.2$ ($\mathrm{SD} = 116.6$) & 0.649 & 0.896\\
\emph{lateTM} & $518\ (3710)$  & $325$ & ${0.357}$ & ${0.833}$ & $133.1$ ($\mathrm{SD} = 41.1$) & ${0.403}$ & ${0.903}$ & $141.75$ ($\mathrm{SD} = 26.8$) & ${0.377}$ & ${0.920}$\\
\bottomrule
    \end{tabular}
    
    }
    \caption{Number of features used, average number of candidate reconstructions per scaffold and top$-k$ accuracy measured during the TMC completion task for the two dataset selections. The quantity in parentheses indicate the number of mined frequent patterns which have been clustered into the final families of typical substructures.}
    \tbllabel{short_results}
\end{table}

\section{Conclusions}\sctlabel{conclusions}

In this paper we introduce the tmQM-RDF knowledge graph. This resource collects detailed information on a large population of TMCs, by carefully extracting, aggregating and processing the data that have been collected in the tmQM dataset series. At the core of this effort lies a dedicated vocabulary, based on the RDF syntax, specifically designed to make chemical and structural information explicit, interpretable and accessible.

It is indeed in the accessibility of the data that our proposal produces its most valuable contribution, as it allows for the extraction of heterogeneous data, ranging from the querying of quantitative properties to that of the structural composition of a TMC either in terms of atoms or ligands, from a unified source using the SPARQL query language. Even in the case in which it is not possible or convenient to retrieve all the desired information from tmQM-RDF, our knowledge graph still includes all the relevant keys used by the original datasets to index the data, hence more complex investigations can be easily carried out starting from the response provided by tmQM-RDF.

We showed in this manuscript how a simple use case aimed at TMC manipulation can benefit from the formulation used by tmQM-RDF. In particular, we showed that it is possible to extract extensive amounts of relevant structural information that spans both the ligand and the atomic level of a TMC by relying on the SPARQL graph pattern formalism. A relatively simple probabilistic model, such as a Bayesian Network, was capable of harnessing the signal embedded in those patterns, in order to successfully reconstruct a new population of TMCs. This suggests a promising future development of data-driven computational TMC manipulation and analysis.

\section{Data and code availability}\sctlabel{data}
The tmQM-RDF knowledge graph and the code used in the experiments can be found at \url{https://github.com/luca-cibinel/tmQM-RDF}. A summary of all the parameters involved in the creation of tmQM-RDF and in the experiments is available in \apxref{parameters}.

\section{Acknowledgements}
This work was supported by the Research Council of Norway through its Centre of Excellence Integreat - The Norwegian Centre for knowledge-driven machine learning, project number 332645. T.\,L. and D.\,B. acknowledge the RCN FRIPRO programs for ground-breaking research (catLEGOS project; number 325003) and National Research Centers of Excellence (Hylleraas; project number 262695). R.\,D.\,B. acknowledges Plumbin’ (RCN, n. 323985). The authors would like to express their gratitude to Hannes Kneiding for his help in the management of the datasets of the tmQM series.

\section{Conflicts of interest}
There are no conflicts of interest to declare.

\bibliographystyle{chicago}
\bibliography{parsed_for_paper_1_main_proj5_bib}

@article{abu2015tmcreview,
 author = {Abu-Dief, Ahmed M and Mohamed, Ibrahim MA},
 journal = {Beni-Suef University Journal of Basic and Applied Sciences},
 note = {\url{https://doi.org/10.1016/j.bjbas.2015.05.004}},
 pages = {119--133},
 title = {A review on versatile applications of transition metal complexes incorporating {S}chiff bases},
 volume = {4},
 year = {2015}
}

@article{ali2022sparqlsurvey,
 author = {Ali, Waqas and Saleem, Muhammad and Yao, Bin and Hogan, Aidan and Ngomo, Axel-Cyrille Ngonga},
 journal = {The VLDB Journal},
 note = {\url{https://doi.org/10.1007/s00778-021-00711-3}},
 pages = {1--26},
 title = {A survey of {RDF} stores \& {SPARQL} engines for querying knowledge graphs},
 volume = {31},
 year = {2022}
}

@article{ankan2024pgmpy,
 author = {Ankan, Ankur and Johannes Textor},
 journal = {Journal of Machine Learning Research},
 note = {\url{http://jmlr.org/papers/v25/23-0487.html}},
 pages = {1--8},
 title = {pgmpy: A {P}ython Toolkit for {B}ayesian Networks},
 volume = {25},
 year = {2024}
}

@misc{ayers2008uri,
 author = {Ayers, Danny and Völkel, Max},
 howpublished = {\url{https://www.w3.org/TR/cooluris/#oldweb}},
 note = {Accessed: 2025-12-03},
 title = {Cool {URI}s for the Semantic Web},
 year = {2008}
}

@book{balasubramanian2014conformal,
 author = {Balasubramanian, Vineeth and Ho, Shen Shyang and Vovk, Vladimir},
 note = {ISBN: 978-0-12-398537-8},
 publisher = {Newnes},
 title = {Conformal Prediction for Reliable Machine Learning: Theory, Adaptations and Applications},
 year = {2014}
}

@article{balcells2020tmqm,
 author = {Balcells, David
and Skjelstad, Bastian Bjerkem},
 journal = {Journal of Chemical Information and Modeling},
 note = {\url{https://doi.org/10.1021/acs.jcim.0c01041}},
 pages = {6135--6146},
 title = {{t}m{QM} Dataset-Quantum Geometries and Properties of 86k Transition Metal Complexes},
 volume = {60},
 year = {2020}
}

@misc{beckett2014turtle,
 author = {Beckett, David and Berners-Lee, Tim and Prud'hommeaux, Eric and Carothers, Gavin},
 howpublished = {\url{https://www.w3.org/TR/turtle/#grammar-production-prefixID}},
 note = {Accessed: 2025-12-03},
 title = {{RDF} 1.1 Turtle},
 year = {2014}
}

@article{belleau2008biotordf,
 author = {Belleau, Fran{\c{c}}ois and Nolin, Marc-Alexandre and Tourigny, Nicole and Rigault, Philippe and Morissette, Jean},
 journal = {Journal of Biomedical Informatics},
 note = {\url{https://doi.org/10.1016/j.jbi.2008.03.004}},
 pages = {706--716},
 title = {{Bio2RDF}: towards a mashup to build bioinformatics knowledge systems},
 volume = {41},
 year = {2008}
}

@book{bergmann2014logic,
 author = {Bergmann, Merrie and Moor, James and Nelson, Jack},
 note = {ISBN: 978-0-07-803841-9},
 publisher = {McGraw-Hill/Connect Learn Succeed},
 title = {The Logic Book},
 year = {2014}
}

@manual{boettiger2021virtuoso,
 author = {Carl Boettiger},
 howpublished = {https://CRAN.R-project.org/package=virtuoso},
 note = {R package version 0.1.8},
 title = {{virtuoso}: Interface to {`Virtuoso'} using {`ODBC'}},
 year = {2021}
}

@article{borkum2014usage,
 author = {Borkum, Mark I and Frey, Jeremy G},
 journal = {Journal of Cheminformatics},
 note = {\url{https://doi.org/10.1186/1758-2946-6-18}},
 pages = {18},
 title = {Usage and applications of Semantic Web techniques and technologies to support chemistry research},
 volume = {6},
 year = {2014}
}

@inproceedings{bratt2005semanticweb,
 author = {Bratt, Steve},
 booktitle = {2005 IEEE International Symposium on Mass Storage Systems and Technology},
 note = {\url{https://doi.org/10.1109/LGDI.2005.1612480}},
 pages = {124--128},
 publisher = {IEEE},
 title = {Toward a web of data and programs},
 year = {2005}
}

@misc{brickley2014rdfs,
 author = {Brickley, Dan and Guha, R.V.},
 howpublished = {\url{https://www.w3.org/TR/rdf-schema/#ch_introduction}},
 note = {Accessed: 2025-06-12},
 title = {{RDF} Schema 1.1},
 year = {2014}
}

@inproceedings{champin2003graphsimilarity,
 author = {Champin, Pierre-Antoine and Solnon, Christine},
 booktitle = {5th {Int}. {Conf}. {On} {Case}-{Based} {Reasoning} ({ICCBR} 2003)},
 note = {\url{https://doi.org/10.1007/3-540-45006-8_9}},
 pages = {80--95},
 publisher = {Springer},
 title = {Measuring the similarity of labeled graphs},
 year = {2003}
}

@inproceedings{degiacomo1996tbox,
 author = {De Giacomo, Giuseppe and Lenzerini, Maurizio},
 booktitle = {Proceedings of the 1996 International Workshop on Description Logics, November 2-4, 1996, Cambridge, MA, {USA}},
 note = {\url{https://aaai.org/papers/037-ws96-05-004/}},
 pages = {37--48},
 publisher = {{AAAI} Press},
 title = {{TBox} and {ABox} reasoning in expressive description logics},
 year = {1996}
}

@article{djoumbou2016classyfire,
 author = {Djoumbou Feunang, Yannick and Eisner, Roman and Knox, Craig and Chepelev, Leonid and Hastings, Janna and Owen, Gareth and Fahy, Eoin and Steinbeck, Christoph and Subramanian, Shankar and Bolton, Evan and others},
 journal = {Journal of Cheminformatics},
 note = {\url{https://doi.org/10.1186/s13321-016-0174-y}},
 pages = {1--20},
 title = {{ClassyFire}: automated chemical classification with a comprehensive, computable taxonomy},
 volume = {8},
 year = {2016}
}

@misc{ell2025graph,
 author = {Ell, Basil},
 howpublished = {\url{https://doi.org/10.48550/arXiv.2512.15308}},
 title = {Graph Pattern-based Association Rules Evaluated Under No-repeated-anything Semantics in the Graph Transactional Setting},
 year = {2025}
}

@article{fu2015pubchemrdf,
 author = {Fu, Gang and Batchelor, Colin and Dumontier, Michel and Hastings, Janna and Willighagen, Egon and Bolton, Evan},
 journal = {Journal of Cheminformatics},
 note = {\url{https://doi.org/10.1186/s13321-015-0084-4}},
 pages = {1--15},
 title = {{PubChemRDF}: towards the semantic annotation of {PubChem} compound and substance databases},
 volume = {7},
 year = {2015}
}

@article{gallegos2024explainable,
 author = {Gallegos, Miguel and Vassilev-Galindo, Valentin and Poltavsky, Igor and Mart{\'\i}n Pend{\'a}s, {\'A}ngel and Tkatchenko, Alexandre},
 journal = {Nature Communications},
 note = {\url{https://doi.org/10.1038/s41467-024-48567-9}},
 pages = {1--13},
 title = {Explainable chemical artificial intelligence from accurate machine learning of real-space chemical descriptors},
 volume = {15},
 year = {2024}
}

@article{garrison2023applying,
 author = {Garrison, Aaron G and Heras-Domingo, Javier and Kitchin, John R and dos Passos Gomes, Gabriel and Ulissi, Zachary W and Blau, Samuel M},
 journal = {Journal of Chemical Information and Modeling},
 note = {\url{https://doi.org/10.1021/acs.jcim.3c01226}},
 pages = {7642--7654},
 title = {Applying Large Graph Neural Networks to Predict Transition Metal Complex Energies Using the {tmQM\_wB97MV} Data Set},
 volume = {63},
 year = {2023}
}

@article{gee2025enriched,
 author = {Gee, Winston and Doyle, Abigail and Vargas, Santiago and Alexandrova, Anastassia N},
 journal = {Digital Discovery},
 note = {\url{https://doi.org/10.1039/D5DD00220F}},
 pages = {3378--3388},
 title = {Multi-level {QTAIM}-enriched graph neural networks for resolving properties of transition metal complexes},
 volume = {4},
 year = {2025}
}

@article{goh2017deepchem,
 author = {Goh, Garrett B and Hodas, Nathan O and Vishnu, Abhinav},
 journal = {Journal of Computational Chemistry},
 note = {\url{https://doi.org/10.1002/jcc.24764}},
 pages = {1291--1307},
 title = {Deep learning for computational chemistry},
 volume = {38},
 year = {2017}
}

@article{gordon1988chemical,
 author = {Gordon, John E},
 journal = {Journal of Chemical Information and Computer Sciences},
 note = {\url{https://doi.org/10.1021/ci00058a011}},
 pages = {100--115},
 title = {Chemical inference. 3. Formalization of the language of relational chemistry: ontology and algebra},
 volume = {28},
 year = {1988}
}

@article{griego2021acceleration,
 author = {Griego, Charles D and Kitchin, John R and Keith, John A},
 journal = {International Journal of Quantum Chemistry},
 note = {\url{https://doi.org/10.1002/qua.26380}},
 pages = {e26380},
 title = {Acceleration of catalyst discovery with easy, fast, and reproducible computational alchemy},
 volume = {121},
 year = {2021}
}

@article{groom2016cambridge,
 author = {Groom, C. R. and Bruno, I. J. and Lightfoot, M. P. and Ward, S. C.},
 journal = {Acta Crystallographica Section B: Structural Science, Crystal Engineering and Materials},
 note = {\url{https://doi.org/10.1107/S2052520616003954}},
 pages = {171--179},
 title = {The {C}ambridge Structural Database},
 volume = {72},
 year = {2016}
}

@misc{harris2013sparql,
 author = {Harris, Steve and Seaborne, Andy},
 howpublished = {\url{https://www.w3.org/TR/2013/REC-sparql11-query-20130321/}},
 note = {Accessed: 2025-11-18},
 title = {{SPARQL} 1.1 Query Language},
 year = {2013}
}

@article{hastings2011chemical,
 author = {Hastings, Janna and Chepelev, Leonid and Willighagen, Egon and Adams, Nico and Steinbeck, Christoph and Dumontier, Michel},
 journal = {PloS One},
 note = {\url{https://doi.org/10.1371/journal.pone.0025513}},
 pages = {e25513},
 title = {The chemical information ontology: provenance and disambiguation for chemical data on the biological semantic web},
 volume = {6},
 year = {2011}
}

@misc{hayes2004rdfsemantics,
 author = {Hayes, Patrick},
 howpublished = {\url{https://www.w3.org/TR/2004/REC-rdf-mt-20040210/#rdfs_interp}},
 note = {Accessed: 2025-06-12},
 title = {{RDF} Semantics},
 year = {2004}
}

@article{he2024reaction,
 author = {He, Linke and Fu, Yulong and Hou, Shaoyi and Wang, Guoqiang and Zhao, Jiabao and Xing, Yipeng and Li, Shuhua and Ma, Jing},
 journal = {Artificial Intelligence Chemistry},
 note = {\url{https://doi.org/10.1016/j.aichem.2023.100034}},
 pages = {100034},
 title = {Reaction condition-and functional group-specific knowledge discovery: Data-and computation-based analysis on transition-metal-free transformation of organoborons},
 volume = {2},
 year = {2024}
}

@article{heid2023uncertainty,
 author = {Heid, Esther and McGill, Charles J and Vermeire, Florence H and Green, William H},
 journal = {Journal of Chemical Information and Modeling},
 note = {\url{https://doi.org/10.1021/acs.jcim.3c00373}},
 pages = {4012--4029},
 title = {Characterizing uncertainty in machine learning for chemistry},
 volume = {63},
 year = {2023}
}

@article{hogan2021knowledge,
 author = {Hogan, Aidan and Blomqvist, Eva and Cochez, Michael and d’Amato, Claudia and Melo, Gerard De and Gutierrez, Claudio and Kirrane, Sabrina and Gayo, Jos{\'e} Emilio Labra and Navigli, Roberto and Neumaier, Sebastian and others},
 journal = {ACM Computing Surveys (Csur)},
 note = {\url{https://doi.org/10.1145/3447772}},
 pages = {1--37},
 title = {Knowledge graphs},
 volume = {54},
 year = {2021}
}

@misc{iupac2025element,
 author = {{IUPAC}},
 howpublished = {https://goldbook.iupac.org/terms/view/C01022},
 note = {Online version 5.0.0. Accessed: 2026-01-28},
 title = {{`chemical element'} in \emph{{IUPAC} compendium of chemical terminology}},
 year = {2025}
}

@article{jiang2001median,
 author = {Jiang, Xiaoyi and Munger, A. and Bunke, H.},
 journal = {IEEE Transactions on Pattern Analysis and Machine Intelligence},
 note = {\url{https://doi.org/10.1109/34.954604}},
 pages = {1144--1151},
 title = {On median graphs: properties, algorithms, and applications},
 volume = {23},
 year = {2001}
}

@misc{jin2019junction,
 author = {Jin, Wengong and Barzilay, Regina and Jaakkola, Tommi},
 howpublished = {\url{https://doi.org/10.48550/arXiv.1802.04364}},
 title = {Junction Tree Variational Autoencoder for Molecular Graph Generation},
 year = {2019}
}

@article{johnstone2009highdimensions,
 author = {Johnstone, Iain M. and Titterington, D. Michael},
 journal = {Philosophical Transactions of the Royal Society A: Mathematical, Physical and Engineering Sciences},
 note = {\url{https://doi.org/10.1098/rsta.2009.0159}},
 pages = {4237--4253},
 title = {Statistical challenges of high-dimensional data},
 volume = {367},
 year = {2009}
}

@article{kanza2019new,
 author = {Kanza, Samantha and Frey, Jeremy Graham},
 journal = {Expert Opinion on Drug Discovery},
 note = {\url{https://doi.org/10.1080/17460441.2019.1586880}},
 pages = {433--444},
 title = {A new wave of innovation in Semantic web tools for drug discovery},
 volume = {14},
 year = {2019}
}

@article{kevlishvili2024language,
 author = {Kevlishvili, Ilia and St. Michel, Roland G. and Garrison, Aaron G. and Toney, Jacob W. and Adamji, Husain and Jia, Haojun and Román-Leshkov, Yuriy and Kulik, Heather J.},
 journal = {Faraday Discussions},
 note = {\url{https://doi.org/10.1039/D4FD00087K}},
 pages = {275--303},
 title = {Leveraging natural language processing to curate the {tmCAT}, {tmPHOTO}, {tmBIO}, and {tmSCO} datasets of functional transition metal complexes},
 volume = {256},
 year = {2025}
}

@article{kitson2023bnsurvey,
 author = {Kitson, Neville Kenneth and Constantinou, Anthony C. and Guo, Zhigao and Liu, Yang and Chobtham, Kiattikun},
 journal = {Artificial Intelligence Review},
 note = {\url{https://doi.org/10.1007/s10462-022-10351-w}},
 pages = {8721--8814},
 title = {A survey of {Bayesian} Network structure learning},
 volume = {56},
 year = {2023}
}

@misc{klyne2004rdfsyntax,
 author = {Klyne, Graham and Carrol, Jeremy J.},
 howpublished = {\url{https://www.w3.org/TR/2004/REC-rdf-concepts-20040210/#section-Graph-URIref}},
 note = {Accessed: 2025-12-03},
 title = {Resource Description Framework ({RDF}): Concepts and Abstract Syntax},
 year = {2004}
}

@article{kneiding2023learning,
 author = {Kneiding, Hannes and Lukin, Ruslan and Lang, Lucas and Reine, Simen and Pedersen, Thomas Bondo and De Bin, Riccardo and Balcells, David},
 journal = {Digital Discovery},
 note = {\url{https://doi.org/10.1039/D2DD00129B}},
 pages = {618--633},
 title = {Deep learning metal complex properties with natural quantum graphs},
 volume = {2},
 year = {2023}
}

@article{kneiding2024optmization,
 author = {Kneiding, Hannes
and Nova, Ainara
and Balcells, David},
 journal = {Nature Computational Science},
 note = {\url{https://doi.org/10.1038/s43588-024-00616-5}},
 pages = {263-273},
 title = {Directional multiobjective optimization of metal complexes at the billion-system scale},
 volume = {4},
 year = {2024}
}

@article{kneiding2025tmqmgstar,
 author = {Kneiding, Hannes and Balcells, David},
 journal = {Journal of Chemical Information and Modeling},
 note = {\url{https://doi.org/10.1021/acs.jcim.5c01958}},
 pages = {11766--11777},
 title = {{tmQMg*} Data Set: Excited State Properties of 74k Transition Metal Complexes},
 volume = {65},
 year = {2025}
}

@misc{kneiding2026delta,
 author = {Kneiding, Hannes and Balcells, David},
 howpublished = {\url{https://doi.org/10.26434/chemrxiv.10001643/v1}},
 title = {{$\Delta$}-Machine Learning for the Prediction of Metal Complex Properties},
 year = {2026}
}

@book{koller2009probabilistic,
 author = {Koller, Daphne and Friedman, Nir},
 note = {ISBN: 978-0-262-01319-2},
 publisher = {MIT Press},
 title = {Probabilistic Graphical Models: Principles and Techniques},
 year = {2009}
}

@article{lai2025silhouette,
 author = {Lai, Huixia and Huang, Tao and Lu, BinLong and Zhang, Shi and Xiaog, Ruliang},
 journal = {Neural Computing and Applications},
 note = {\url{https://doi.org/10.1007/s00521-024-10706-0}},
 pages = {3061–-3075},
 title = {Silhouette coefficient-based weighting k-means algorithm},
 volume = {37},
 year = {2025}
}

@inproceedings{lapin2016topk,
 author = {Lapin, Maksim and Hein, Matthias and Schiele, Bernt},
 booktitle = {Proceedings of the IEEE Conference on Computer Vision and Pattern Recognition},
 note = {\url{https://doi.org/10.1109/CVPR.2016.163}},
 pages = {1468--1477},
 publisher = {IEEE},
 title = {Loss functions for top-k error: Analysis and insights},
 year = {2016}
}

@book{lavravc2021representation,
 author = {Lavra{\v{c}}, Nada and Podpe{\v{c}}an, Vid and Robnik-{\v{S}}ikonja, Marko},
 note = {ISBN: 3-030-68817-8},
 publisher = {Springer},
 title = {Representation Learning},
 year = {2021}
}

@book{legendre2012numeco,
 author = {Legendre, Pierre},
 note = {ISBN: 1-283-73471-0},
 publisher = {Elsevier},
 title = {Numerical Ecology},
 year = {2012}
}

@article{liang2023transition,
 author = {Liang, Qi and Wang, Sizhe and Yao, Yao and Dong, Peng and Song, Haojie},
 journal = {Advanced Functional Materials},
 note = {\url{https://doi.org/10.1002/adfm.202300825}},
 pages = {2300825},
 title = {Transition metal compounds family for {Li--S} batteries: the {DFT}-guide for suppressing polysulfides shuttle},
 volume = {33},
 year = {2023}
}

@article{meggio2026path,
 author = {Claudio Meggio and Johan Pensar and David Balcells and Riccardo De Bin},
 journal = {Statistical Modelling},
 note = {\url{https://doi.org/10.1177/1471082X251407458}},
 pages = {1471082X251407458},
 title = {A path-based boosting algorithm for exploring transition metal compounds},
 volume = {OnlineFirst},
 year = {2026}
}

@article{nandy2021computational,
 author = {Nandy, Aditya and Duan, Chenru and Taylor, Michael G and Liu, Fang and Steeves, Adam H and Kulik, Heather J},
 journal = {Chemical Reviews},
 note = {\url{https://doi.org/10.1021/acs.chemrev.1c00347}},
 pages = {9927--10000},
 title = {Computational discovery of transition-metal complexes: from high-throughput screening to machine learning},
 volume = {121},
 year = {2021}
}

@misc{nia2025personalized,
 author = {Nia, Sohrab Namazi and Ghosh, Subhodeep and Roy, Senjuti Basu and Amer-Yahia, Sihem},
 howpublished = {\url{https://doi.org/10.48550/arXiv.2502.12998}},
 title = {Personalized Top-k Set Queries Over Predicted Scores},
 year = {2025}
}

@inproceedings{nozza2022completion,
 author = {Nozza, Debora and Bianchi, Federico and Lauscher, Anne and Hovy, Dirk and others},
 booktitle = {Proceedings of the Second Workshop on Language Technology for Equality, Diversity and Inclusion},
 note = {\url{https://doi.org/10.18653/v1/2022.ltedi-1.4}},
 pages = {26--34},
 publisher = {Association for Computational Linguistics},
 title = {Measuring harmful sentence completion in language models for {LGBTQIA+} individuals},
 year = {2022}
}

@article{park2020assessment,
 author = {Park, Heewoong and Park, Jonghun},
 journal = {Applied Sciences},
 note = {\url{https://doi.org/10.3390/app10041340}},
 pages = {1--18},
 title = {Assessment of Word-Level Neural Language Models for Sentence Completion},
 volume = {10},
 year = {2020}
}

@article{pascazio2023chemical,
 author = {Pascazio, Laura and Rihm, Simon and Naseri, Ali and Mosbach, Sebastian and Akroyd, Jethro and Kraft, Markus},
 journal = {Journal of Chemical Information and Modeling},
 note = {\url{https://doi.org/10.1021/acs.jcim.3c00820}},
 pages = {6569--6586},
 title = {Chemical species ontology for data integration and knowledge discovery},
 volume = {63},
 year = {2023}
}

@manual{rcoreteam2024r,
 author = {{R Core Team}},
 howpublished = {https://www.R-project.org/},
 organization = {R Foundation for Statistical Computing},
 title = {{R}: A Language and Environment for Statistical Computing},
 year = {2024}
}

@article{salton1988termweighting,
 author = {Salton, Gerard and Buckley, Christopher},
 journal = {Information Processing \& Management},
 note = {\url{https://doi.org/10.1016/0306-4573(88)90021-0}},
 pages = {513--523},
 title = {Term-weighting approaches in automatic text retrieval},
 volume = {24},
 year = {1988}
}

@article{scutari2010bnlearn,
 author = {Marco Scutari},
 journal = {Journal of Statistical Software},
 note = {\url{https://doi.org/10.18637/jss.v035.i03}},
 pages = {1--22},
 title = {Learning {Bayesian} Networks with the {bnlearn} {R} Package},
 volume = {35},
 year = {2010}
}

@inproceedings{sorlin2005reactive,
 author = {Sorlin, Sébastien and Solnon, Christine},
 booktitle = {Graph-{Based} {Representations} in {Pattern} {Recognition}},
 note = {\url{https://doi.org/10.1007/978-3-540-31988-7_16}},
 pages = {172--182},
 publisher = {Springer},
 title = {Reactive Tabu Search for Measuring Graph Similarity},
 year = {2005}
}

@article{willighagen2013chembl,
 author = {Willighagen, Egon L and Waagmeester, Andra and Spjuth, Ola and Ansell, Peter and Williams, Antony J and Tkachenko, Valery and Hastings, Janna and Chen, Bin and Wild, David J},
 journal = {Journal of Cheminformatics},
 note = {\url{https://doi.org/10.1186/1758-2946-5-23}},
 pages = {1--12},
 title = {The {ChEMBL} database as linked open data},
 volume = {5},
 year = {2013}
}

\begin{artappendix}
\renewcommand{\thepart}{}
\renewcommand{\partname}{}
\addcontentsline{toc}{section}{Appendix}
\part{Appendix}
\parttoc
\section{An introduction to the Resource Description Framework}\apxlabel{rdf_intro}
This section introduces the concept of knowledge graph and the vocabulary of the Resource Description Framework (RDF) and of its extension RDF Schema (RDFS).

\subsection{Knowledge graphs and the Resource Description Framework}\apxlabel{methods_kg}
A \emph{knowledge graph} (KG) can be roughly defined as a \say{\emph{graph of data intended to accumulate and convey knowledge [...], whose nodes represent entities [...] and whose edges represent [...] relations between these entities}} \citep{hogan2021knowledge}. In practice, a KG is usually modelled as a \emph{directed labelled graph} (DLG), i.e., a set of labelled nodes (entities) paired with a set of labelled edges (relations). Labelled edges, together with their endpoints, are supposed to encode elementary sentences of the form
\begin{equation*}
\mathrm{subject} \xrightarrow{\mathrm{predicate}} \mathrm{object},
\end{equation*}
such as 
\begin{equation*}
    \mathrm{KCEYPT} \xrightarrow{\mathrm{hasMetalCentre}} \mathrm{Pt}. 
\end{equation*}
More formally, given a (possibly infinite) set of terms $\mathcal{T}$, a DLG can be defined as a subset $G \subseteq \mathcal{T}^3$, thus as a set of triples $(t_s, t_p, t_o) \in G$ \citep{ell2025graph}. The first and the third element represent, respectively, the \emph{subject} and the \emph{object}, and are to be interpreted as the \emph{nodes} of the graph. The \emph{edges} of $G$ are identified by the triples themselves, with the convention that the edge is directed from $t_s$ (the tail) to $t_o$ (the head), with the label being $t_p$. This formulation is highly flexible, as it does not prescribe any specific scheme or topological constraint (like tree-shaped graphs). As a consequence, the resulting knowledge base can be easily modified: new pieces of information can be added by simply adding new triples. This flexibility, together with the availability of a standardised vocabulary (introduced below), makes this framework particularly accessible, both from a development and a deployment perspective.\\

Before proceeding, it is fundamental to stress the asymmetric treatment of nodes and edges implicitly stated in this definition, in order to avoid confusion. While edges requires a label assignment, via a triple structure, nodes are identified with the term (i.e., the label) used to represent them. This automatically implies that distinct nodes cannot be labelled using the same term. Contrarily, it is permitted for two different edges to receive the same label $t_p$ but, on the other hand, we adopt here the restriction that at most one edge can exist between any two nodes. This last constraint is particularly restrictive with respect to the KG literature, where it is typically absent, but, as will become clear below, multiple edges are not necessary for this work and hence are prohibited solely for a matter of notational simplicity. This restriction does not lead to any loss of generality as all the experimental methodologies reported in \sctref{experiments} and \apxref{experiments_methods} can be immediately adapted to the multi-edge case. Moreover, since nodes are implicitly defined via edge specifications, a natural consequence is that there can never be an isolated node within a DLG.
Finally, the definition of a graph given here is notably different from the classical one, which usually explicitly prescribe a node set and an edge set. Again, this choice is motivated by a desire for simplicity. Further details and more rigorous definitions are available in \apxref{similarity_metrics}.

\subsubsection{The Resource Description Framework}\apxlabel{methods_rdf}
One specific model for KGs that makes use of the DLG formulation is the \emph{Resource Description Framework (RDF)} \citep{hogan2021knowledge}, which categorises terms into three different classes: \emph{Uniform Resource Identifiers} (URIs), which uniquely identify entities, \emph{literals}, which are meant to represent quantitative data such as strings or numbers, and \emph{blank nodes}, which are only used to state the existence of entities of interest and, in technical terms, behave as existentially quantified variables \citep{hayes2004rdfsemantics, bergmann2014logic}. At the heart of RDF lies the intention of combining flexibility, simplicity and power of expressivity and this effort ultimately lead to a solid and standardised framework \citep{klyne2004rdfsyntax, hayes2004rdfsemantics, brickley2014rdfs}.

Using the mathematical language introduced above, if $\mathcal{U}, \mathcal{L}$ and $\mathcal{B}$ are the (pairwise disjoint) sets of URIs, literals and blank nodes respectively, then $\mathcal{T} = \mathcal{U} \cup \mathcal{L} \cup \mathcal{B}$ and an RDF graph is a DLG $G \subseteq (\mathcal{U} \cup \mathcal{B}) \times \mathcal{U} \times \mathcal{T}$ \citep{ell2025graph}. In other words, literals are only allowed to be placed in object position and predicates may only be represented via URIs.

Example \ref{ex:rdf} shows a simple RDF graph, encoding few elementary facts using the elements just introduced.

\begin{example}\label{ex:rdf}
In this example, URIs, blank nodes and literals are all represented as strings. URIs begin with \texttt{:}, blank nodes with \texttt{\_:} and literals are enclosed in quotation marks (\say{\texttt{...}}).

Here we present a simple knowledge graph that encodes that within a given molecule (\texttt{:molecule1}) there exist two atoms (\texttt{:hasAtom}), one with chemical label (\texttt{:isAtom}) H (\texttt{:H}), and the other with atomic number (\texttt{:atomicNumber}) $6$. Notice that this knowledge base does not explicitly identify the atoms as named entities or resources, but merely states its existence within the molecule, hence the statements will involve two blank nodes (\texttt{\_:aaa} and \texttt{\_:bbb}). \figref{example_rdf} shows the DLG representation of this set of facts, as well as the corresponding set of triples.
\end{example}

\begin{figure}
    \begin{center}
        \begin{tabular}{c c}
        \begin{subfigure}[t]{0.4\linewidth}
            \raisebox{2cm}{%
            \begin{tabular}{l}
                 \texttt{:molecule1 :hasAtom \_:aaa .} \\
                 \texttt{\_:aaa :isAtom :H .} \\
                 \texttt{:molecule1 :hasAtom \_:bbb .} \\
                 \texttt{\_:bbb :hasAtomicNumber} \say{\texttt{6}} \texttt{.}
            \end{tabular}
            }
            \caption{}
        \end{subfigure} &
        \begin{subfigure}[t]{0.4\linewidth}
            \includegraphics[width=\linewidth]{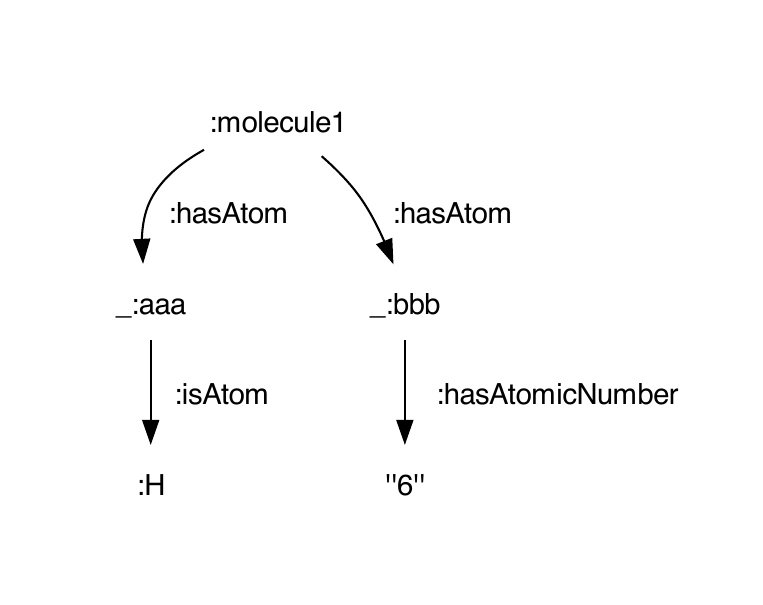}
            \caption{}
        \end{subfigure}
        \end{tabular}
    \end{center}
    \caption{(a) A simple set of facts, in the form of a set of triples written in Turtle syntax. These facts are stated using URIs (\texttt{:*}),  blank nodes (\texttt{\_:*}) and literals (\say{\texttt{*}}) (b) The equivalent directed labelled graph representation of the same set of facts.}\figlabel{example_rdf}
\end{figure}

\paragraph{URI syntax}\apxlabel{methods_rdf_uri}
URIs are, as mentioned above, tools meant to uniquely identify a resource. In practice, a URI is a string that satisfies specific syntactical constraints \citep{klyne2004rdfsyntax}. A more specific class of URI is made up of so called \emph{http URIs}, which identify resources available on the Web and point to the locations of those resources on their respective networks \citep{ayers2008uri}. An example of such URI is
\begin{center}
    \texttt{http://www.w3.org/1999/02/22-rdf-syntax-ns\#type},
\end{center}
which represents a predicate used below (see \apxref{methods_kg_rdfs}) to express class assignments. This example also shows how URIs can be particularly long and cumbersome to express in plain text, which is why in the RDF literature it is common to find \emph{prefixed names} \citep{beckett2014turtle}, i.e. a notation that allows to express URIs by contracting part of the string using uniquely identified \emph{namespaces}. Following the previous example, even though this convention is not restricted to http URIs, we can define the namespace \texttt{rdf} as
\begin{center}
    \texttt{@prefix rdf: http://www.w3.org/1999/02/22-rdf-syntax-ns\#}
\end{center}
in compliance with the Turtle syntax \citep{beckett2014turtle}.
In this statement, \texttt{@prefix} is an instruction that defines a new namespace, \texttt{rdf:} is the chosen namespace ID and \texttt{http://www.w3.org/1999/02/22-rdf-syntax-ns\#} is the associated URI string. Now, whenever we wish to use the URI above, we can simply use \texttt{rdf:type} instead. Prefixes can also be empty strings, as it happened in Example~\ref{ex:rdf} \citep{beckett2014turtle}. We recall that a short summary of the namespaces used in this work is available in \tblref{prefixes_short} while the full list can be found in \apxref{comprehensive}.

\subsubsection{The RDF Schema extension}\apxlabel{methods_kg_rdfs}
RDF alone is a simple assertional language, and it is not capable of specifying descriptions or relations among \emph{types} of objects in a way that also allow for the \emph{entailment} of new knowledge (for instance, if we know that A is of type X and that X is a subtype of Y, then it is natural to infer that A is also of type Y; this is not natively possible with the sole RDF vocabulary). This functionality is provided by the \emph{RDF Schema} (RDFS) semantic extension, in which proper concepts of \emph{classes} and \emph{subclasses} are introduced, together with a set of entailment rules that define natural inheritance mechanisms \citep{hayes2004rdfsemantics, brickley2014rdfs}. Class assignments are expressed via the predicate \texttt{rdf:type}, whereas subclass inheritance is expressed via \texttt{rdfs:subClassOf}. In a similar way, an appropriate semantic structure regarding properties is introduced. The predicate \texttt{rdfs:subPropertyOf} allows to mark two properties as being one a subproperty of the other (and hence the superproperty will apply every time that the subproperty does). The predicates \texttt{rdfs:domain} and \texttt{rdfs:range}, on the other hand, are used to define the domain and the range of a property (in a similar way to which domain and range is defined for a mathematical function), thus enabling automatic class assignments for the subject and the object of any triple in which the property of interest acts as a predicate.\\

By leveraging on the RDFS vocabulary, it is possible to separate the content of a knowledge base into two components, namely the terminology box (TBox), a set of assertions about the concepts used in the knowledge base, and the assertion box (ABox), a set of assertions on individual instances and objects \citep{degiacomo1996tbox, pascazio2023chemical}. Example~\ref{ex:rdfs} extends the RDF graph from Example~\ref{ex:rdf} by defining an appropriate RDFS terminology.\\

\begin{example}\label{ex:rdfs}
Here we extend the RDF graph presented in Example~\ref{ex:rdf} by including RDFS terminology (see \figref{example_rdfs}). In particular, we define the classes \texttt{:ClassMolecule}, \texttt{:ClassAtom} and \texttt{:ClassChemicalElement} that represent molecules, atoms and chemical labels respectively. Class definition statements are formed by declaring that an object is an instance (\texttt{rdf:type}) of \texttt{rdfs:Class}. Afterwards, we specify domain and range of the three predicates \texttt{:hasAtom}, \texttt{:isAtom} and \texttt{hasAtomicNumber}.\\
For example, the domain of \texttt{:hasAtom} is \texttt{:ClassMolecule}, which means that any term found in subject position in a triple which has \texttt{:hasAtom} as predicate is of class \texttt{:ClassMolecule}. Since the triple $(\texttt{:molecule1}, \texttt{:hasAtom}, \texttt{\_:aaa})$ is present in the RDF graph, this means that we can infer the triple $(\texttt{:molecule1}, \texttt{rdf:type}, \texttt{:ClassMolecule})$, i.e., \texttt{:molecule1} is an instance of the class \texttt{:ClassMolecule}.\\
In the case of the predicate \texttt{hasAtomicNumber}, the range is defined as \texttt{xmls:integer}. This is a class defined within the \texttt{xmls} namespace.
\end{example}

\begin{figure}
    \begin{center}
            \begin{tabular}{l l}
                \textbf{TBox:} & \\
                & \texttt{:ClassMolecule rdf:type rdfs:Class .} \\
                & \texttt{:ClassAtom rdf:type rdfs:Class .} \\
                & \texttt{:ClassChemicalElement rdf:type rdfs:Class .} \\[3pt]
                & \texttt{:hasAtom rdfs:domain :ClassMolecule .} \\
                & \texttt{:hasAtom rdfs:range :ClassAtom .} \\[3pt]
                & \texttt{:isAtom rdfs:domain :ClassAtom .} \\
                & \texttt{:isAtom rdfs:range :ClassChemicalElement .} \\[3pt]
                & \texttt{:hasAtomicNumber rdfs:domain :ClassAtom .} \\
                & \texttt{:hasAtomicNumber rdfs:range xmls:integer .} \\
                \textbf{ABox:} & \\
                & \texttt{:molecule1 :hasAtom \_:aaa .} \\
                & \texttt{\_:aaa :isAtom :H .} \\
                & \texttt{:molecule1 :hasAtom \_:bbb .} \\
                & \texttt{\_:bbb :hasAtomicNumber} \say{\texttt{6}} \texttt{.}
            \end{tabular}
    \end{center}
    \caption{The set of facts from \figref{example_rdf} is expanded by adding a TBox. The TBox defines several classes that clarify the role and nature of the entities in the ABox, while also specifying appropriate domain and range constraints for the predicates.}\figlabel{example_rdfs}
\end{figure}

This work makes use of the RDF data model and vocabulary and of the RDFS vocabulary and entailment scheme in order to provide semantic representations of TMCs. There are also other frameworks that can be used to model KGs, such as \emph{property graphs}, which also allow for the annotation of nodes with key-value pairs. While such a framework could seem to be more desirable in a context in which quantitative properties (like node features) are naturally paired with structural descriptors, this formulation is not as standardised as RDF \citep{hogan2021knowledge}. For simplicity, then, RDF and RDFS are the chosen means of representation.\\

\section{Comprehensive tmQM-RDF namespaces, classes and properties list}\apxlabel{comprehensive}
This section offers a comprehensive summary of the namespaces, RDFS classes and properties that can be found in tmQM-RDF. Namespaces are reported in \tblref{prefixes}. The full list of RDFS classes can be found in \tblref{classes}. Properties for the complex, ligand and atomic level, respectively, are listed in Tables \ref{tbl:properties_complex}, \ref{tbl:properties_ligand} and \ref{tbl:properties_atomic}.

{
\begin{table}
    \centering
    \makebox[\textwidth][c]{%
    \begin{tabular}{l l l}
         \toprule
         Prefix & Namespace & Short description \\
         \midrule
         \midrule
         rdf & http://www.w3.org/1999/02/22-rdf-syntax-ns\# & The RDF Concepts Vocabulary (RDF)\\
         rdfs & http://www.w3.org/2000/01/rdf-schema\# & The RDF Schema vocabulary (RDFS)\\
         xmls & http://www.w3.org/2001/XMLSchema\# & The XML Schema vocabulary\\[0.25cm]
         inp5 & resource://integreat/p5/ & Baseline non-specific URI path\\[0.25cm]
         cm & resource://integreat/p5/complex/ & Base URI path for complex level terms\\
         cmT & resource://integreat/p5/complex/TMC/ & General TMC-related facts\\
         cmTp & resource://integreat/p5/complex/TMC/property/ & TMC properties\\[0.25cm]
         lg & resource://integreat/p5/ligand/ & Base URI path for ligand level terms\\
         lgC & resource://integreat/p5/ligand/centre/ & General metal centre-related facts\\
         lgCp & resource://integreat/p5/ligand/centre/property/ & Metal centre properties\\
         lgCr & resource://integreat/p5/ligand/centre/reference/ & Abstract metal centre representations\\
         lgCrp & resource://integreat/p5/ligand/centre/reference/property/ & Properties of metal centre representations\\
         lgL & resource://integreat/p5/ligand/ligand/ & General ligand-related facts\\
         lgLp & resource://integreat/p5/ligand/ligand/property/ & Ligand properties\\
         lgLr & resource://integreat/p5/ligand/ligand/reference/ & Abstract ligand representations\\
         lgLrp & resource://integreat/p5/ligand/ligand/reference/property/ & Properties of abstract ligand representations\\
         lgLro & resource://integreat/p5/ligand/ligand/reference/occurrence/ & Ligand occurrences used as references\\
         lgLrm & resource://integreat/p5/ligand/ligand/reference/motif/ & Structural elements of ligand representations\\
         lgB & resource://integreat/p5/ligand/bond/ & Ligand-centre bond objects\\
         lgBp & resource://integreat/p5/ligand/bond/property/ & Ligand-centre bond properties\\
         lgBr & resource://integreat/p5/ligand/bond/reference/ & Abstract ligand-centre bond representations\\
         lgBrp & resource://integreat/p5/ligand/bond/reference/property/ & Properties of ligand-centre bond representations\\
         lgS & resource://integreat/p5/ligand/structure/ & Ligand level structural connectivity\\[0.25cm]
         tm & resource://integreat/p5/atomic/ & Base URI path for atomic level terms\\
         tmA & resource://integreat/p5/atomic/atom/ & General atom-related facts\\
         tmAp & resource://integreat/p5/atomic/atom/property/ & Atomic properties\\
         tmAr & resource://integreat/p5/atomic/atom/reference/ & Abstract atom representations\\
         tmArp & resource://integreat/p5/atomic/atom/reference/property/ & Atomic representation properties\\
         tmB & resource://integreat/p5/atomic/bond/ & Atomic bond objects\\
         tmBp & resource://integreat/p5/atomic/bond/property/ & Atomic bond properties\\
         tmBr & resource://integreat/p5/atomic/bond/reference/ & Abstract atomic bond representations\\
         tmBrp & resource://integreat/p5/atomic/bond/reference/property/ & Properties of atomic bonds representations\\
         tmS & resource://integreat/p5/atomic/structure/ & Atomic level structural connectivity\\[0.25cm]
         ds & resource://integreat/p5/datasets/ & Datasets\\
         dsC & resource://integreat/p5/datasets/complexes/ & tmQM\\
         dsG & resource://integreat/p5/datasets/graphs/ & tmQMg\\
         dsL & resource://integreat/p5/datasets/ligands/ & tmQMg-L\\[0.25cm]
         nm & resource://integreat/p5/numerical/ & Numerical methods\\
         \bottomrule
    \end{tabular}%
    }
    \caption{The prefixes used in this paper and the corresponding namespaces.}
    \tbllabel{prefixes}
\end{table}
}

{
\begin{table}[p]
    \centering
    \makebox[\textwidth][c]{%
    \begin{tabular}{l l}
        \toprule
        Class & Description \\
        \midrule
        \midrule
        inp5:PropertyType & \makecell[tl]{A possible property type, intended as one of the many types\\ of properties that are reported in the tmQM series} \\
        inp5:ComplexPropertyValue & \makecell[tl]{A complex property value that cannot be represented using a single URI/literal} \\
        inp5:CartesianCoordinates3D & \makecell[tl]{A set of coordinates in the Cartesian coordinate system of the \\three-dimensional Euclidean space $\R^3$} \\
        inp5:CountDescription & \makecell[tl]{The description of a count property, intended as a bag of items, each\\ one representing the count of a specific object} \\
        inp5:CountDescriptionItem & \makecell[tl]{An item in the description of a count property, representing the count\\ of a specific object} \\
        inp5:ObservedProperty & \makecell[tl]{An observation of a property} \\
        inp5:Countable & \makecell[tl]{A countable entity} \\[0.25cm]
        cmT:TransitionMetalComplex & \makecell[tl]{Transition Metal Complexes} \\
        cmTp:MetaDataDescription & \makecell[tl]{An object whose purpose is to collect metadata about a TMC} \\[0.25cm]
        lgC:MetalCentre & \makecell[tl]{A metal centre} \\
        lgCr:MetalCentreClass & \makecell[tl]{A metal centre class, intended as a label that can be assigned to a metal centre} \\
        lgCr:MetalCentre\_\texttt{*} & \makecell[tl]{A specific metal centre class, subclass of MetalCentre and instance of\\ MetalCentreClass. \texttt{*} must be replaced with the symbol of the desired metal centre} \\
        lgL:Ligand & \makecell[tl]{A ligand} \\
        lgLr:LigandClass & \makecell[tl]{A ligand class, intended as a label that can be assigned to a ligand} \\
        lgLr:Ligand\_\texttt{*} & \makecell[tl]{A specific ligand class, subclass of Ligand and instance of LigandClass.\\ \texttt{*} must be replaced with the ID of the desired ligand class} \\
        lgLr:StableOccurrenceDescription & \makecell[tl]{The description of the most stable occurrence of a ligand class, made\\ of a subgraph name and a ligand instance (if available in the dataset)} \\
        lgLr:BindingAtomsSMILES & \makecell[tl]{The description of the binding atoms of a ligand in terms of indices \\within a SMILES string. It is a collection of BindingSiteSMILES objects} \\
        lgLr:BindingSiteSMILES & \makecell[tl]{The description of a group of haptic atoms in a ligand in terms of \\indices within a SMILES string} \\
        lgLro:LigandOccurrenceClass & \makecell[tl]{The types of occurrences/representations of a ligand that can be \\encountered and used for computations (i.e., SMILES string, most\\ stable occurrence, relaxed structure)} \\
        lgLrm:StructuralFeature & \makecell[tl]{A structural feature of a ligand that can be counted} \\
        lgB:LigandBond & \makecell[tl]{A chemical bond between a metal centre and a binding site within a ligand} \\[0.25cm]
        tmA:Atom & \makecell[tl]{An atom} \\
        tmAr:Element & \makecell[tl]{A chemical element, intended as a label that can be assigned to an atom} \\
        tmAr:\texttt{*} & \makecell[tl]{A specific chemical element, subclass of Atom and instance of Element.\\ \texttt{*} must be replaced with the chemical symbol of the desired element} \\
        tmB:AtomicBond & \makecell[tl]{A chemical bond between two atoms} \\
        tmB:NBOType & \makecell[tl]{A possible NBO type, intended as a label that can be assigned to a\\ chemical bond} \\[0.65cm]
        ds:Dataset & \makecell[tl]{A dataset} \\[0.25cm]
        nm:Optimisation & \makecell[tl]{A method for the optimisation of molecular geometries via relaxation\\ of energy gradients} \\
        nm:Singlepoint & \makecell[tl]{A method for the computation of the energy and other electronic structure\\ properties of optimized geometries} \\
        \bottomrule
    \end{tabular}%
    }
    \caption{The RDFS classes employed in tmQM-RDF.}
    \tbllabel{classes}
\end{table}
}

\begin{table}
    \centering
    \makebox[\textwidth][c]{%
    \begin{tabular}{l l}
        \toprule
        Property & Description \\
        \midrule
        \midrule
        cmTp:lowest\_vibrational\_frequency & Lowest vibrational frequency in cm-1 \\
        cmTp:metal\_node\_degree & Metal node degree \\
        cmTp:enthalpy\_energy\_correction & Enthalpy energy correction \\
        cmTp:homo\_lumo\_gap\_delta & Change of HOMO-LUMO between different levels of theory \\
        cmTp:heat\_capacity & Heat capacity \\
        cmTp:metal\_center\_period & Periodic table period of the metal \\
        cmTp:metal\_center\_element & Metal element \\
        cmTp:n\_atoms & Total number of atoms \\
        cmTp:dipole\_moment\_delta & Change of dipole moment between different levels of theory \\
        cmTp:homo\_lumo\_gap & HOMO-LUMO gap \\
        cmTp:CSD\_years & Year of the CSD update to which the data refers to \\
        cmTp:gibbs\_energy\_correction & Gibbs free energy correction \\
        cmTp:dipole\_moment & Dipole moment in Debyes \\
        cmTp:metal\_center\_group & Periodic table group of the metal \\
        cmTp:electronic\_energy\_delta & Electronic energy difference between the different levels of theory \\
        cmTp:highest\_vibrational\_frequency & Vibrational frequency with the largest value, in cm-1 \\
        cmTp:charge & Overal charge of the metal complex \\
        cmTp:stoichiometry & Stoichiometry of the metal complex, in Hill format \\
        cmTp:spin & Overal spin multiplicity of the metal complex \\
        cmTp:dispersion\_energy\_delta & Dispersion energy difference between the different levels of theory \\
        cmTp:metal\_node\_natural\_charge & Natural atomic charge of the metal node, in e units \\
        cmTp:polarisability & Overal polarizability of the metal complex \\
        cmTp:zpe\_correction & Zero-point energy \\
        cmTp:homo\_energy & HOMO energy \\
        cmTp:electronic\_energy & Electronic energy \\
        cmTp:dispersion\_energy & Dispersion energy \\
        cmTp:lumo\_energy & LUMO energy \\
        cmTp:molecular\_mass & Molecular mass \\
        cmTp:gibbs\_energy & Gibbs energy \\
        cmTp:element\_counts & Element counts \\
        cmTp:n\_electrons & Total number of electrons in the metal complex \\
        cmTp:enthalpy\_energy & Enthalpy energy \\
        cmTp:entropy & Entropy energy \\
        \bottomrule
    \end{tabular}%
    }
    \caption{The URIs representing the complex level properties.}
    \tbllabel{properties_complex}
\end{table}

\begin{table}
    \centering
    \makebox[\textwidth][c]{%
    \begin{tabular}{l l}
        \toprule
        Property & Description \\
        \midrule
        \midrule
        lgLrp:n\_haptic\_bound & Number of haptic bonds \\
        lgLrp:metal\_bound\_homo\_energy & HOMO energy involving metal bound atom \\
        lgLrp:I2\_over\_I3 & Moments of inertia I2 and I3 ratio \\
        lgLrp:metal\_bound\_homo\_p & p orbital character of metal-bound HOMO \\
        lgLrp:smiles\_metal\_bond\_node\_idx\_groups & Indicies of metal-bound atoms in the SMILES \\
        lgLrp:solid\_angle & Ligand coordination solid angle \\
        lgLrp:n\_atoms & Ligand total number of atoms \\
        lgLrp:homo\_lumo\_gap & Ligand HOMO-LUMO gap \\
        lgLrp:sasa\_area\_stable & Ligand SASA area (stable) \\
        lgLrp:buried\_volume & Ligand buried volume \\
        lgLrp:metal\_bound\_homo\_f & f orbital character of metal-bound HOMO \\
        lgLrp:sasa\_volume\_free & Ligand SASA volume (free) \\
        lgLrp:metal\_bound\_lumo\_d & d orbital character of metal-bound LUMO \\
        lgLrp:metal\_bound\_lumo\_f & f orbital character of metal-bound LUMO \\
        lgLrp:haptic\_element\_counts & Element-wise counting of haptic-metal-bound atoms \\
        lgLrp:dipole\_moment & Ligand dipole moment \\
        lgLrp:is\_alternative\_charge & \makecell[tl]{Whether the charge reported in tmQMg-L disagrees with \\the charge reported in the OctLig dataset for the same ligand} \\
        lgLrp:metal\_bound\_lumo\_p & p orbital character of metal-bound LUMO \\
        lgLrp:charge & Ligand charge \\
        lgLrp:largest\_frequency & Ligand largest frequency \\
        lgLrp:n\_dentic\_bound & Number of dentic-metal-bound atoms \\
        lgLrp:stable\_occurrence & The most stable occurrence of a ligand \\
        lgLrp:n\_metal\_bound & Number of metal-bound atoms \\
        lgLrp:metal\_bound\_lumo\_s & s orbital character of metal-bound LUMO \\
        lgLrp:metal\_bound\_lumo\_energy & Metal-bound LUMO energy \\
        lgLrp:solid\_cone\_angle & Ligand solid cone angle \\
        lgLrp:structure\_counts & Counts of selected structures within ligands \\
        lgLrp:polarisability & Ligand polarizability \\
        lgLrp:exact\_cone\_angle & Ligand exact cone angle \\
        lgLrp:molar\_volume & Ligand molar volume \\
        lgLrp:sasa\_area\_free & Ligand SASA area (free) \\
        lgLrp:I1\_over\_I3 & Moments of inertia I1 and I3 ratio \\
        lgLrp:dentic\_element\_counts & Number of dentic-metal-bound atoms \\
        lgLrp:logp & Ligand logP \\
        lgLrs:SMILES & Ligand SMILES \\
        lgLrp:G\_parameter & Ligand G parameter \\
        lgLrp:element\_counts & Element-wise count of ligand atoms \\
        lgLrp:sasa\_volume\_stable & Ligand SASA volume (stable) \\
        lgLrp:metal\_bound\_homo\_d & d orbital character of metal-bound HOMO \\
        lgLrp:metal\_bound\_homo\_s & s orbital character of metal-bound HOMO \\
        \bottomrule
    \end{tabular}%
    }
    \caption{The URIs representing the ligand level properties.}
    \tbllabel{properties_ligand}
\end{table}

\begin{table}
    \centering
    \makebox[\textwidth][c]{%
    \begin{tabular}{l l}
        \toprule
        Property & Description \\
        \midrule
        \midrule
        tmAp:lone\_pair\_max\_occupation & Maximum electron occupation of the NBO lone pairs \\
        tmAp:natural\_electron\_configuration\_d\_occupation & Electron occupation of the d orbitals \\
        tmAp:n\_lone\_vacancies & Number of lone vacancy orbitals \\
        tmAp:natural\_electron\_population\_valence & Electron occupation of valence orbitals \\
        tmAp:lone\_vacancy\_min\_energy & Minimum energy of lone vacancies \\
        tmAp:lone\_vacancy\_min\_p\_occupation & Minimum p occupation of lone vacancy orbitals \\
        tmAp:lone\_pair\_max\_p\_occupation & Maximum p occupation of lone pair orbitals \\
        tmAp:lone\_pair\_max\_d\_occupation & Maximum d occupation of lone pair orbitals \\
        tmAp:lone\_vacancy\_min\_occupation & Minimum occupation of lone vacancy orbitals \\
        tmAp:lone\_pair\_max\_energy & Maximum energy of lone pair orbitals \\
        tmAp:lone\_vacancy\_min\_s\_occupation & Minimum s occupation of lone vacancy orbitals \\
        tmAp:natural\_electron\_configuration\_s\_occupation & Electron occupation of the s orbitals \\
        tmAp:node\_position & Atom index in associated xyz geometry \\
        tmAp:lone\_pair\_energy\_min\_max\_difference & \makecell[tl]{Energy difference between lowest- and\\ highest-energy lone pair orbitals} \\
        tmAp:natural\_atomic\_charge & Natural atomic charge, in e units \\
        tmAp:hydrogen\_count & Number of hydrogen atoms bound to atomic node \\
        tmAp:n\_lone\_pairs & Number of lone pair orbitals \\
        tmAp:lone\_vacancy\_min\_d\_occupation & Minimum d occupation of lone vacancy orbitals \\
        tmAp:lone\_vacancy\_energy\_min\_max\_difference & \makecell[tl]{Energy difference between lowest- and\\ highest-energy lone vacancy orbitals} \\
        tmAp:valency\_index & Valency index from NBO analysis \\
        tmAp:natural\_electron\_configuration\_p\_occupation & Electron occupation of the p orbitals \\
        tmAp:node\_id & \makecell[tl]{A unique internal identifier used to\\ distinguish the nodes of the atomic graph} \\
        tmAp:lone\_pair\_max\_s\_occupation & Maximum s occupation of lone pair orbitals \\
        tmBp:bond\_max\_s\_occupation & Maximum s occupation of bonding orbitals \\
        tmBp:bond\_max\_p\_occupation & Maximum p occupation of bonding orbitals \\
        tmBp:wiberg\_bond\_order & Wiberg bond order \\
        tmBp:nbo\_type & NBO type defining the bond \\
        tmBp:bond\_max\_energy & Highest energy of bonding orbitals \\
        tmBp:bond\_max\_d\_occupation & Maximum d occupation of bonding orbitals \\
        tmBp:bond\_energy\_min\_max\_difference & \makecell[tl]{Energy difference between lowest- and\\ highest-energy bonding orbitals} \\
        tmBp:bond\_distance & Bond distance, in Å \\
        tmBp:antibond\_min\_s\_occupation & Minimum s occupation of antibonding orbitals \\
        tmBp:bond\_max\_occupation & Maximum overall occupation of bonding orbitals \\
        tmBp:n\_bn & Number of bonding orbitals for bond \\
        tmBp:antibond\_energy\_min\_max\_difference & \makecell[tl]{Energy difference between lowest- and\\ highest-energy antibonding orbitals} \\
        tmBp:antibond\_min\_occupation & Minimum overall occupation of antibonding orbitals \\
        tmBp:n\_nbn & Number of antibonding orbitals for bond \\
        tmBp:antibond\_min\_d\_occupation & Minimum d occupation of antibonding orbitals \\
        tmBp:antibond\_min\_p\_occupation & Minimum p occupation of antibonding orbitals \\
        tmBp:antibond\_min\_energy & Minimum energy of antibonding orbitals \\
        tmArp:atomic\_number & Atomic number \\
        \bottomrule
    \end{tabular}%
    }
    \caption{The URIs representing the atomic level properties.}
    \tbllabel{properties_atomic}
\end{table}

\section{Thematic selections for experimental purposes}\apxlabel{selection}
As the tmQM-RDF dataset is extremely rich and diverse, it can be convenient to have access to a smaller subset of TMCs on which to conduct preliminary experiments. To this end, we have introduced, in \sctref{dataset_selection}, two selections of $1600$ TMCs extracted from the population of TMCs in tmQM-RDF, each divided into a training set ($1000$ TMCs), a validation set ($300$ TMCs) and a test set ($300$ TMCs). These two selections have been designed with the intent of replicating the variety found in tmQM-RDF while also restricting the scope of the datasets to a precise chemically relevant subpopulation of complexes. In particular, one selection (the \emph{earlyTM} selection) only contains complexes whose metal centres are either Cr, Mo, W (which are early transition metals) whereas the other (the \emph{lateTM} selection) only admits the metal centres Pd, Ni, Pt (late transition metals; these also happen to be the three most frequent centres that appear in tmQM-RDF, as shown in \figref{tmQM_RDF_summary_centres}).\\

The two selections have been sampled from the full dataset following the procedure below:
\begin{enumerate}
    \item a \say{seed} of the $N_\mathrm{seed}$ most frequent ligands found within the subset of tmQM-RDf identified by the desired metal centres has been computed.
    \item a candidate set made of all the complexes whose ligands are all included in the \say{seed} has been computed;
    \item from the candidate set, $1600$ TMCs have been sampled with probability inversely proportional to the number of atoms\footnote{In this phase, the relative proportions of TMCs having a given metal centres are constrained to match those found in tmQM-RDF};
    \item the train/validation/test split has been performed.
\end{enumerate}

A quick assessment of the representativeness of each selection can be performed by comparing the sets of the most frequent ligands within the selection to the same set extracted from tmQM-RDF (by only considering, however, complexes which have the same metal centres as the selection). \figrefmult{lig_sum_earlyTM}{lig_sum_lateTM} portray the results of this kind of analysis by showing the counts of the occurrences of the most frequent ligands in the original dataset and in the three partitions of each selection. Overall it can be seen that, while the sets of most frequent ligands do not coincide perfectly, for both selections at least half of the ligands retain their status of most frequent across all the datasets (specifically, for the \emph{earlyTM} selection those ligands are \emph{ligand1-0}, \emph{ligand24-0}, \emph{ligand4-0}, \emph{ligand45-0} and \emph{ligand38-0}, whereas for the \emph{lateTM} selection the ligands are \emph{ligand1-0}, \emph{ligand3-0}, \emph{ligand40-0}, \emph{ligand188-0}, \emph{ligand38-0} and \emph{ligand51-0}). Moreover, by looking at the expressiveness of each ligand-metal centre pair, it is possible to notice that, especially in the \emph{earlyTM} selection, there are only minor qualitative differences (in the \emph{lateTM} selection, on the other hand, it is more frequent to observe ligands for which one centre becomes overrepresented in the selection with respect to tmQM-RDF).

\begin{figure}
        \begin{subfigure}[b]{\linewidth}
            \includegraphics[width=\linewidth]{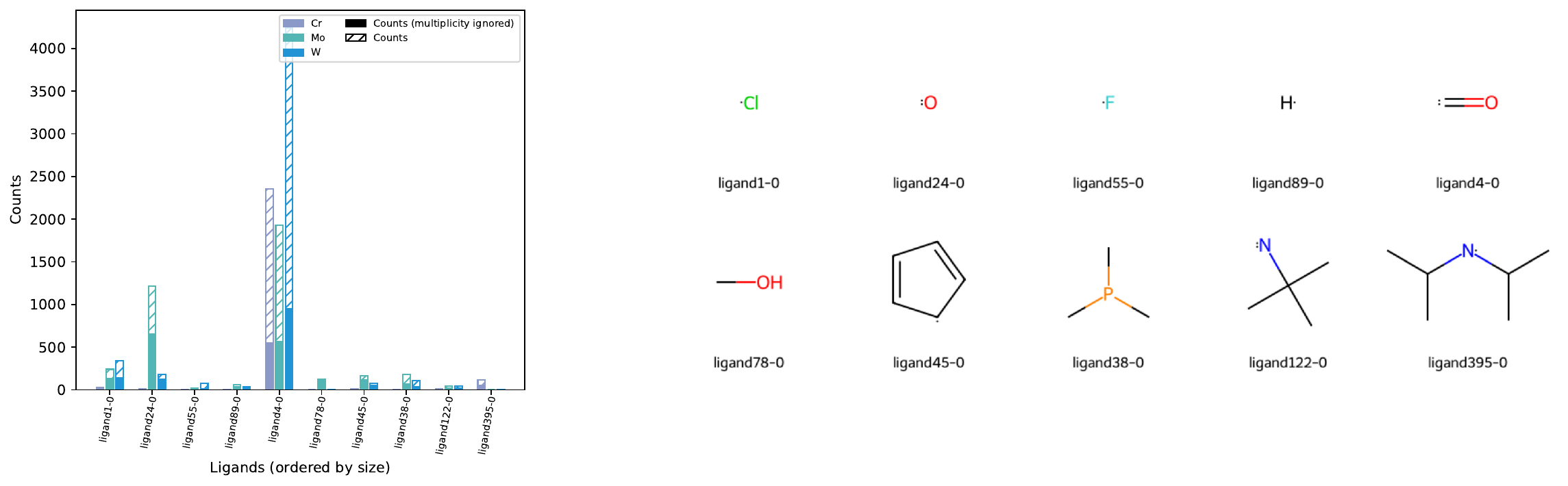}
            \caption{tmQM-RDF (Cr - Mo - W)}\figlabel{lig_sum_tmqm_early}
        \end{subfigure}
        
        \begin{subfigure}[b]{\linewidth}
            \includegraphics[width=\linewidth]{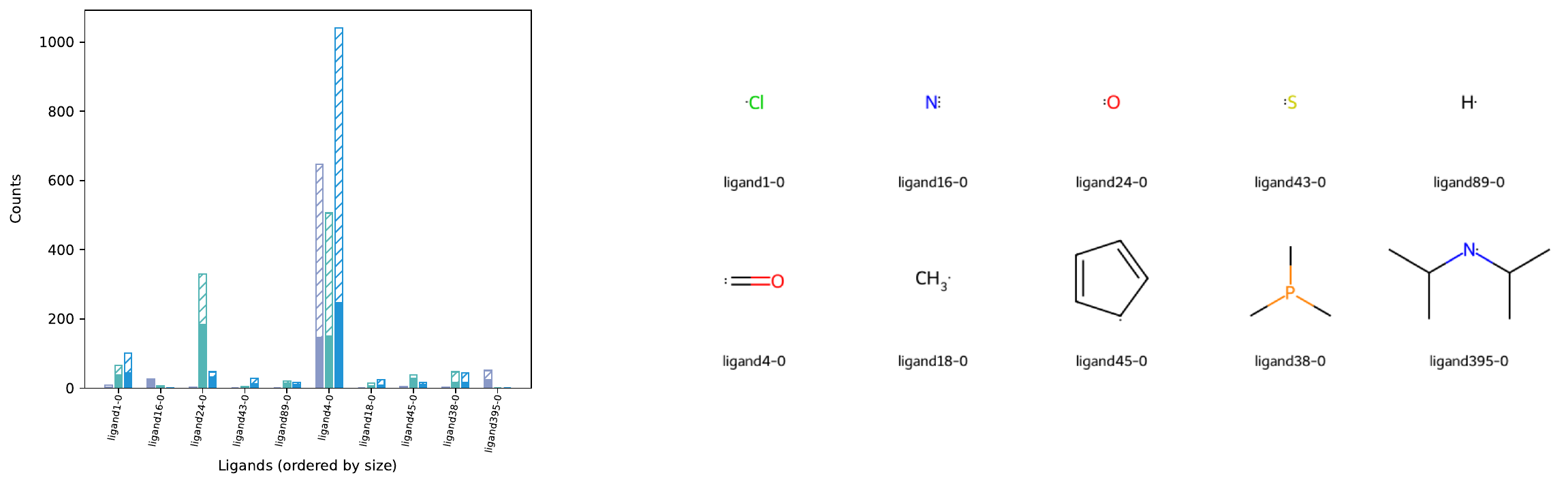}
            \caption{Train}\figlabel{lig_sum_train_early}
        \end{subfigure} 
        
        \begin{subfigure}[b]{\linewidth}
            \includegraphics[width=\linewidth]{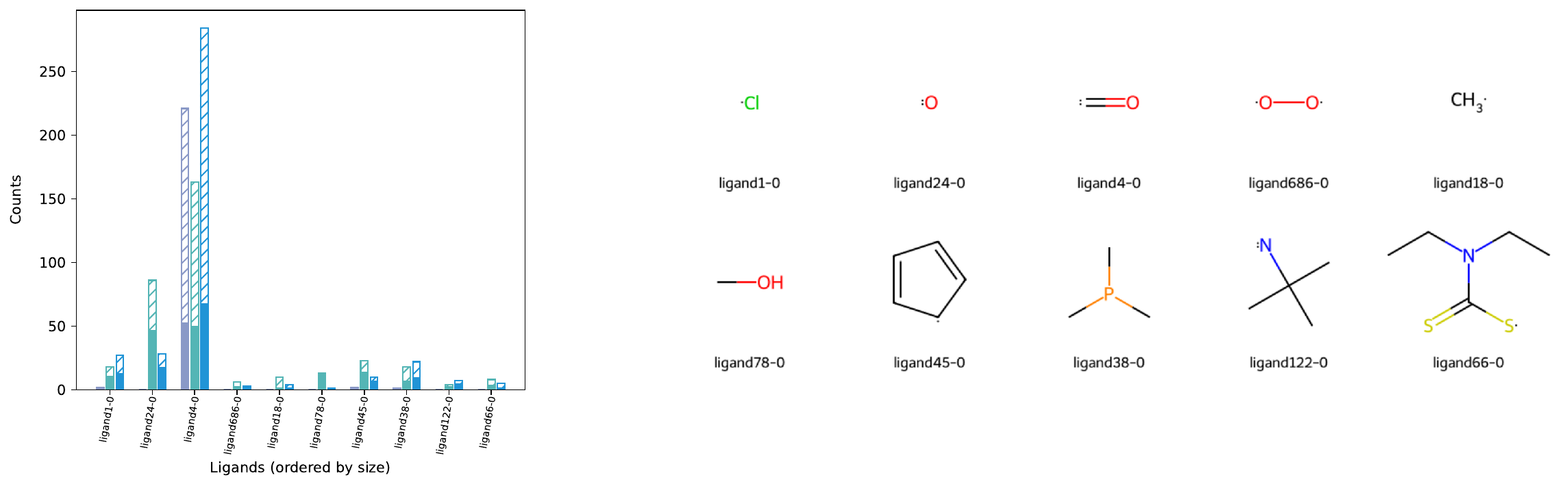}
            \caption{Validation}\figlabel{lig_sum_validation_early}
        \end{subfigure}

        \begin{subfigure}[b]{\linewidth}
            \includegraphics[width=\linewidth]{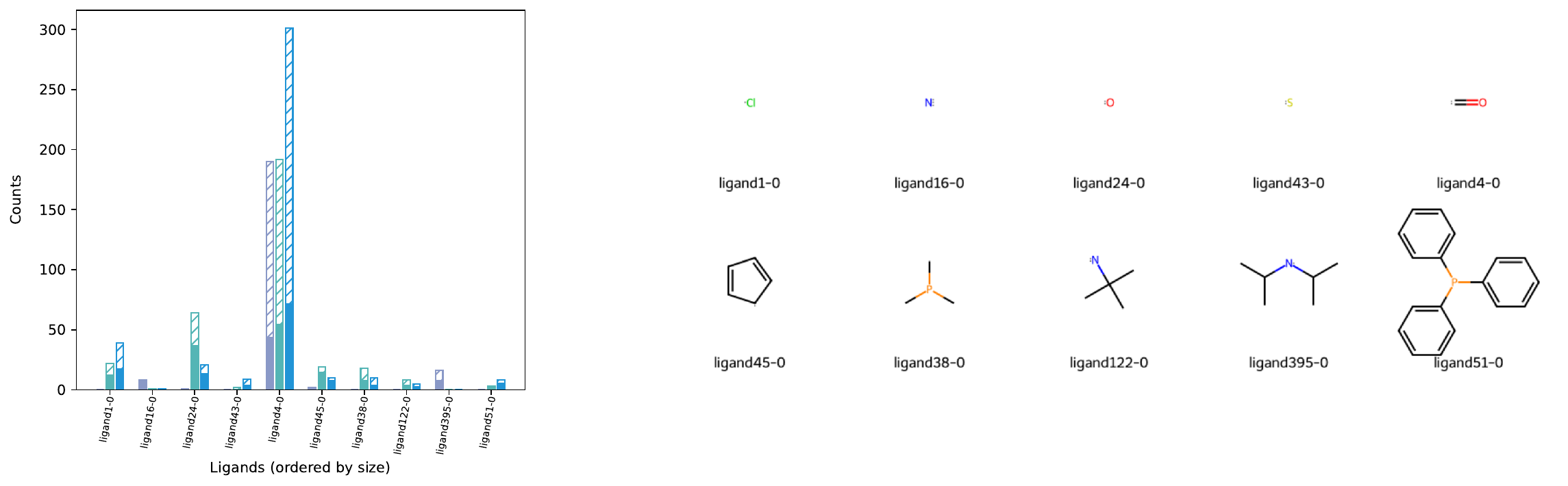}
            \caption{Test}\figlabel{lig_sum_test_early}
        \end{subfigure}
    
    \caption{The bar plots of the counts of the appearances of the 10 most frequent ligands, divided by metal centre, (left) and a visual representation of the same ligands (right) in: (A) tmQM-RDF, accounting only for Cr, Mo and W centres, (B), the training set, (C) the validation set and (D) the test set of the \emph{earlyTM} selection. In the barplots, two different counting methods are being represented: in one scenario, every single occurrence of a ligand is considered (dashed bars), in the other, multiple occurrences of a ligand within a single complex are not considered (dashed bars).}\figlabel{lig_sum_earlyTM}
\end{figure}

 \begin{figure}
        \begin{subfigure}[b]{\linewidth}
            \includegraphics[width=\linewidth]{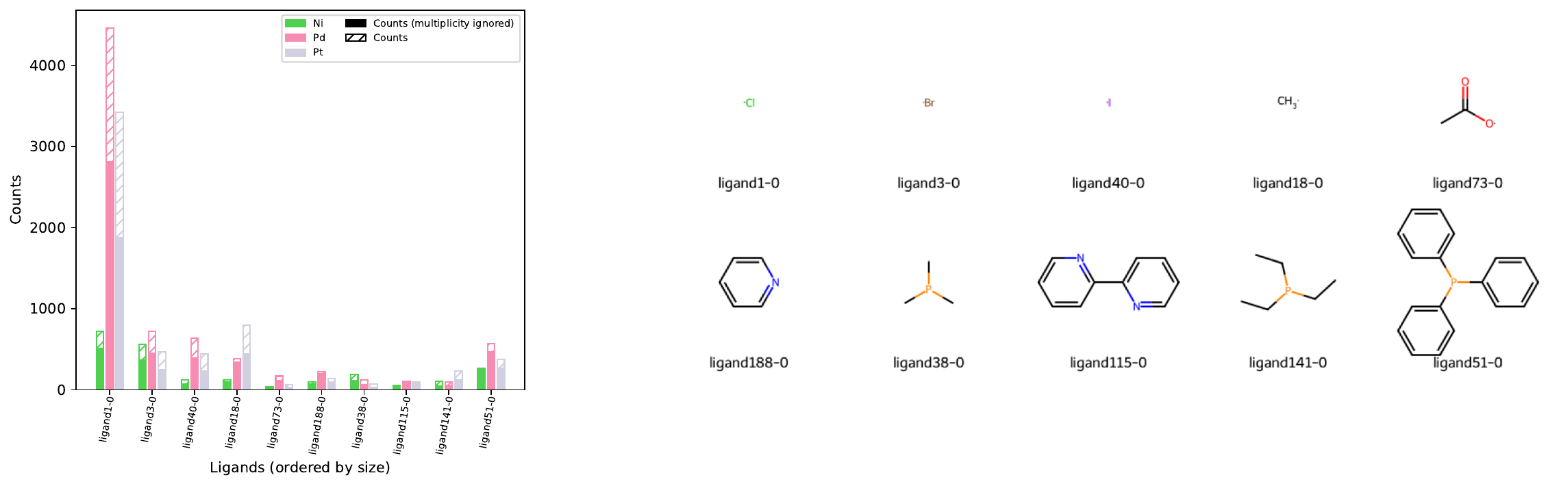}
            \caption{tmQM-RDF (Ni - Pd - Pt)}\figlabel{lig_sum_tmqm_late}
        \end{subfigure}
        
        \begin{subfigure}[b]{\linewidth}
            \includegraphics[width=\linewidth]{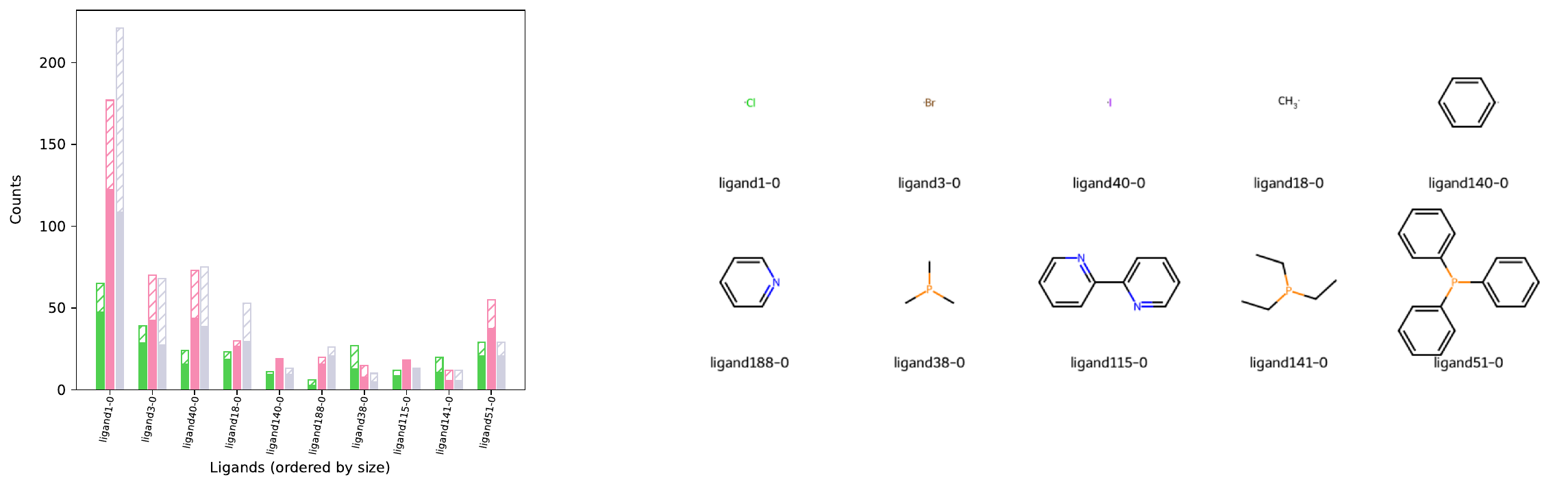}
            \caption{Train}\figlabel{lig_sum_train_late}
        \end{subfigure} 
        
        \begin{subfigure}[b]{\linewidth}
            \includegraphics[width=\linewidth]{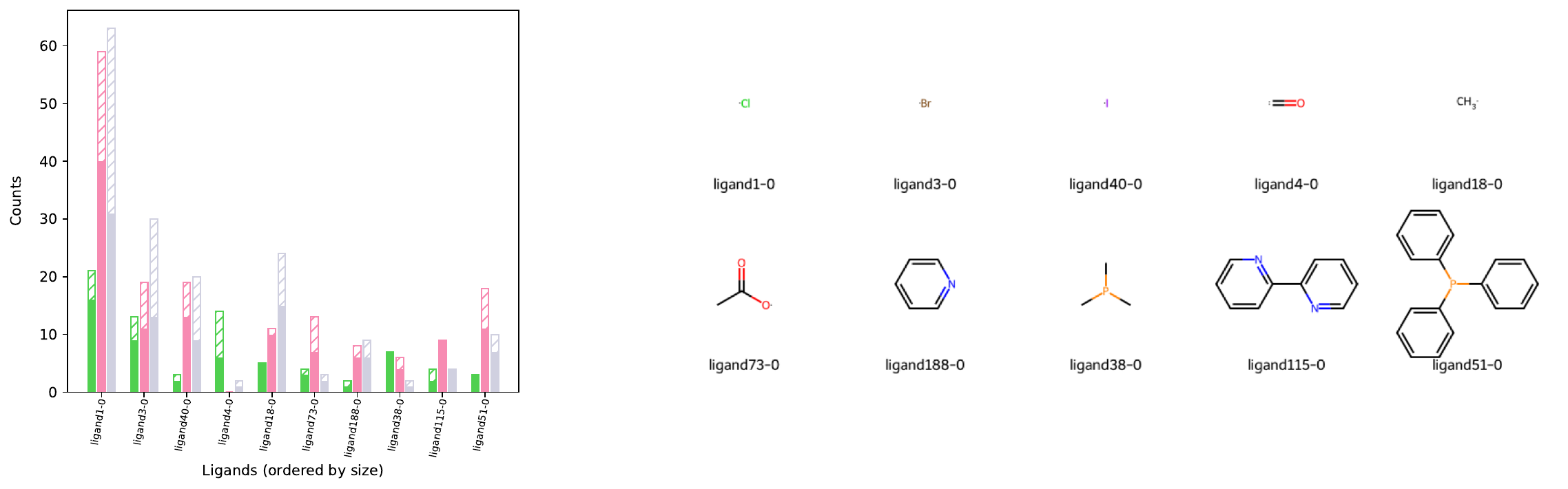}
            \caption{Validation}\figlabel{lig_sum_validation_late}
        \end{subfigure}

        \begin{subfigure}[b]{\linewidth}
            \includegraphics[width=\linewidth]{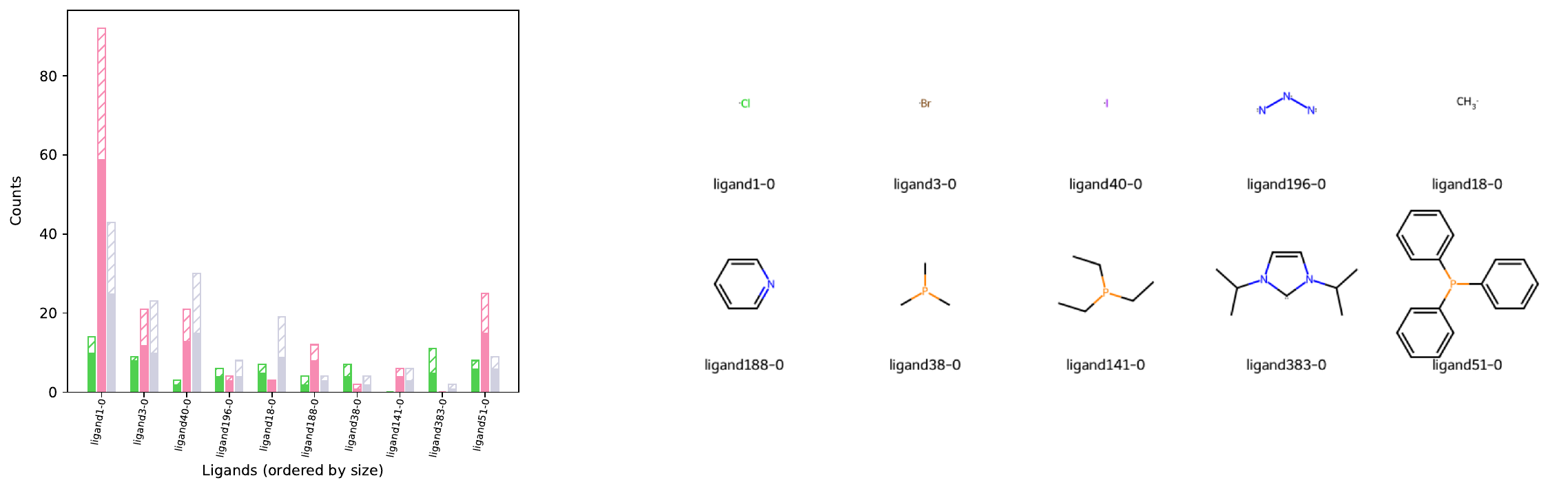}
            \caption{Test}\figlabel{lig_sum_test_late}
        \end{subfigure}
    
    \caption{The bar plots of the counts of the appearances of the 10 most frequent ligands, divided by metal centre, (left) and a visual representation of the same ligands (right) in: (A) tmQM-RDF, accounting only for Ni, Pd and Pt centres, (B), the training set, (C) the validation set and (D) the test set of the \emph{lateTM} selection. In the bar plots, two different counting methods are being represented: in one scenario, every single occurrence of a ligand is considered (dashed bars), in the other, multiple occurrences of a ligand within a single complex are not considered (dashed bars).}\figlabel{lig_sum_lateTM}
\end{figure}

\section{Experimental methods}\apxlabel{experiments_methods}
In this section we provide an initial account of the technical details behind the experimental methods employed in \sctref{experiments}. 

We recall that the purpose of the experiment is to extract relevant structural features from tmQM-RDF in the form of graph patterns. The patterns are clustered into families of typical substructures, which are then converted into binary graph-level features that indicate whether that family is expressed or not in a given graph. The joint distribution of the features is estimated and used to assign a score to a TMC. New TMCs are obtained by substituting one of the ligands in an existing TMC and then evaluated via the score function.

All the material needed to perform this learning task is formally introduced here. In the following, let $\mathcal{G}$ be the graph dataset representing the population of TMCs which is the object of the study. We partition $\mathcal{G}$ in a training set $\mathcal{G}_\mathrm{train}$, used to estimate the score function, and a test set $\mathcal{G}_{\mathrm{test}}$, which will be manipulated on the basis of the estimated score.

In \sctref{experiments} $\mathcal{G}$ has been simply presented as a subpopulation of the family of TMCs that are described in tmQM-RDF. Here, $\mathcal{G}$ is explicitly defined as a set of RDF graphs, as defined in \apxref{methods_rdf}, which correspond each to the subgraph of tmQM-RDF that encodes the information about the respective TMC of interest.\\

The first problem that has to be addressed is that of appropriately defining what a family of substructures is and how it can lead to the definition of a structural feature, in the context of this work and the tmQM-RDF dataset. Both of these concepts strongly rely on the notion of graph pattern, in the SPARQL sense \citep{harris2013sparql, ell2025graph}, which is here regarded as the primary data exploration tool used to acquire structural information about the TMCs under examination.

We will give a formal introduction to graph patterns and explain how patterns naturally lead to an elementary definition of a graph feature vector. Afterwards, we will move to the problem of clustering patterns into families of substructures using an agglomerative clustering algorithm \citep{legendre2012numeco} based on the notion of graph similarity. Finally, we will define the actual feature vector we will employ for our score assessment, by accordingly aggregating the elementary features.\\

These newly-defined features need to be studied from a statistical point of view, with the objective of defining a score function that can assess a TMC on the basis of its structural components. Here score is intended as the probability of the resulting feature vector, as measured by an estimate of the joint distribution of the features themselves. The challenges that arise when trying to estimate such a high-dimensional distribution are examined and dealt with by means of the Bayesian Network formulation \citep{koller2009probabilistic, kitson2023bnsurvey}.

\subsection{Graph patterns}\apxlabel{experiments_pattern_definition}
Informally, a graph pattern is a knowledge graph that employs the same language of the RDF dataset of interest (here, tmQM-RDF), with the addition of variables.\\
Using the notation introduced in \apxref{methods_kg}, if $\mathcal{T} = \mathcal{U} \cup \mathcal{L} \cup \mathcal{B}$ is the usual set of terms and $\mathcal{V}$ is a set of \emph{variables}, an \emph{RDF graph pattern} is a graph $p \subseteq (\mathcal{U} \cup \mathcal{B} \cup \mathcal{V}) \times (\mathcal{U} \cup \mathcal{V}) \times (\mathcal{T} \cup \mathcal{V})$ \citep{harris2013sparql, ell2025graph}. The sets of terms and variables that actually appear in $p$ are denoted with $\mathcal{T}_p$ and $\mathcal{V}_p$ respectively.\\

The presence of variables allows us to employ a graph pattern $p$ as a query against an RDF graph $G \in \mathcal{G}$, by asking if $G$ is expressing (in some sense) the structure represented by $p$. This leads to the definition of a \emph{match} of $p$ in $G$, which is intended to be a mapping $\mu: \mathcal{T}_p \cup \mathcal{V}_p \to \mathcal{T}_G$ such that $\mu(t) = t$ for every term $t \in \mathcal{T}_p$ and that $\mu(p) \subseteq G$, where $\mu(p)$ is a shorthand notation for the graph obtained by replacing every node and edge label in $p$ with its image under $\mu$. In fewer words, a match is a replacement (or grounding) of the variables in $p$ such that the resulting graph is a subgraph of $G$ \citep{harris2013sparql, ell2025graph}.

This work imposes an additional restriction to the definition of $\mu$, known as the \emph{\nrasem} \citep{ell2025graph}. This means that no variable in $p$ can be mapped to a term already present in $\mathcal{T}_p$ nor can two different variables be mapped to the same term (i.e., $\mu$ is required to be injective). The set of the matches of $p$ against $G$ under this evaluation semantics is denoted with $\Omega_{G,\, p}$.

Details on how $\mu$ and $\Omega_{G,p}$ are computed in this work, using the SPARQL query language \citep{harris2013sparql}, can be found in \apxref{pattern_matching}. \figref{example_graph_patterns} shows two example graph patterns, while also demonstrating how this concept can be used to capture structural information.\\

\begin{figure}[!t]
    \begin{center}
    \makebox[\textwidth][c]{%
    \begin{tabular}{c c}
    \begin{subfigure}[t]{0.37\linewidth}
        \includegraphics[width=\linewidth]{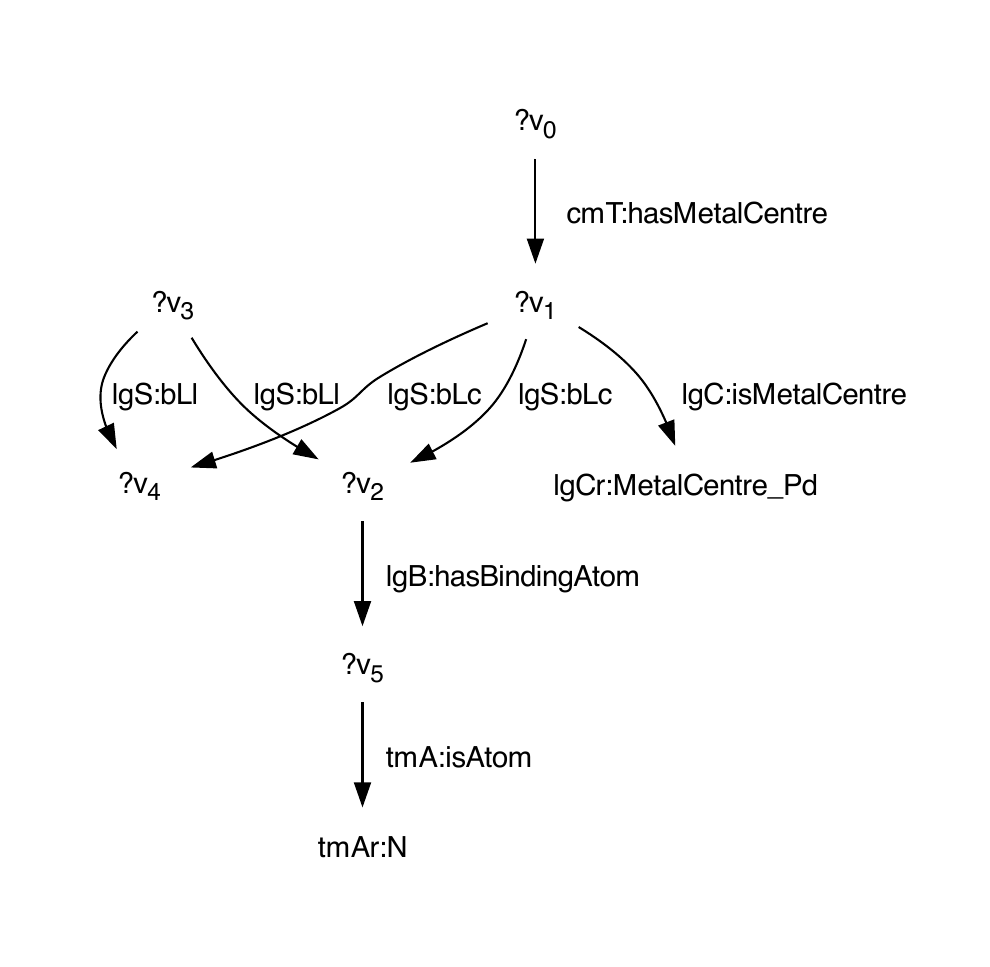}
        \caption{}\figlabel{ex_graph_pattern_1}
    \end{subfigure} &
    \begin{subfigure}[t]{0.62\linewidth}
        \includegraphics[width=\linewidth]{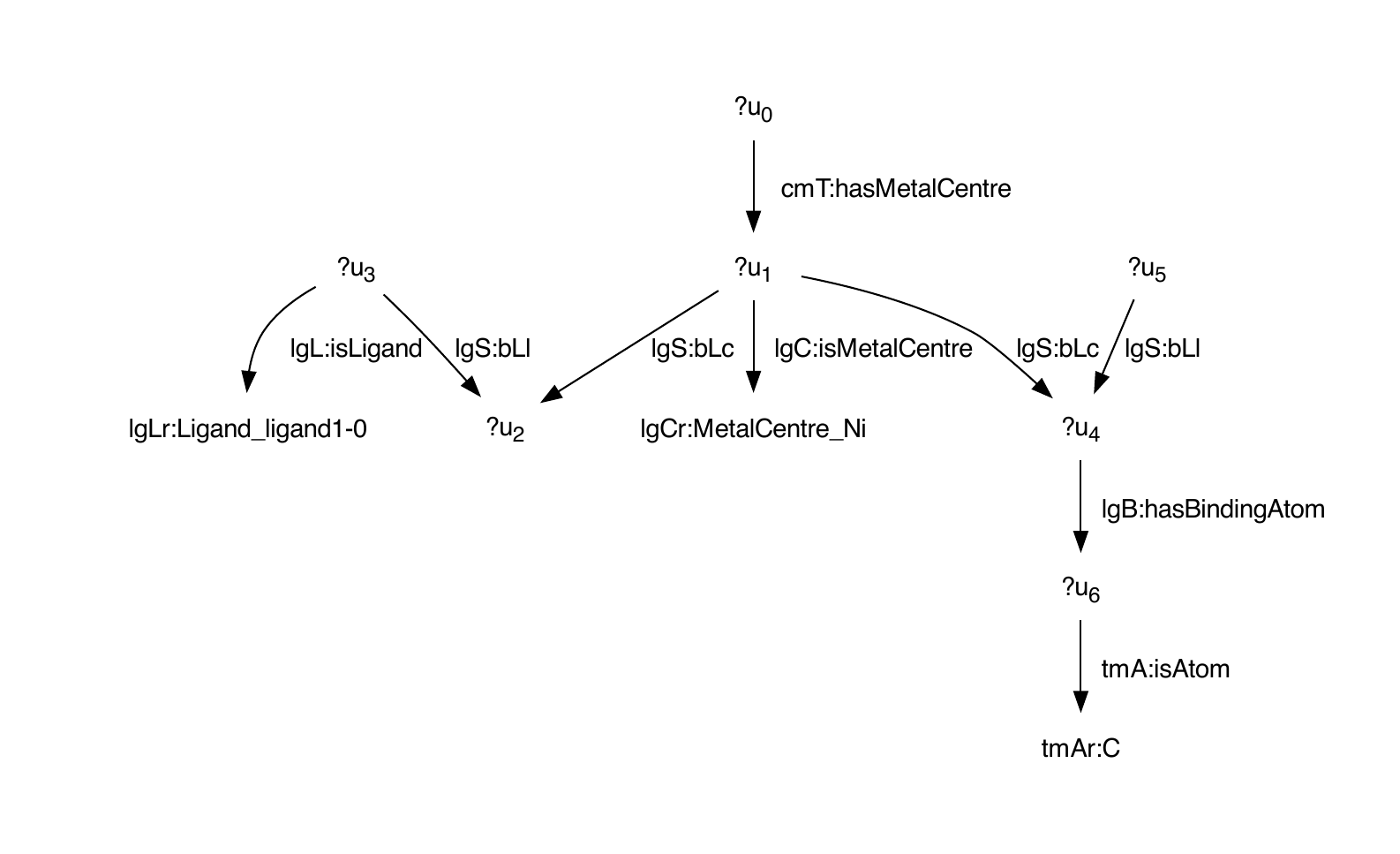}
        \caption{}\figlabel{ex_graph_pattern_2}
    \end{subfigure}
    \end{tabular}
    }
    \end{center}
    
    \caption{Two examples of graph patterns. Nodes starting with \say{?} represent variables. (a) A graph pattern that will match any graph representing a TMC with a bidentate ligand, with N as one of its binding atoms, and Pd as its metal centre. (b) A graph pattern that will match any graph representing a TMC that contains one copy of ligand1-0, another ligand, which has C as one of its binding atoms, and Ni as the metal centre. Notice that, by the \nrasem{}, these two ligands have to be distinct.}
    \figlabel{example_graph_patterns}
\end{figure}

\subsubsection{Frequent pattern mining and data exploration}\apxlabel{experiments_pattern_mining}
Tasks like that of identifying frequent substructures in a dataset of graphs that can be represented using the RDF syntax, as is the case in the proposed experiment, can be naturally addressed by identifying frequently occurring patterns (via pattern mining).

Here, the term \emph{pattern mining} refers to the computation of a set $\widetilde{\mathcal{P}}$ of \emph{frequent patterns} starting from a training set $\mathcal{G}_\mathrm{train}$ of RDF graphs. A pattern $p$ is considered frequent when
\begin{equation}\eqnlabel{frequent_pattern_definition}
 |\{G \in \mathcal{G}_\mathrm{train}: \Omega_{G,\,p} \neq \emptyset\}| \geq \alpha,
\end{equation}
for a given user\--defined threshold $\alpha \in \N$.\\

The pattern mining procedure that we employ in order to compute $\widetilde{\mathcal{P}}$ is hierarchical in nature. Given an initial \emph{seed pattern} $p_0$, i.e., a pattern made of a single triple, the algorithm produces candidate patterns by extending $p_0$ by one triple at a time and then by evaluating whether these candidates are frequent or not. Subsequent iterations proceed in the same way, by extending each of the patterns found at the previous step using only one triple and by looking for frequent patterns in the resulting set of candidates. The algorithm stops when a maximum number of triples per pattern is reached or no pattern can be extended into a frequent pattern.

The procedure is described at a technical level in \apxref{pattern_mining}.\\

Notice that pattern mining is combinatoric in nature, hence it should be expected that the cardinality of $\widetilde{\mathcal{P}}$ will be high.

\subsubsection{Filtering by relevance}\apxlabel{experiments_filtering}
In order to reduce the risk of employing patterns which are scarcely relevant with respect to the task being considered, we compute a subset $\mathcal{P} \subseteq \widetilde{\mathcal{P}}$ of patterns which are both frequent and encode information that we deem to be potentially useful.

In particular, since the experiment focuses on ligands and their binding atoms, we ignore out patterns in which at least one binding atom is specified. This means that a pattern $p \in \widetilde{\mathcal{P}}$ is considered to be interesting if the existence of at least one binding atom is asserted and the identity of that atom is specified. To avoid redundancy, we also impose that the chemical identity of the ligand the atom bleongs to must not be specified in the same pattern.

This condition alone, however, may be too restrictive, as more unusual TMCs may present configurations which are too rare to be encoded by a pattern which is both frequent and adhering to the rule above. For this reason, we also include in $\mathcal{P}$ patterns which are, in some sense, \say{precursors} of interesting patterns.

The most intuitive way of defining a precursor of a pattern $p \in \widetilde{\mathcal{P}}$, especially in view of the pattern mining algorithm in \apxref{experiments_pattern_mining}, is to consider a pattern $q \in \widetilde{\mathcal{P}}$ which is a \emph{subpattern} of $p$, i.e., $q \subseteq p$. In this simple case, it immediately holds that any match $\mu_p$ of $p$ against any graph $G$ can be restricted to a match $\mu_q$ of $q$ against that same graph. In particular, if $p$ matches a graph, then $q$ \emph{must} match that graph as well. Equivalently, if $q$ does not match a graph, $p$ \emph{cannot} match that graph.

As it turns out, this last property is more general than the strict subpattern relation. We say that $q$ \emph{dominates} $p$ if, for any RDF graph $G$, $|\Omega_{G,q}| = 0$ implies $|\Omega_{G,p}| = 0$. This relation may appear significantly more complicated than the subpattern relation, but, as shown in \apxref{domination}, it is sufficient to be able to compute pattern matches against an RDF graph to be able to check for pattern domination.

Somewhat informally, we could then define the set $\mathcal{P}$ as the set of frequent patterns which are either interesting, as per the rule given above, or dominate an interesting pattern. Precise details about the computation of $\mathcal{P}$ are given in \apxref{pattern_filtering}.

\subsection{Elementary graph pattern-based features}\apxlabel{experiments_naive_features}
Given a set $\mathcal{P}$ of frequent patterns, it is possible to compute a feature matrix $\widetilde{X} \in \N^{|\mathcal{G}| \times |\mathcal{P}|}$ such that 
 \[\widetilde{X}_{G,\,p} = |\Omega_{G,\, p}|,\]
where the notation $\widetilde{X}_{G,\,p}$ implicitly assumes that the entries of $\widetilde{X}$ are indexed by the elements of $\mathcal{G}$ and $\mathcal{P}$.

Due to the combinatorial explosion of the number of possible groundings of the variables in a pattern, it may happen that the columns of $\widetilde{X}$ possess largely different scales (see \figref{matches_sd}). For this reason, the actual elementary feature matrix $X$ is defined as a binarised version of $\widetilde{X}$, where 
\begin{equation}
X_{G,\,p} =
\begin{cases}
1&\qquad\text{if }\widetilde{X}_{G,\,p} > 0\\
0&\qquad\text{otherwise}
\end{cases}.
\end{equation}

\begin{figure}
    \centering
    \includegraphics[width=0.8\linewidth]{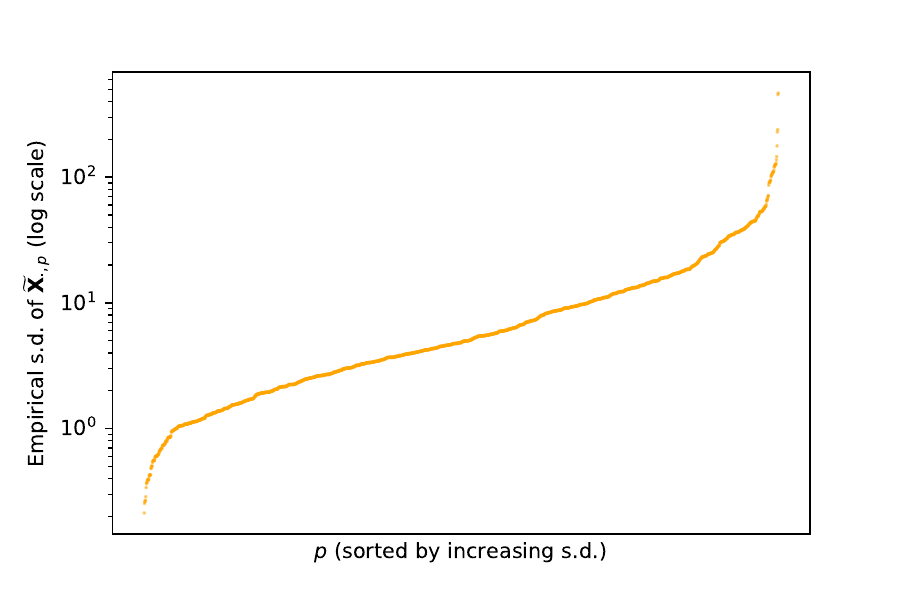}
    \caption{The empirical standard deviation (in log scale) of the columns of the elementary feature matrix $\widetilde{X}$ computed from the \emph{earlyTM} dataset selection. The patterns are sorted according to the increasing value of the empirical standard deviation of the corresponding column of $\widetilde{X}$.}
    \figlabel{matches_sd}
\end{figure}

\subsection{Clustering into families of substructures}\apxlabel{experiments_clustering}
A graph pattern is forced to encode exactly one type of structure, up to the uncertainty eventually granted by the variables, and therefore any \say{naive} pattern mining procedure (as per \apxref{experiments_pattern_mining}) will likely capture structures which are way more specific than the desired fundamental behaviours. For example, in \figref{ex_graph_pattern_1}, it may be that specifying that N participates in a ligand-Pd bond is just a form of \say{noise}, whereas the truly relevant information is the bidentate nature of the ligand. One way to address the issue is to employ a clustering algorithm to identify families of patterns that are similar according to some appropriate criterion. In this way, it is possible to capture the most meaningful substructure types, while also effectively reducing the dimensionality of the feature matrix $X$ introduced above.\\

The clustering method employed in this work falls within the category of hierarchical agglomerative clustering \citep{legendre2012numeco}. These methods proceed by progressively agglomerating items together, until all items have been merged into a single cluster or another stopping criterion is reached. Agglomeration is performed, at each step, by merging the two clusters that are closest to each other, according to a prespecified metric. Usually, a similarity (or dissimilarity) metric $s$ between pairs of objects is employed, together with a so called \emph{linkage criterion}, that specifies how to extend $s$ to groups of objects. The chosen criterion, here, is the \emph{average linkage} criterion, or \emph{UPGMA (Unweighted Pair-Group Method using Arithmetic averages)}. This amounts to
\[s(A, B) = \frac{1}{|A|\cdot|B|}\sum_{\substack{p \in A\\q \in B}}s(p, q),\]
for each pair of clusters $A$ and $B$. The linkage criterion also allows us to define a convenient stopping criterion. By specifying a threshold $\delta$, we can terminate the procedure when all the pairwise similarities between clusters fall beneath $\delta$. We determine the optimal value of $\delta$ among a set of candidates $\Delta$ by maximising the Silhouette coefficient \citep{lai2025silhouette} of the resulting clustering, under the constraint that there must be at least $M^{\mathcal{C}}_\mathrm{min}$ clusters\footnote{This is enforced in order to reduce the possibility of mistakenly aggregating together patterns which represent different chemical modalities.}.\\

As for the similarity metric $s$, we consider $2$ candidate families of metrics and, for each family, we explore two possible configurations, giving in total $4$ possible ways of computing $s$. In \sctref{experiments} we only considered one of these configurations, based on considerations that we make explicit in \apxref{extensive_discussion} For reasons of readability, we introduce here a brief description of each configuration, leaving the formal definitions for \apxref{similarity_metrics}.
\begin{enumerate}
\item Cosine similarity: given two vectors in a $d$\--dimensional space, the cosine similarity is defined as the cosine of the angle between them. In order to use this similarity in the context of graph clustering, it is necessary to define a vector representation of a graph pattern $p$:
\begin{enumerate}
    \item proxy vectors ($s_{\mathrm{cos};\,p}$): under the assumption that similar patterns will match a similar set of graphs, the similarity between two patterns $p$ and $q$ can be computed as the cosine similarity between the columns $\bm{X}_{\cdot, p}$ and $\bm{X}_{\cdot, q}$ of the feature matrix $X$ introduced in \apxref{experiments_naive_features};
    \item semantically-informed feature vectors ($s_{\mathrm{cos};\,s}$): as DLGs, graph patterns can also be interpreted as sets of terms, other than sets of triples, hence it is possible to assign a weight to each term using the term frequency inverse document frequency \citep[tf-idf;][]{salton1988termweighting} scheme, allowing for the definition of a feature vector where each entry is indexed by a term and it evaluates to its tf-idf weight.
\end{enumerate}
\item DLG similarity by \citet{champin2003graphsimilarity}: this metric is closely related to the graph edit distance \citep{sorlin2005reactive} and it uses multivalent mappings between the nodes of two graphs to define an \say{overlap} or \say{intersection} of the graphs. If a weighting scheme for the terms in the graphs is available, a candidate similarity value can be obtained as the ratio between the total weight in the overlap and the total weight in the two graphs. The final similarity is computed as the highest candidate similarity value achievable in this way. The choice of the weighting scheme greatly affects the results:
\begin{enumerate}
    \item naive weights ($s_{\mathrm{DLG};\,n}$): a constant weight $\omega_1$ is assigned to labels which specify a precise chemical identity (e.g., \texttt{tmAr:C}, \texttt{tmAr:Pt}, ...) and a constant weight $\omega_0 < \omega_1$ is assigned to non-specific labels (e.g. \texttt{tmA:Atom}, \texttt{lgC:MetalCentre}, ...);
    \item learned weights ($s_{\mathrm{DLG};\,l}$): a more advanced weighting scheme inferred from the semantic content of $\mathcal{P}$ using the principles of the tf-idf scheme (see \apxref{similarity_metrics}).
\end{enumerate}
\end{enumerate}

The behaviour of these $4$ similarity metrics is investigated via a pairplot reported in \figref{sim_pairs}

\begin{figure}
    \centering
    \includegraphics[width=\linewidth]{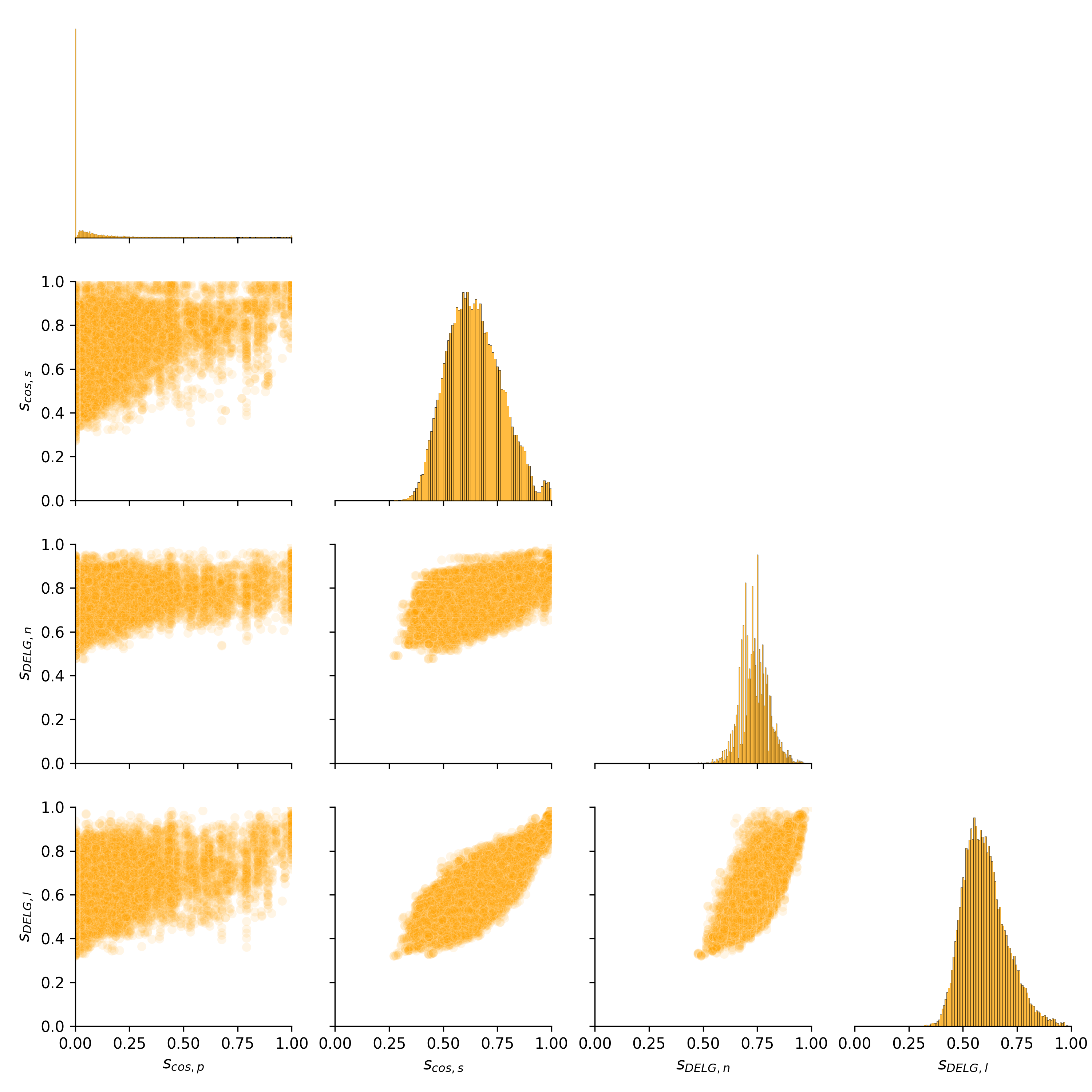}
    \caption{A pairplot of all the pairwise similarities among a random population of $300$ patterns, sampled from the frequent patterns mined from the \emph{earlyTM} dataset selection, for all the possible configurations of $s$. The plots on the diagonal of the grid are histograms of the similarity values computed for a given configuration, whereas the remaining plots are correlation plots.}
    \figlabel{sim_pairs}
\end{figure}

\subsection{Aggregated graph features}\apxlabel{experiments_graph_features}
Once a clustering $\mathcal{C}$ of the given set of patterns is obtained, it is possible to reduce the dimensionality of the feature matrix $X$ and define a new set of features. If $\mathcal{C} = \{C_1, \dots, C_P\}$, with $C_i \cap C_j = \emptyset$ for $i \neq j$, we can obtain a new feature matrix $Y \in \{0,1\}^{|\mathcal{G}|\times |\mathcal{C}|}$ by specifying an aggregation function $a$ and then computing
\[Y_{G, C_i} = a(\bm{X}_{G, C_i}, C_i),\]
where $\bm{X}_{G,C_i} = (X_{G,p})_{p \in C_i}$ (i.e., it is a view of the row vector $\bm{X}_{G,\cdot}$). In this work we consider two different aggregation functions:
\begin{enumerate}
    \item $a_\mathrm{max}(\bm{X}_{G,C_i}, C_i):= \max\{X_{G,p} \mid p\ \in C_i\}$, which implies that
\begin{equation}\eqnlabel{variable_definition}
    Y_{G,C_i} = \begin{cases}
        1&\qquad\text{if }\exists\, p \in C_i \text{ such that } \Omega_{G,p} \neq \emptyset\\
        0&\qquad\text{otherwise}
    \end{cases};
\end{equation}
    \item $a_\mathrm{median}(\bm{X}_{G,C_i}, C_i):= X_{G,p_{(C_i)}}$, where $p_{(C_i)} \in C_i$ is the \emph{set median graph} of $C_i$ \citep{jiang2001median}, i.e.
    \begin{equation*}
        p_{(C_i)} = \argmax_{p \in C_i}\sum_{p_j \in C_i}s(p, p_j).
    \end{equation*}
    In practice, this aggregation function has the effect of \say{collapsing} each cluster $C_i \in \mathcal{C}$ onto its median $p_{(C_i)}$.
\end{enumerate}

Notice that each similarity metric described in \apxref{experiments_clustering} gives rise to a different clustering and, for each clustering, each aggregation function gives rise to different features. This means that, overall, there will be $8$ versions of the feature matrix $Y$. As with the similarity metrics previously introduced, only one aggregation function is eventually presented in \sctref{experiments}. See \apxref{extensive_discussion} for further details.

\subsection{Probability-based TMC scoring model}\apxlabel{experiments_model}
The rows of the feature matrix $Y$ introduced in \apxref{experiments_graph_features} can be interpreted as the observations of a numeric random vector $\bm{Y}_G$, derived as a deterministic function of a random TMC $G$, given the set of frequent patterns $\mathcal{P}$ and its clustering $\mathcal{C}$. In terms of defining a scoring function for $G$, this formulation becomes useful as it allows us to define $f(G) := \log{\mathbb{P}(\bm{Y}_G\mid \mathcal{P}, \mathcal{C})}$. The distribution of $\bm{Y}_G$ can therefore be interpreted as a way to assess the score of a TMC in the context of a TMC training population.\\

The combinatoric nature of pattern mining and the consequently elevated cardinality of $\mathcal{P}$ make $\bm{Y}_G$ a high-dimensional vector, even after clustering. This poses a serious challenge to the statistical learning problem at hand \citep{johnstone2009highdimensions} and requires appropriate considerations. In this work we address the problem by modelling $\mathbb{P}(\bm{Y}_G\mid \mathcal{P}, \mathcal{C})$ using a \emph{Bayesian Network} \citep[BN;][]{koller2009probabilistic, kitson2023bnsurvey}, a framework specifically designed to handle high-dimensional data.\\

In the BN model, the distribution of $\bm{Y}_G$ factorises according to a Directed Acyclic Graph $B$, whose nodes are the very clusters of $\mathcal{C}$ who also index the entries of the random vector. This factorisation relies on the concept of \emph{parents} of a node $C \in \mathcal{C}$, i.e., the set of nodes that posses an outgoing edge pointing directly to $C$. This is denoted as $\text{pa}(C) = \{C_j \in \mathcal{C} \mid C_j \to C \text{ in } B\}$. Intuitively, the parents are those variables that have a \emph{direct} influence on $C$ \citep{koller2009probabilistic}. Building on this idea, a BN factorises the joint distribution of $\bm{Y}_G$ using the conditional distributions of $Y_{C_i} \mid \bm{Y}_{\text{pa}(C_i)}, \mathcal{P},\mathcal{C}$, with $\bm{Y}_{\text{pa}(C_i)} = (Y_{C_j})_{C_j \in \text{pa}(C_i)}$, i.e., 
\begin{equation}\eqnlabel{bn_fact}
    \mathbb{P}(\bm{Y}_G \mid \mathcal{P},\mathcal{C}) = \prod_{i = 1}^{|\mathcal{C}|}\mathbb{P}(Y_{C_i} | \bm{Y}_{\text{pa}(C_i)}, \mathcal{P},\mathcal{C}).
\end{equation}
Since $\bm{Y}_G$ is a binary feature vector, the conditional distributions in \eqnref{bn_fact} can be modelled using Bernoulli distributions, namely
\begin{align}\eqnlabel{bn_dist}
\begin{split}
    &Y_{C_i} \mid \bm{Y}_{\text{pa}(C_i)}, \mathcal{P},\mathcal{C} \sim \text{Bernoulli}(\theta_{\bm{Y}_{\text{pa}(C_i)}}),\\
    &\mathbb{P}(Y_{C_i} = 1| \bm{Y}_{\text{pa}(C_i)}, \mathcal{P},\mathcal{C}) = \theta_{\bm{Y}_{\text{pa}(C_i)}}.
\end{split}
\end{align}
Formula \eqnref{bn_fact} provides a series of benefits, the most evident of which is the factorisation of $\mathbb{P}(\bm{Y}_G \mid \mathcal{P},\mathcal{C})$ into the product of simpler univariate distributions. Additionally, it is possible to naturally derive a set  of conditional independence statements about $\bm{Y}_G$, specifically
\[Y_{C_i} \indep \bm{Y}_{\text{nd}(C_i)} \mid \bm{Y}_{\text{pa}(C_i)},\]
where $\text{nd}(C_i)$ is the set of the \emph{non-descendents} of $C_i$, i.e., the set of all the clusters $C_j$ such that there does not exist a directed path from $C_i$ to $C_j$. These are directly inferrable from the topology of $B$, meaning that a BN also offers a way to graphically inspect and represent a probability distribution.\\

\subsubsection{Learning the DAG}\apxlabel{experiments_learning_dag}
The crucial point of training a BN is, of course, acquiring the graph $B$. There are several structure learning algorithm available \citep[see][for a comprehensive survey] {kitson2023bnsurvey} which can be used to obtain an estimate of $B$ from observational data, possibly starting from a preliminary subgraph $B_0$, if such prior knowledge is available.\\

In this work we chose to employ the Hill-Climbing (HC) algorithm from the \texttt{bnlearn} package \citep{scutari2010bnlearn} in \texttt{R}. This is a greedy algorithm that searches for the graph $B$ that optimises the \emph{Bayesian Information Criterion}
\[\mathrm{BIC}(B) = \sum_{G \in \mathcal{G}_\mathrm{train}}\ln{\hat{\mathbb{P}}(\bm{y}_{G} \mid \mathcal{P}, \mathcal{C})} - \frac{\ln{|\mathcal{G}_\mathrm{train}|}}{2}F,\]
where the first term is the maximised likelihood of the training data according to the model specified by \eqnref{bn_fact} and \eqnref{bn_dist} and the second is a penalty proportional to the number of parameters in the model $F = \sum_{C_i \in \mathcal{C}}2^{|\text{pa}(C_i)|}$ \citep{kitson2023bnsurvey}.
The optimisation is carried out in a sequence of steps, where each step corresponds to applying one operation, i.e. edge insertion, deletion or removal, on the current estimate of $B$. We use prior information, in the form of the subgraph $B_0$, to provide a starting point to HC. This subgraph is derived by means of structural considerations on the patterns in $\mathcal{P}$, with respect of the notion of pattern domination introduced in \apxref{experiments_filtering}.\\

By definition, if $p, q \in \mathcal{P}$ are such that $q$ dominates $p$, it holds that, for any RDF graph $G$, if $q$ does not match $G$, then $p$ cannot match $G$ either. In terms of the elementary graph-level feature matrix $X$ defined in \apxref{experiments_naive_features}, this reads as: $X_{G,q} = 0$ implies $X_{G,p} = 0$. It is then evident that such patterns $p$ and $q$ have a direct influence on each other, in a similar sense in which a node experiences a direct influence from its parents in a BN. The matter is however made more complicated by the fact that the variables of interest in our model are not patterns, but clusters of patterns. As this work is intended as a mere example of the potential of our proposed data representation, we make here the simplifying assumption that the existence of a pair $p$ and $q$ as above, with patterns belonging to different clusters, implies an influential effect of one cluster on the other. Deciding on the direction of the effect, i.e., which node is the parent, is another delicate point, which we also choose not to discuss in detail. Since HC has the ability to reverse edge orientations, we simply choose a random orientation and leave the matter undecided.

Formally, then, we obtain $B_0$ by first defining an undirected graph $\tilde{B}_0$, whose nodes are the clusters in $\mathcal{C}$, and we draw an edge between $C_i, C_j \in \mathcal{C}$ if and only if there exists $p \in C_i, q \in C_j$ such that $q$ dominates $p$. The directed graph $B_0$ is derived from $\tilde{B}_0$ by randomly orienting all edges in such a way that no cycles are formed.\\

As in \apxref{experiments_graph_features}, each version of the feature matrix $Y$, deriving from a different similarity metric-aggregation function pair, will produce a different BN, thus leading to $8$ possible networks. 

\subsubsection{Parameter estimation}\apxlabel{experiments_bn_parameters}
Estimating the parameters in \eqnref{bn_dist} is the final step in training the BN model. We employ a Bayesian approach, using the Dirichlet priors
\begin{align*}
    \bm{\theta}_{C_i} &= \{\theta_{\bm{Y}_{\text{pa}(C_i)}}\}_{\bm{Y}_{\text{pa}(C_i)} \in \{0,\,1\}^{|\text{pa}(C_i)|}},\\
    \bm{\theta}_{C_i} &\sim \text{Dirichlet}(\lambda_0\cdot\bm{1}_{2^{|\text{pa}(C_i)|}}),
\end{align*}
where $\lambda_0 >0$, $\bm{1}_d = (1, \dots, 1)^\top \in \R^d$ and $\bm{\theta}_{C_i}$ is the family of the parameters of the conditional distributions $Y_{C_i} \mid \bm{Y}_{\text{pa}(C_i)}, \mathcal{P},\mathcal{C}$ for all the possible parent states $\bm{Y}_{\text{pa}(C_i)} \in \{0,\,1\}^{|\text{pa}(C_i)|}$. The estimation is carried out using the \texttt{pgmpy} package in Python \citep{ankan2024pgmpy}.

\subsection{The complete training phase}\apxlabel{experiments_training}
With the concepts introduced above, it is now possible to completely describe the steps of the training phase of the proposed experiment. This phase is divided in six fundamental steps, listed below. The aim is the isolation the relevant structural information, via preprocessing, pattern mining and clustering, and the fitting of the BN model.
\begin{enumerate}
    \item Acquiring $\mathcal{G}$ from tmQM-RDF: given a TMC population of interest, the RDF subgraphs describing those TMCs are extracted from tmQM-RDF, by locating the nodes of the form \texttt{cmT:XXYYZZ}, where \texttt{XXYYZZ} is a CSD code, and by then following the outgoing edges. Since only structural information is of interest here, the extracted subgraphs will only include the predicates listed in \tblref{reduced_rdf_predicates}. See \apxref{selection} for details on the nature of the populations described by $\mathcal{G}$ and the partition in $\mathcal{G}_\mathrm{train}$ and $\mathcal{G}_\mathrm{test}$.
    \item Graph pattern mining: frequent patterns are mined from $\mathcal{G}_\mathrm{train}$, using the algorithm in \apxref{experiments_pattern_mining}.
    \item Pattern filtering: the mined patterns are filtered according to the principles described in \apxref{experiments_filtering}.
    \item Elementary feature computation: the elementary graph-level features described in \apxref{experiments_naive_features} are computed.
    \item Identification of structural families: frequent patterns are clustered together using the similarity-based clustering technique described in \apxref{experiments_clustering}. This is repeated for each possible configuration of the similarity metric, as explained in \apxref{experiments_clustering}.
    \item Aggregated features computation: the elementary features are aggregated according to the identified structural families. This is repeated for each possible configuration of similarity metric and aggregation function, as explained in \apxref{experiments_graph_features}.
    \item Statistical learning: the joint distribution of the graph features is estimated as a BN, accounting for statistical interactions (i.e., correlations) and structural relationships among the features (i.e., dominations).  This is repeated for each possible configuration of similarity metric and aggregation function, as explained in \apxref{experiments_learning_dag}.
\end{enumerate}

\begin{table}[tb]
    \centering
    \begin{tabular}{l l}
        \toprule
         Identified structural element & URI \\
         \midrule
         \midrule
         Metal centre & cmT:hasMetalCentre \\
                                     & lgC:isMetalCentre \\[0.25cm]
         Ligand & cmT:hasLigand \\
                                & lgS:bLc \\
                                & lgS:bLl \\
                                & lgL:isLigand \\[0.25cm]
         Binding atom & lgB:hasBindingAtom \\
                                     & tmA:isAtom\\
         \bottomrule
    \end{tabular}
    \caption{The URIs of the only predicates allowed in frequent patterns. The predicates are grouped according to the structural element they primarily identify via their domain/range.}\tbllabel{reduced_rdf_predicates}
\end{table}

\section{Technical details about frequent pattern mining}\apxlabel{pattern_mining}
This section describes the idea behind the algorithm\footnote{This algorithm has been developed and implemented by B.\,E. and will be thoroughly addressed in a separate publication.} used to mine frequent patterns, as per the definition given in \eqnref{frequent_pattern_definition}.\\

In what follows, let $\mathcal{G}_\mathrm{train}$ be a dataset of RDF graphs and let $\alpha \in \N$ be the threshold for pattern frequency. The ultimate goal is to mine a set $\widetilde{\mathcal{P}}$ of frequent graph patterns. The algorithm hereby described proceeds hierarchically, by progressively expanding patterns of a given size (i.e., number of triples) by one additional triple (i.e., by adding a pattern of size 1) at a time.

\subsection{Notation}
We briefly recall here that we define a DLG as a subset $G \subseteq \mathcal{T}^3$, where $\mathcal{T}$ is a (possibly infinite) set of terms, eventually with the restrictions imposed by the RDF syntax described in \apxref{methods_kg}. Using this notation, merging two DLGs $G_1, G_2$ is achieved simply via the usual set union operation $G_1 \cup G_2$, whereas the subgraph relation coincides with the subset relation $G_1 \subseteq G_2$. 

A graph pattern, on the other hand, is a DLG defined as a subset of $(\mathcal{T} \cup \mathcal{V})^3$, where $\mathcal{V}$ is an infinite set of symbols called variables. Given a pattern $p$, we denote with $\mathcal{T}_p$ and $\mathcal{V}_p$ the set of terms and variables, respectively, that appear in $p$. Given a function $m:\mathcal{T} \cup \mathcal{V} \to \mathcal{T}$, we denote with $m(p)$ the DLG obtained by applying $m$ to every element of every triple of $p$. This work only considers patterns that admit variables either in subject or object position (hence never in predicate position). For simplicity, we will still write $p \subseteq (\mathcal{T} \cup \mathcal{V})^3$ and consider this constraint to be enforced implicitly.

Finally, given a DLG $G$ and a pattern $p$, a match of $p$ against $G$ under the \nrasem{} is defined as an injective function $\mu:\mathcal{T}_p \cup \mathcal{V}_p \to \mathcal{T}_G$ such that $\mu(t) = t$ for every $t \in \mathcal{T}_G$ and $\mu(p) \subseteq G$. The set of all such functions is denoted as $\Omega_{G,p}$.

\subsection{Preliminary definitions}\apxlabel{pattern_definitions}
Let $\widetilde{\mathcal{P}}_i$ be the set of frequent patterns of size $i$. In particular, patterns of $\widetilde{\mathcal{P}}_1$, i.e., patterns made of a single triple, are called \emph{triple patterns}.\\

Consider now two patterns $p_a \in \widetilde{\mathcal{P}}_i$, for some $i \geq 1$, and $p_b \in \widetilde{\mathcal{P}}_1$. The set of the \emph{possible extensions of $p_a$ given $p_b$}, denoted as $\mathcal{E}(p_a\mid p_b)$ is defined as the set of all the injective mappings $m: \mathcal{V}_{p_b} \to \mathcal{V}_{p_a}$ such that $m(p_b) \not\subseteq p_a$ and there exists at least one variable $v \in \mathcal{V}_{p_b}$ such that $m(v) \in \mathcal{V}_{p_a}$ (in other words, $\mathcal{E}(p_a\mid p_b)$ is the set of all the possible relabelling of the variables in $p_b$ that ensure that the transformed triple pattern is not a triple already present in $p_a$ and that the extended pattern $p_a \cup m(p_b)$ is connected). Two mappings that only differ in how they label elements that are not mapped into $\mathcal{V}_{p_a}$ are considered equivalent. 

Given two matches $\mu_a \in \Omega_{G,p_a}, \mu_b \in \Omega_{G, p_b}$, for some $G \in \mathcal{G}$, they are said to be \emph{compatible}, given a mapping $m \in \mathcal{E}(p_a\mid p_b)$ if the following conditions hold:
\begin{enumerate}
    \item $\mu_a(\mathcal{V}_{p_a}) \cap \mathcal{T}_{p_b} = \emptyset$;
    \item $\mu_b(\mathcal{V}_{p_b}) \cap \mathcal{T}_{p_a} = \emptyset$;
    \item for any $v_a \in \mathcal{V}_{p_a}, v_b \in \mathcal{V}_{p_b}$, $m(v_b) = v_a$, if and only if $\mu_a(v_a) = \mu_b(v_b)$.
\end{enumerate}
Given two compatible matches $\mu_a \in \Omega_{G,p_a}, \mu_b \in \Omega_{G, p_b}$, given $m \in \mathcal{E}(p_a\mid p_b)$, we define the \emph{extension of $\mu_a$ by $\mu_b$ via $m$} as the function $\mu_a \cup m(\mu_b) : \mathcal{T}_{p_a} \cup \mathcal{V}_{p_a} \cup \mathcal{T}_{p_b} \cup m(\mathcal{V}_{p_b}) \to \mathcal{T}$ given by
\begin{equation*}
    \left(\mu_a \cup m(\mu_b)\right)(x) := \begin{cases}
        x \qquad&\text{if } x \in \mathcal{T}_{p_a} \cup \mathcal{T}_{p_b}\\
        \mu_a(x)\qquad&\text{if } x \in \mathcal{V}_{p_a}\\
        \mu_b(m^{-1}(x))\qquad&\text{otherwise}
    \end{cases}.
\end{equation*}
The expression $m^{-1}(x)$ is a slight abuse of notation as we are implicitly considering the codomain of $m$ to be restricted to $m(\mathcal{V}_{p_b})$. Notice also how the compatibility of $\mu_a$ and $\mu_b$ ensures that $\mu_a \cup m(\mu_b)$ is well defined and that $\mu_a \cup m(\mu_b) \in \Omega_{G, p_a\cup m(p_b)}$.\\

Two patterns $p_a, p_b \in \widetilde{\mathcal{P}}_i$ are said to be \emph{isomorphic} if there exists a bijective function $m: \mathcal{V}_{p_a}\cup\mathcal{T}_{p_a} \to \mathcal{V}_{p_b}\cup\mathcal{T}_{p_b}$ such that $m(t) = t$ for each $t \in \mathcal{T}_{p_a}$ and $m(p_a) = p_b$.

\subsection{Frequent pattern mining algorithm}\apxlabel{experiments_pattern_mining_algo}
The outline of the pattern mining algorithm used in this work is described in this section. As anticipated, the basic idea is to progressively extend patterns of size $i$ by adding a single triple pattern. Extension of a pattern $p_a \in \widetilde{\mathcal{P}}_i$ with a triple pattern $p_b \in \widetilde{\mathcal{P}}_1$ is performed by considering the possible extensions $m \in \mathcal{E}(p_a \mid p_b)$ to form extended patterns $p_a \cup m(p_b)$. The matches of these extended patterns against graphs $G \in \mathcal{G}$ are immediately computed by considering all the compatible matches $\mu_a \in \Omega_{G, p_a}, \mu_b \in \Omega_{G,p_b}$ and by computing the extension $\mu_a \cup m(\mu_b)$. The mining procedure starts from a single triple pattern $p_0 \in (\mathcal{T} \cup \mathcal{V})^3$, also called a \emph{seed pattern}, which is extended to obtain larger graph patterns.

As the task can be particularly demanding from a computational point of view, a random sampling procedure is introduced in order to use only a fraction of the possible extensions of a pattern. Let then $(\pi_i)_{i\geq2}$ be a sequence of probabilities. When considering a pattern $p_a \in \widetilde{\mathcal{P}}_i$, to be extended with $p_b \in \widetilde{\mathcal{P}}_1$ to form a pattern of size $i+1$, each extension $m \in \mathcal{E}(p_a \mid p_b)$ has a probability $\pi_{i+1}$ of being \emph{discarded}.

In what follows, we then denote with $\hat{\mathcal{E}}(p_a \mid p_b)$ the set of accepted extension. Similarly, we denote with $\hat{\mathcal{P}}_i$ the set of \emph{mined} patterns of size $i$ that will be part of the output of the algorithm (these set are subsets of the collections $\widetilde{\mathcal{P}}_i$ of \emph{all} the frequent patterns in the dataset).\\

The full set of inputs of the algorithm is then the following: $\mathcal{G}_{train}$, a set of training DLGs from which frequent patterns are to be mined; $\alpha \in \N$, the threshold used to determine if a pattern if frequent; $p_0$, a seed triple pattern; and $(\pi_i)_{i\geq2}$, the sequence of rejection probabilities, a maximum pattern size $M_\mathrm{max}$.

\begin{denumerate}
    \item Compute the set $\widetilde{\mathcal{P}}_1$ of frequent triple patterns that will be used to extend the patterns:
    \begin{denumerate}
        \item Extract all the triples from all the graphs of $\mathcal{G}_{train}$ and replace either the subject or the object (or both) with a variable to obtain an initial set of triple patterns\footnote{As mentioned, we do not consider here triple patterns that have a variable in predicate position. Similarly, in this phase it is also possible to implement other forms of constraints regarding variable placement or predicate blacklisting.} (for each triple pattern, keep track of the graph $G$ from which it was extracted and the original triple).
        \item For each triple pattern $p$ so obtained, populate the sets $\Omega_{G,p}$ with the matches that replace the new variables with the original terms before the substitution.
        \item Retain only the triple patterns that match at least $\alpha$ different graphs and use them to populate $\widetilde{\mathcal{P}}_1$.
    \end{denumerate}
    \item Initialise $\hat{\mathcal{P}}_1 \leftarrow \{p_0\}$.
    \item Initialise $i \leftarrow 1$.
    \item Repeat the following as long as $\hat{\mathcal{P}}_i \neq \emptyset$ and $i \leq M_\mathrm{max}$:
    \begin{denumerate}
        \item Initialise $\hat{\mathcal{P}}_{i + 1} \leftarrow \emptyset$.
        \item For each $p_a \in \hat{\mathcal{P}}_{i},p_b \in \widetilde{\mathcal{P}}_1$:
        \begin{denumerate}
            \item Compute $\hat{\mathcal{E}}(p_a \mid p_b)$:
            \begin{denumerate}
                \item Initialise $\hat{\mathcal{E}}(p_a \mid p_b) \leftarrow \emptyset$
                \item For each $\mathcal{V}_{p_a}' \subseteq \mathcal{V}_{p_a}, \mathcal{V}_{p_b}' \subseteq \mathcal{V}_{p_b}$, with $|\mathcal{V}_{p_a}'| = |\mathcal{V}_{p_b}'| > 0$:
                \begin{denumerate}
                    \item For each injective mapping $m_0:\mathcal{V}_{p_b}' \to \mathcal{V}_{p_a}'$:
                    \begin{denumerate}
                        \item Extend $m_0$ to $m: \mathcal{V}_{p_b} \to \mathcal{V}$ by assigning to each variable of $\mathcal{V}_{p_b} \setminus \mathcal{V}_{p_b}'$ a symbol from $\mathcal{V} \setminus \mathcal{V}_{p_a}$.
                        \item If $m(p_b) \not\subseteq p_a$, with probability $1 - \pi_{i+1}$, update $\hat{\mathcal{E}}(p_a \mid p_b) \leftarrow \hat{\mathcal{E}}(p_a \mid p_b) \cup \{m\}$.
                    \end{denumerate}
                \end{denumerate}
            \end{denumerate}
            \item For each $m \in \hat{\mathcal{E}}(p_a \mid p_b)$:
            \begin{denumerate}
                \item Let $p_a' = p_a \cup m(p_b)$
                \item For each $G \in \mathcal{G}_\mathrm{train}$:
                \begin{denumerate}
                    \item For each $\mu_a \in \Omega_{G,p_a}, \mu_b \in \Omega_{G, p_b}$ such that $\mu_a$ and $\mu_b$ are compatible, update $\Omega_{G,p_a'} \leftarrow \Omega_{G,p_a'} \cup \{\mu_a \cup m(\mu_b)\}$
                \end{denumerate}
                \item If $p_a'$ is frequent and not isomorphic to any pattern in $\hat{\mathcal{P}}_{i+1}$, update $\hat{\mathcal{P}}_{i+1} \leftarrow \hat{\mathcal{P}}_{i+1} \cup \{p_a'\}$.
            \end{denumerate}
        \end{denumerate}
        \item Update $i \leftarrow i + 1$.
    \end{denumerate}
    \item Return $\hat{\mathcal{P}}_{1} \cup \hat{\mathcal{P}}_{2} \cup \dots \cup \hat{\mathcal{P}}_{M_\mathrm{max}}.$
\end{denumerate}

\section{Isolating interesting patterns}\apxlabel{pattern_filtering}
Let $\widetilde{\mathcal{P}}$ be the set of frequent patterns produced by the algorithm in \apxref{pattern_mining}. We recall that, as per \apxref{experiments_filtering}, we consider a pattern $p \in \widetilde{\mathcal{P}}$ to be interesting if it specifies the existence and the identity of at least one binding atom, in a non redundant way.\footnote{We recall that a specified binding atom is considered non-redundant if the pattern does not also specify the chemical identity of the ligand the atom belongs to.} Moreover, we say that a pattern $q \in \widetilde{\mathcal{P}}$ dominates a pattern $p$ if, for any RDF graph $G$, $|\Omega_{G,q}| = 0$ implies $|\Omega_{G,p}| = 0$.

Our aim is then that of computing a subset $\mathcal{P} \subseteq \widetilde{\mathcal{P}}$ made of patterns that are either interesting, or dominate an interesting pattern. Using the notation introduced in \apxref{pattern_mining}, we consider the decomposition $\widetilde{\mathcal{P}} = \hat{\mathcal{P}}_1 \cup \dots \cup \hat{\mathcal{P}}_{M_\mathrm{max}}$. For simplicity, we construct $\mathcal{P}$ using only patterns from $\hat{\mathcal{P}}_{M_\mathrm{min}} \cup \dots \cup \hat{\mathcal{P}}_{M_\mathrm{max}}$, where $M_\mathrm{min} \in \N$ is a user-defined minimum pattern size.\\

The construction of $\mathcal{P}$ proceeds as follows.
\begin{enumerate}
    \item Initialise $\mathcal{P} \leftarrow \emptyset$.
    \item Initialise $\mathcal{R} \leftarrow \mathcal{P}$
    \item Initialise $i \leftarrow M_\mathrm{max}$.
    \item Repeat the following as long as $i \geq M_\mathrm{min}$:
    \begin{enumerate}
        \item Update $\mathcal{P} \leftarrow \mathcal{P} \cup \{p \in \hat{\mathcal{P}}_i\mid p\text{ is interesting}\}$.
        \item Update $\mathcal{P} \leftarrow \mathcal{P} \cup \{q \in \hat{\mathcal{P}}_i\mid \exists p \in \mathcal{R} : q \text{ dominates } p\}$.
        \item Update $\mathcal{R} \leftarrow \mathcal{P} \setminus \mathcal{R}$.
        \item Update $i \leftarrow i - 1$.
    \end{enumerate}
    \item Return $\mathcal{P}$.
\end{enumerate}

\section{Implementation of pattern matching}\apxlabel{pattern_matching}

In practice, the matches of a pattern $p$ against a graph $G$ can be easily computed using the SPARQL query language \citep{harris2013sparql}.

As SPARQL is a graph pattern matching-based query language \citep{harris2013sparql}, a graph pattern $p$ as introduced in \apxref{experiments_pattern_definition} can be naturally transformed into a query. In particular, the set of triples of $p$ will form the body of \texttt{WHERE} clause of the query. The \nrasem{} is also easily implemented by including a \texttt{FILTER} instruction that prevents any variable to match terms already present in the pattern and any two variables to match the same term.\\

Once a pattern is written into a SPARQL query, we employ the \texttt{R} package \texttt{virtuoso} \citep{boettiger2021virtuoso, rcoreteam2024r} as a convenient interface to SPARQL.

\section{On the domination relationship between graph patterns}\apxlabel{domination}
Given two patterns $p$ and $q$, we recall that we say that $q$ dominates $p$ if, for any RDF graph $G$, $|\Omega_{G,q}| = 0$ implies $|\Omega_{G,p}| = 0$. Equivalently, $|\Omega_{G,p}| > 0$ implies $|\Omega_{G,q}|>0$.

\begin{prop}\label{prop:domination}
    Let $p, q$ be two frequent patterns given a graph dataset $\mathcal{G}$. Then the following statements are equivalent:
    \begin{enumerate}
        \item $q$ dominates $p$.
        \item There exists an injective mapping $\mu:\mathcal{T}_q \cup \mathcal{V}_q \to \mathcal{T}_p \cup \mathcal{V}_p$ such that $\mu(t) = t$ for each term $t \in \mathcal{T}_q$ and $\mu(q)$ is a subgraph of $p$.
    \end{enumerate}
\end{prop}
\begin{proof}
    First, consider the case in which $q$ dominates $p$. Construct then an RDF graph $G$ by grounding $p$ using terms from $\mathcal{T} \setminus \mathcal{T}_q$, so that $\mathcal{T}_G = \mathcal{T}_p \cup \mathcal{T}_0,\ \mathcal{T}_q \cap \mathcal{T}_0 = \emptyset$. Trivially, there exists a mapping $\mu_p \in \Omega_{G,p}$ and, in particular, $p$ and $G$ are isomorphic via $\mu_p$. It also holds that
    \begin{align*}
        \forall t \in \mathcal{T}_p,&\ \mu_p^{-1}(t) = t,\\
        \forall t \in \mathcal{T}_0,&\ \mu_p^{-1}(t) \in \mathcal{V}_p.
    \end{align*}
    But now, since $q$ dominates $p$, there must also be a mapping $\mu_q \in \Omega_{G,q}$. By the construction of $G$, it must then be that $\mathcal{T}_q \subseteq \mathcal{T}_p$ and therefore
    \[\forall t \in \mathcal{T}_q,\ \mu_p^{-1}(\mu_q(t)) = \mu_q(t) = t.\]
    Moreover, since $\mu_q(q)$ is a subgraph of $G$, which is isomorphic to $p$ via $\mu_p$, it must follow that $\mu_p^{-1}(\mu_q(q))$ is a subgraph of $p$. Finally, since both $\mu_p$ and $\mu_q$ are injective, so is the mapping $\mu_p^{-1} \circ \mu_q$, which then satisfies all the required properties.\\

    Suppose now that there exists a mapping $\mu$ as described in statement (2) and let $G$ be a graph such that there exists $\mu_p \in \Omega_{G,p}$. It is clear that the mapping $\mu_p \circ \mu:\mathcal{T}_q \cup \mathcal{V}_q \to \mathcal{T}_G$ is injective and that $\mu_p(\mu(q))$ is a subgraph of $G$. Moreover, 
    \[\forall t \in \mathcal{T}_q,\ \mu_p(\mu(t)) = \mu_p(t) = t,\]
    where the last equality is based on the fact that, by the properties of $\mu$, it necessarily follows that $\mathcal{T}_q \subseteq \mathcal{T}_p$. We can then conclude that $\mu_p \circ \mu \in \Omega_{G,q}$.
\end{proof}

This proposition offers a convenient way of checking whether two patterns $p$ and $q$ are related via domination. In fact, condition (2) can be equivalently written as follows:
\begin{enumerate}
    \setcounter{enumi}{2}
    \item Let $G_p$ be an RDF graph built by grounding $p$ using terms from a set $\mathcal{T}_0,\ \mathcal{T}_0 \cap \left(\mathcal{T}_p \cup \mathcal{T}_q\right) = \emptyset$. Then $|\Omega_{G_p,q}| > 0$.
\end{enumerate}
It follows, then, that pattern domination is equivalent to pattern matching. We provide a concrete example of this principle in Example~\ref{ex:domination}

\begin{figure}[!t]
    \begin{center}
    \makebox[\textwidth][c]{%
    \begin{tabular}{c c c}
    \begin{subfigure}[t]{0.48\linewidth}
        \includegraphics[width=\linewidth]{pattern1.pdf}
        \caption{}\figlabel{example_pattern_1}
    \end{subfigure} &
    \begin{subfigure}[t]{0.48\linewidth}
        \includegraphics[width=\linewidth]{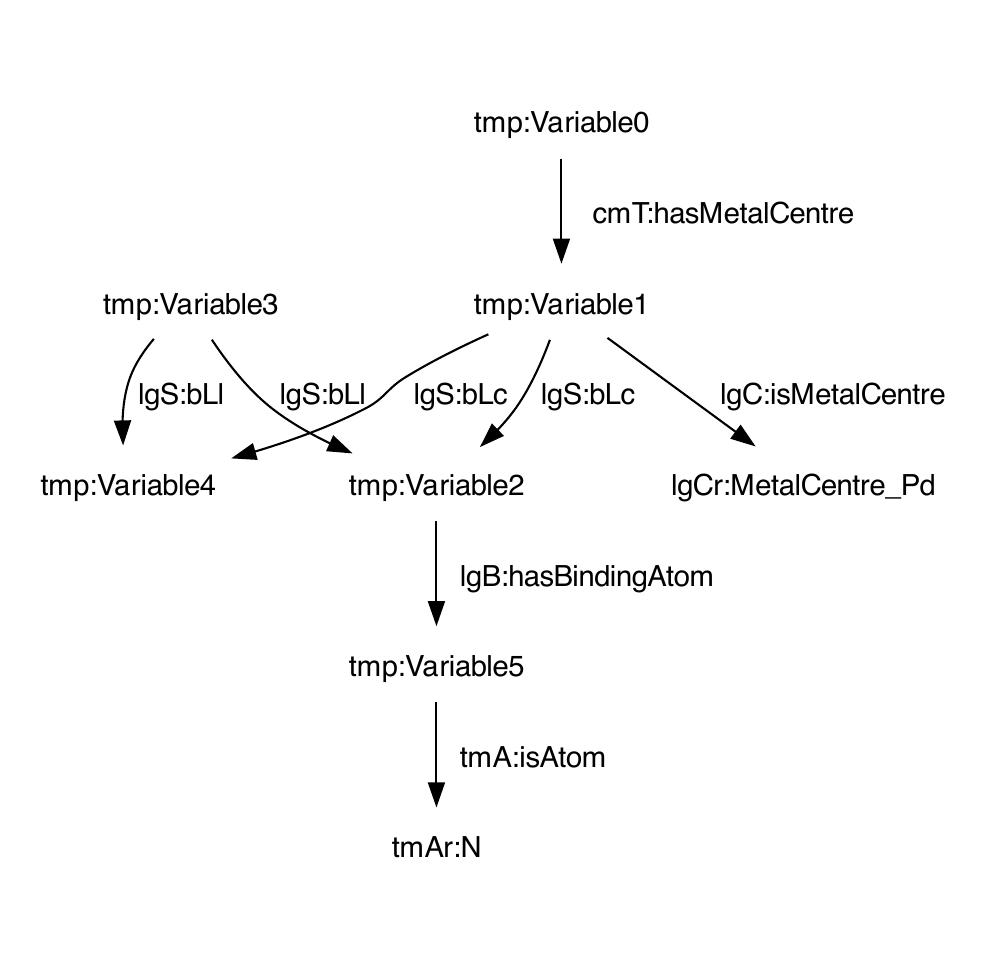}
        \caption{}\figlabel{example_pattern_grounded}
    \end{subfigure} 
    \end{tabular}
    }
    \end{center}
    
    \caption{(a) An example of a graph pattern. (b) An example of a graph matched by the graph pattern, obtained by grounding the variables in the pattern itself.}
    \figlabel{example_pattern}
\end{figure}

\begin{example}\label{ex:domination}
    Let $p$ be the pattern shown in \figref{example_pattern_1}. This pattern will match a TMC $G$ if and only if this TMC has Pd as its metal centre and a bidentate ligand which binds to the centre with at least one N atom. Consider now the triple pattern $q = \left\{(?w0,\, \texttt{tmA:isAtom,}\, ?w1)\right\}.$

    If $\mu_p$ is a match of $p$ against $G$, then a match $\mu_q$ of $q$ against $G$ is immediately obtained as
    \begin{align*}
        w_0 &\xmapsto{\mu_q} \mu_p(v_5)\\
        w_1 &\xmapsto{\phantom{\mu_q}} \texttt{tmAr:N}.
    \end{align*}
    Since $p$ matches $G$, it must be that $(\mu_p(v_5),\, \texttt{tmA:isAtom},\, \texttt{tmAr:N}) \in G$, which implies that $\mu_q(q) \subseteq G$. It follows then that $q$ dominates $p$. This can be equivalently written in terms of the composition $\mu_p^{-1}\circ\mu_q$, as
    \begin{align*}
        w_0 &\xmapsto{\mu_p^{-1}\circ\mu_q} v_5\\
        w_1 &\xmapsto{\phantom{\mu_p^{-1}\circ\mu_q}} \texttt{tmAr:N},
    \end{align*}
    which provides us with the mapping $\mu$ whose existence is assured by Proposition~\ref{prop:domination}.\\

    Suppose now that, instead of the considerations above, we approach the situation by considering a graph $G$ obtained by grounding $p$, for example the graph shown in \figref{example_pattern_grounded}. Then, a match $\mu_q'$ of $q$ against this $G$ is simply given by
    \begin{align*}
        w_0 &\xmapsto{\mu_q'} \texttt{tmp:Variable5}\\
        w_1 &\xmapsto{\phantom{\mu_q'}} \texttt{tmAr:N}.
    \end{align*}
    But an equally immediate match $\mu_p'$ of $p$ is now given by
    \begin{align*}
        v_i &\xmapsto{\mu_p'} \texttt{tmp:Variable}i \qquad i=0,\dots,5
    \end{align*}
    and this match is, in particular, an invertible function when its codomain is restricted to $\{\texttt{tmp:Variable}i\mid i = 0,\dots,5\}.$ We can then define the mapping $\mu := (\mu_p')^{-1}\circ\mu_q$ which has the following action:
    \begin{align*}
        w_0 \xmapsto{\mu_q'} &\texttt{tmp:Variable5} \xmapsto{(\mu_p')^{-1}} v_5\\
        \texttt{?w1} \xmapsto{\phantom{\mu_q'}} &\texttt{tmAr:N} \xmapsto{\phantom{(\mu_p')^{-1}}} \texttt{tmAr:N}.
    \end{align*}
    By Proposition \ref{prop:domination}, the existence of this mapping ensures that $q$ dominates $p$. We have then been able to verify the validity of the domination relationship simply by virtue of our capability of computing pattern matches.
\end{example}

\section{Similarity metrics}\apxlabel{similarity_metrics}
In order to properly discuss the metrics introduced in \apxref{experiments_clustering}, it is useful to introduce a more rigorous definition of a DLG. For the remainder of this section, a DLG will be intended to be a triple $G = (V, E, \pi)$, where $V$ is the set of nodes, $E \subseteq V \times V$ is the set of (directed) edges and $\pi:V \cup E \to \mathcal{T}$ is a labelling function. The notation $\pi(v) = t$, for $v \in V$, means that the node $v$ is labelled with the term $t \in \mathcal{T}$, whereas $\pi(v_s, v_o) = t_p$ signifies that the edge which has tail with label $\pi(v_s)$ (the subject) and head with label $\pi(v_o)$ (the object) is given the label $t_p$ (the predicate). It is straightforward to move from the definition given in \sctref{methods} to the one just stated. If $G_\mathrm{set} \subseteq \mathcal{T}^3$ is a set of triples, we can obtain an equivalent graph $G$ by defining $V = \{t \in \mathcal{T} \mid \exists t_1, t_2 \in \mathcal{T}: (t,\, t_1,\, t_2) \in G_\mathrm{set}\ \vee\ (t_1,\, t_2,\, t) \in G_\mathrm{set}\}$, $E = \{(t_s,\, t_o) \in V \mid \exists t_p \in \mathcal{T}: (t_s,\, t_p,\, t_o) \in G_\mathrm{set}\}$ and $\pi$ such that $\pi(v) = v$ for every $v \in V \subseteq \mathcal{T}$ and $\pi(t_s, t_o) = t_p$ whenever $(t_s, t_p, t_o) \in G_{set}$\footnote{Remember that this work only considers sets of triples such that the corresponding graph has at most one edge between any two nodes and at most one node per term.}.

In cases in which there may be ambiguity as to which DLG we refer to, we shall specify that by denoting the node set, the edge set and the labelling function as $V_G, E_G$ and $\pi_G$ respectively.

When DLGs are expressed using the RDF language, the restrictions described in \apxref{methods_rdf} naturally apply.

\subsection{Cosine similarity}
The cosine similarity is a similarity metric defined over vectors of $\R^d$. Given $\bm{x}, \bm{y} \in \R^d$, their cosine similarity is
\[s_\mathrm{cos}(\bm{x},\, \bm{y}) = \frac{\bm{x}\cdot \bm{y}}{\|\bm{x}\|_2\|\bm{y}\|_2},\]
where $\|\cdot\|_2$ is the Euclidean norm and $\cdot$ is the usual dot product in $\R^d$.

Now, since the objects between which a similarity is to be computed are graph patterns (and therefore labelled graphs, not vectors), in order to use the cosine similarity it is necessary to introduce appropriate feature vectors.

\subsubsection{Proxy feature vectors}
The easiest way to address the similarity problem is probably to exploit the columns of the (binarised) feature matrix $X$ introduced in \apxref{experiments_naive_features} as proxy for the actual patterns, under the assumption that similar patterns will match a similar set of graphs. Using this approach, the similarity between the patterns $p$ and $q$, measured as the cosine similarity between the corresponding columns of $X$, is
\[s_{\mathrm{cos};\,p}(p, q) := s_\mathrm{cos}(\bm{X}_{\cdot, p}, \bm{X}_{\cdot, q}) = \frac{\bm{X}_{\cdot, p}\cdot \bm{X}_{\cdot, q}}{\|\bm{X}_{\cdot, p}\|_2\|\bm{X}_{\cdot, q}\|_2},\]
where $\|\cdot\|_2$ is the Euclidean norm and $\cdot$ is the usual dot product in $\R^{|\mathcal{G}|}$.

\subsubsection{Semantic feature vectors}
It is not hard to imagine why the feature vectors introduced above may be unsatisfactory. First of all, the very idea of relying on $X$ instead of directly employing the semantic information within the patterns means that we are not using all the available information. Secondly, there may be other reasons, apart from structural similarity, for which two patterns may share a similar matching profile. For instance, if two patterns encode two different structures that are strongly positively correlated, then the two patterns will tend to match a similar set of graphs, even though they may be remarkably different. This implies that statistically significant information may be masked within clusters.\\

Since the task is to identify families of patterns that encode similar information, an ideal metric should then exploit the semantic content of the patterns themselves, as this is the most complete source of knowledge about the information encoded therein. In this context, \say{semantic content} could be interpreted as the set of RDFS classes that is present within a pattern. Notice that an item can either receive an explicit class assignment via a predicate that is a subproperty of \texttt{rdf:type} (e.g. \texttt{tmA:isAtom}) or its class can be inferred from the TBox in \figref{tbox} and, if multiple classes can be determined, within tmQM-RDF it is always possible to determine a most specific class with respect to the subclass relation. This concept can be effectively synthesised by introducing a \say{compressed} representation $\tilde{p}$ of a pattern $p$, which amounts to a labelled graph derived from $p$ by deleting all triples in which the predicate is a subproperty of \texttt{rdf:type} and by labelling all the remaining nodes according to their most specific RDFS class. See \figref{example_compressed_pattern} for an example.

\begin{figure}[!t]
    \begin{center}
    \makebox[\textwidth][c]{%
    \begin{tabular}{c c}
    \begin{subfigure}[t]{0.45\linewidth}
        \includegraphics[width=\linewidth]{pattern1.pdf}
        \caption{}\figlabel{ax_example_comp_pattern_1}
    \end{subfigure} &
    \begin{subfigure}[t]{0.55\linewidth}
        \raisebox{0.65cm}{%
        \includegraphics[width=\linewidth]{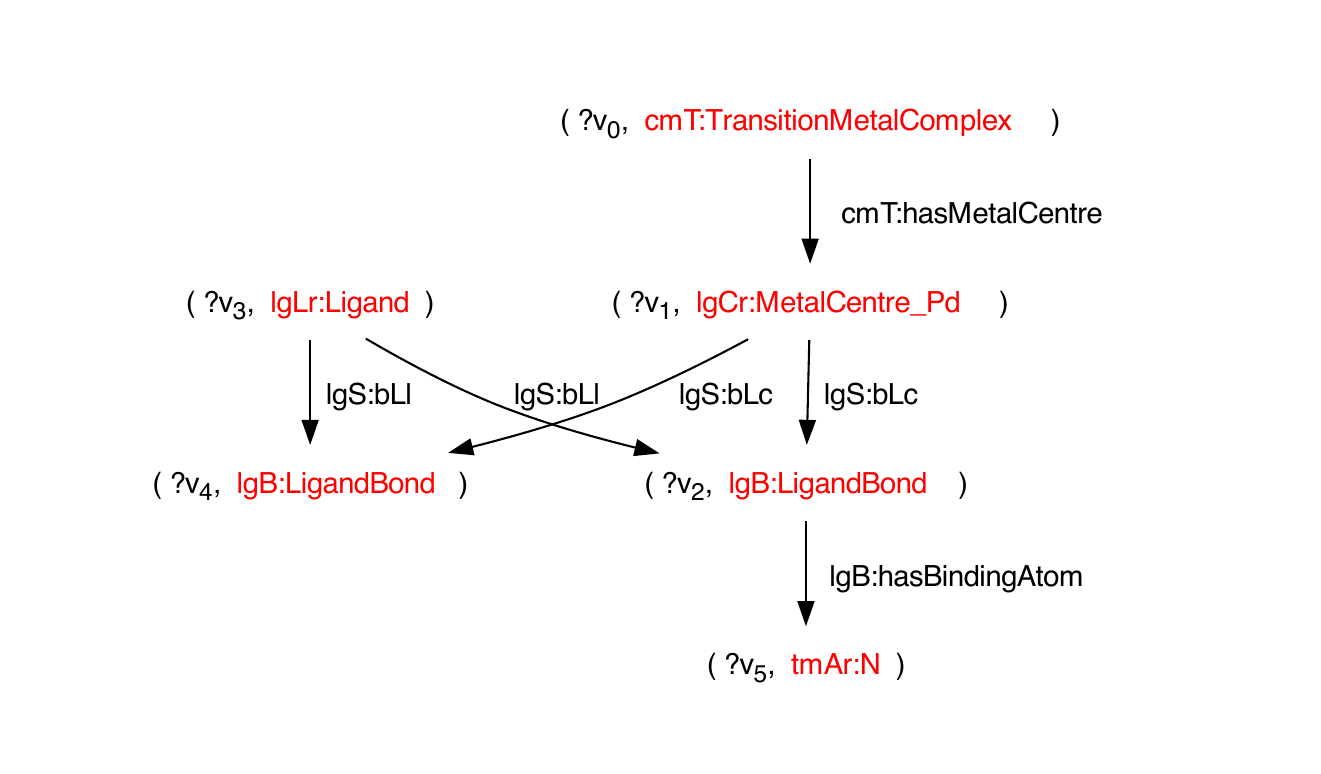}
        }
        \caption{}\figlabel{ax_example_comp_pattern_c1}
    \end{subfigure}%
    \end{tabular}
    }
    \end{center}
    
    \caption{(a) An example of a graph pattern. (b) Its corresponding compressed representation.}
    \figlabel{example_compressed_pattern}
\end{figure}

This new compressed representation enables the construction of a different kind of proxy for graph patterns. Consider the (logarithmically smoothed) tf-idf weighting scheme \citep{salton1988termweighting} defined as
\begin{equation*}
    \mathrm{tf-idf}(\ell, p, \mathcal{P}) = \frac{N_{\ell, p}}{\sum_{\ell' \in p}N_{\ell',p}} \cdot \left( 1 + \log\frac{|\mathcal{P}|}{1 + |\mathcal{P}_\ell|}\right)
\end{equation*} 
where $\ell$ is any of the labels appearing in the compressed representation $\tilde{p}$ of $p$, $N_{\ell,p}$ is the number of times that $\ell$ appears in $\tilde{p}$\footnote{In an RDF graph $G$, there exists a unique node (URI) representing an RDF class, and class assignment is performed via an edge that points to that URI. Hence, in compliance with the constraint imposing at most one node per label, there exists a unique instance of a class label in $G$. In the compressed representation, however, class assignement is performed via label assignment, hence there will be as many instances of a class as entities belonging to that class.}, $\mathcal{P}$ is the set of frequent patterns and $\mathcal{P}_\ell$ is the set of (compressed) patterns in which $\ell$ appears at least once. Now, if $L_\mathcal{P}$ is the set of all the labels that appear in the graphs of $\{\tilde{p} \mid p \in \mathcal{P}\}$, then the feature vector is
\[\bm{z}_p := \big(\mathrm{tf-idf}(\ell,p,\mathcal{P})\big)_{\ell \in L_\mathcal{P}}\]
and the corresponding similarity metric becomes
\[s_{\mathrm{cos};\,s}(p, q) := s_{\mathrm{cos}}(\bm{z}_p, \bm{z}_q).\]

\subsection{Similarity metric for DLGs}
Another interesting metric is the similarity metric for DLGs introduced by \citet{champin2003graphsimilarity}. This is a metric that strongly relies on explicit label assignment, hence it naturally applies to the compressed representation introduced above.

Given two patterns (DLGs) $p$ and $q$, and their compressed representations $\tilde{p} = (V_{\tilde{p}}, E_{\tilde{p}}, \pi_{\tilde{p}})$ and $\tilde{q} = (V_{\tilde{q}}, E_{\tilde{q}}, \pi_{\tilde{q}})$, the similarity of \citet{champin2003graphsimilarity} relies on a \emph{multivalent mapping} $m \subseteq V_{\tilde{p}} \times V_{\tilde{q}}$, which maps every vertex in any of the two patterns with one or more (or even zero) nodes in the other. Following \citet{champin2003graphsimilarity}, when two nodes $v \in V_{\tilde{p}}$ and $w \in V_{\tilde{q}}$ are mapped onto each other, we will denote it using the functional notation $m(v) = w$ or $m(w) = v$, depending on the context (but we consider the two expressions to be equivalent). Edges are mapped when their extremes are mapped. Each mapping $m$ induces an \say{intersection} $\tilde{p}\, \sqcap_m \tilde{q}$ made of the nodes and edges that are mapped onto each other and that share the same label. By assigning a score to nodes, edges and their labels, and by eventually penalising features that are \say{split} (i.e., mapped onto two or more objects), it is possible to evaluate the entire intersection. Finally, the similarity score is defined as the evaluation of the intersection induced by best possible mapping $m$. Formally,
\begin{equation}\eqnlabel{s_delg_def}
    s_\mathrm{DLG}(p, q\mid \omega_{f,V},\omega_{f,E},\omega_{g}) = \max_{m \subseteq V_{\tilde{p}} \times V_{\tilde{q}}} \frac{f(\tilde{p} \sqcap_m \tilde{q}\mid \omega_{f,V},\omega_{f,E}) - g(\text{splits}(m)\mid \omega_G)}{f(\tilde{p} \cup \tilde{q}\mid \omega_{f,V},\omega_{f,E})},
\end{equation}
where
\begin{align*}
    \tilde{p}\, \sqcap_m \tilde{q} &:= \bigcup_{\substack{\tilde{r},\,\tilde{s} \in \{\tilde{p}, \tilde{q}\}\\\tilde{r} \neq \tilde{s}}}\bigg(\big\{v \in V_{\tilde{r}} \mid \exists v' \in V_{\tilde{s}}:\, m(v) = v',\, \pi_{\tilde{r}}(v) = \pi_{\tilde{s}}(v')\big\}\\[-0.5cm]
    &\qquad\qquad\quad\cup\ \big\{(v, w) \in E_{\tilde{r}} \mid \exists (v', w') \in E_{\tilde{s}}:\, m(v) = v', m(w) = w' \in m,\pi_{\tilde{r}}(v,w) = \pi_{\tilde{s}}(v',w')\big\}\bigg),\\
    \text{splits}(m) &:= \bigcup_{\substack{\tilde{r},\,\tilde{s} \in \{\tilde{p}, \tilde{q}\}\\\tilde{r} \neq \tilde{s}}}\big\{(v, s_v) \mid v \in V_{\tilde{r}},\, s_v = \{v' \in V_{\tilde{s}}\mid m(v) = v'\},\, |s_v| \geq 2\big\}\\
    \tilde{p} \cup \tilde{q} &:= V_{\tilde{p}} \cup E_{\tilde{p}} \cup V_{\tilde{q}} \cup E_{\tilde{q}} 
\end{align*}
and
\begin{align*}
    f(F\mid \omega_{f,V},\omega_{f,E}) &:= \sum_{\tilde{r} \in \{\tilde{p}, \tilde{q}\}}\left(\sum_{v \in F \cap V_{\tilde{r}}}\omega_{f,V}(v, \pi_{\tilde{r}}(v)) + \sum_{(v, w) \in F \cap E_{\tilde{r}}}\omega_{f,E}(v, w, \pi_{\tilde{r}}(v,w))\right),\\
    g(S\mid \omega_{g}) &:= \sum_{(v, s_v) \in S}\omega_g(v, s_v).
\end{align*}
The choice of the weight functions $\omega_{f,V}, \omega_{f,E}$ and $\omega_g$ determine the behaviour of the similarity metric.

The optimisation problem in \eqnref{s_delg_def} can be approached via  greedy algorithm presented in \citet{champin2003graphsimilarity}. This solution, however, is non-deterministic, possibly suboptimal and particularly expensive to compute, especially for large sets of graph patterns.\\

\subsubsection{Naive weighting scheme}\apxlabel{ax_similarity_s_delg_naive}
Minimally informative weights can be chosen as
\begin{align*}
    \omega_{f,V}^\mathrm{naive}(v, \ell\mid\omega_0,\omega_1) &= \begin{cases}
        \omega_0&\quad\text{if } \ell\in\{\text{\texttt{lgCr:MetalCentre, lgLr:Ligand, tmA:Atom}}\}\\
        \omega_1&\quad\text{otherwise}
    \end{cases},\\
    \omega_{f,E}^\mathrm{naive}(v, w ,\ell\mid\omega_2) &= \omega_2,\\
    \omega_g(v, s_v) &= |s_v|,
\end{align*}
for some non-negative constants $\omega_0, \omega_1, \omega_2$ such that $\omega_0 < \omega_1$. Notice that $\pi(v)\in\{\text{\texttt{lgCr:Me\-tal\-Cen\-tre}},\allowbreak{}\- \text{\texttt{lgLr:Li\-gand}},\allowbreak{}\- \text{\texttt{tmA:A\-tom}}\}$ can only happen if $v$ represents a metal centre, a ligand or an atom and the pattern does not explicitly assign a class, meaning that the chemical identity of the node is left unspecified. We then define
\[s_{\mathrm{DLG};\,n}(p,q\mid\omega_0,\omega_1,\omega_2) := s_\mathrm{DLG}\left(p,q\mid\omega_{f,V}^\mathrm{naive}(\cdot,\cdot\mid\omega_0,\omega_1),\omega_{f,E}^\mathrm{naive}(\cdot, \cdot,\cdot\mid\omega_2), \omega_{g}\right).\]
In this case, the constraint $\omega_0 < \omega_1$ has the effect of penalising pairs of patterns that specify contradictory chemical identities for the same entity versus pairs in which one pattern does not specify any identity (see Example \ref{ex:s_delg}). 

\begin{example}\label{ex:s_delg}
    Consider the two compressed graph representations shown in \figrefmult{ax_example_pattern_c1}{ax_example_pattern_c2} and name them $\tilde{p}$ and $\tilde{q}$ respectively.

    \begin{figure}[!t]
    \begin{center}
    \makebox[\textwidth][c]{%
    \begin{tabular}{c c}
    \begin{subfigure}[t]{0.55\linewidth}
        \includegraphics[width=\linewidth]{compressed1.pdf}
        \caption{}\figlabel{ax_example_pattern_c1}
    \end{subfigure} &
    \begin{subfigure}[t]{0.55\linewidth}
        \includegraphics[width=\linewidth]{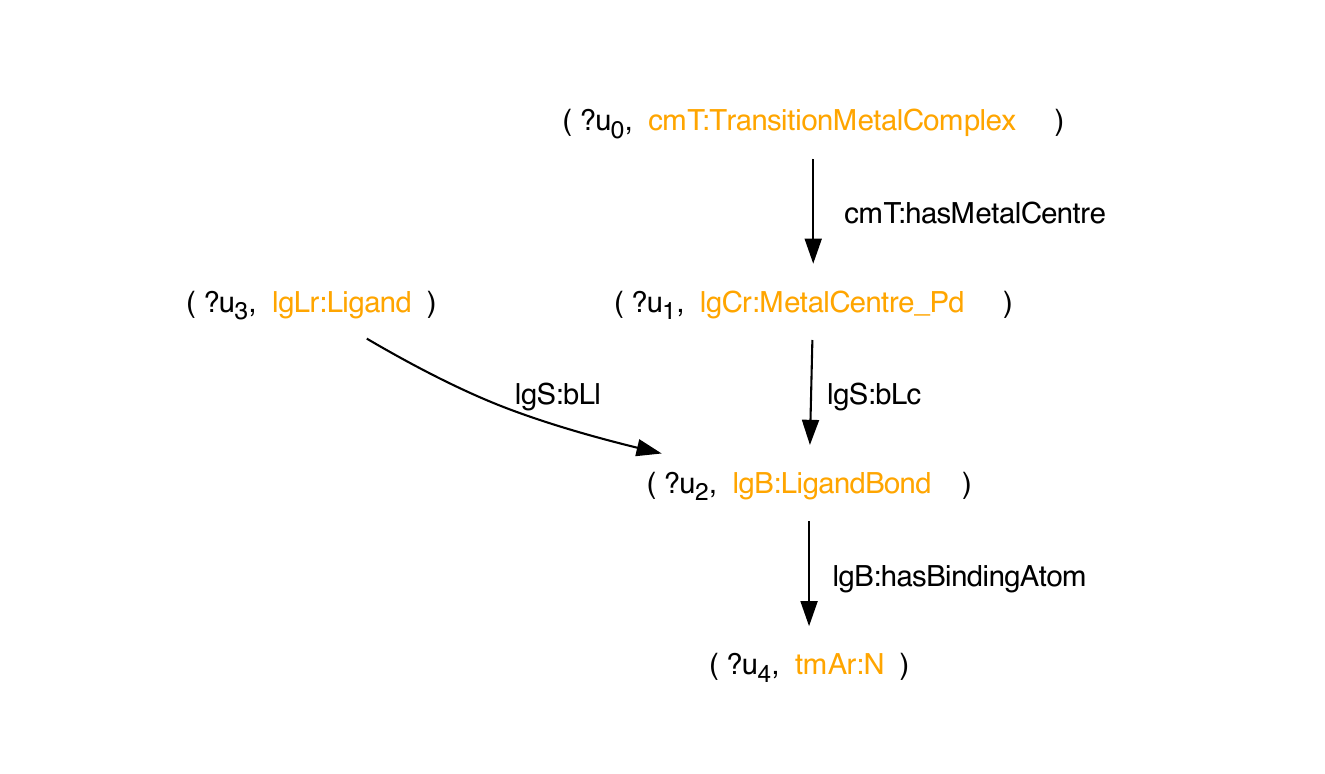}
        \caption{}\figlabel{ax_example_pattern_c2}
    \end{subfigure} \\
    \begin{subfigure}[t]{0.55\linewidth}
        \includegraphics[width=\linewidth]{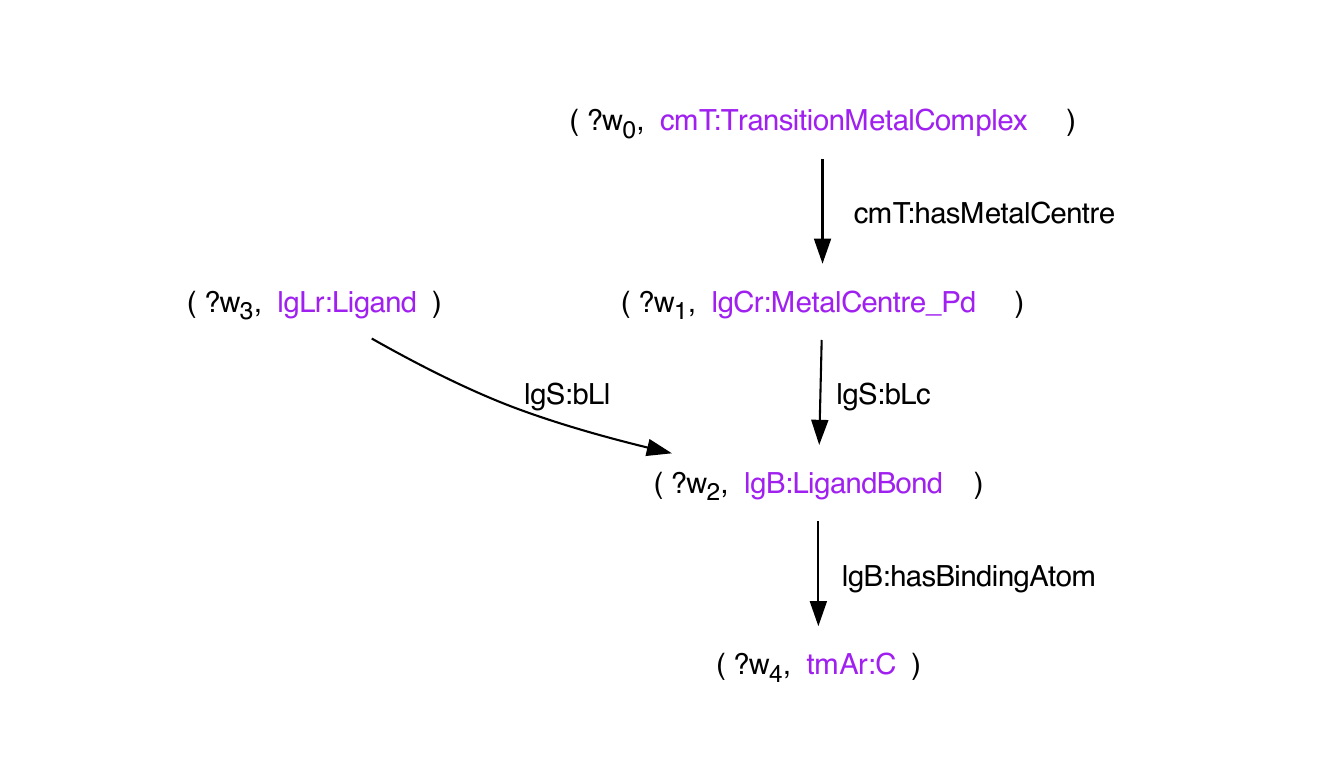}
        \caption{}\figlabel{ax_example_pattern_c3}
    \end{subfigure} &
    \begin{subfigure}[t]{0.55\linewidth}
        \includegraphics[width=\linewidth]{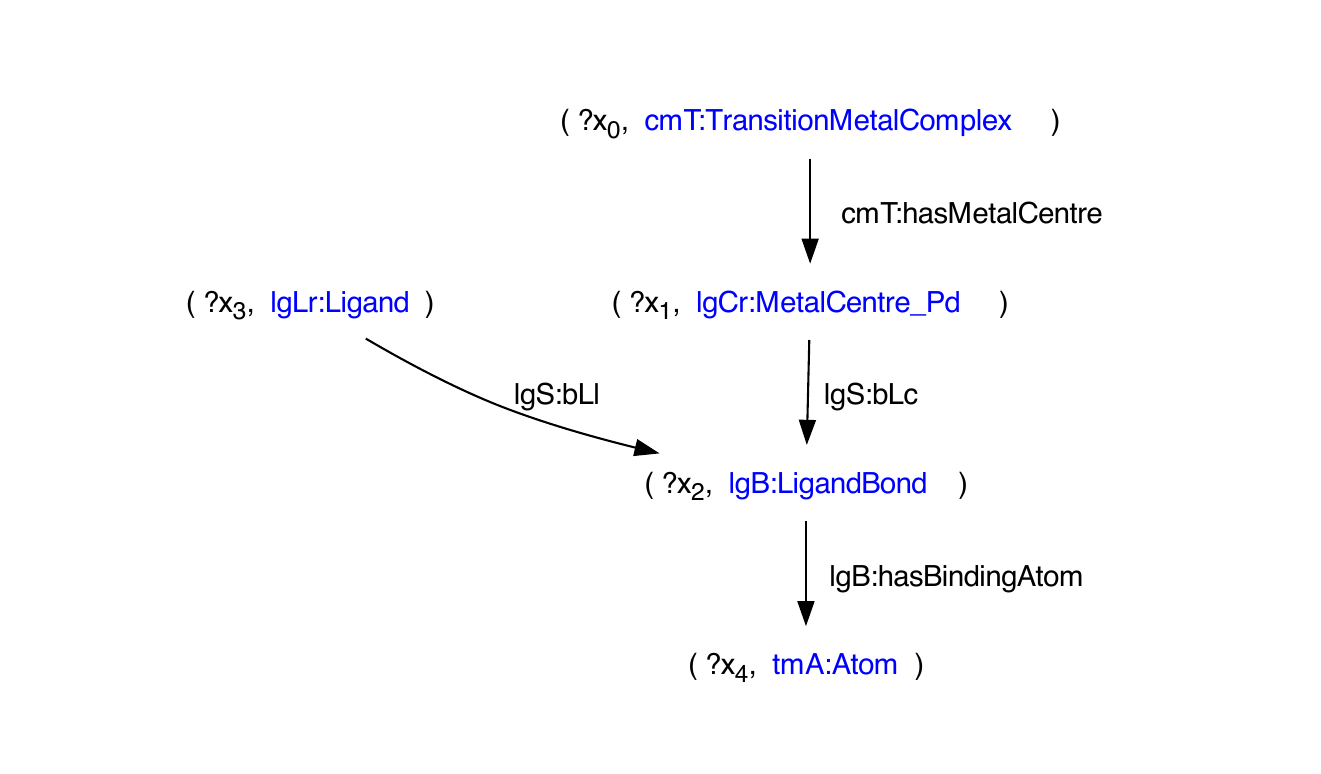}
        \caption{}\figlabel{ax_example_pattern_c4}
    \end{subfigure} \\
    \end{tabular}
    }
    \end{center}
    
    \caption{Four examples of compressed patterns. In Example \ref{ex:s_delg}, these are referred to as $\tilde{p},\, \tilde{q},\, \tilde{q}'$ and $\tilde{q}''$ respectively.}
    \figlabel{ax_example_pattern}
\end{figure}

    Consider then the following multivalent mapping:
    \begin{align*}
        v_0 &\mapsto u_0\\
        v_1 &\mapsto u_1\\
        v_2 &\mapsto u_2\\
        v_3 &\mapsto u_3\\
        v_4 &\mapsto u_2\\
        v_5 &\mapsto u_4.
    \end{align*}

    This mapping essentially preserves the identity and role of each node and deals with the different denticity of the specified ligand by collapsing the two bond objects found in $\tilde{p}$ onto the single bond object in $\tilde{q}$.

    If we evaluate the objective function of \eqnref{s_delg_def} in this scenario we find that
    \begin{align*}
        \tilde{p}\, \sqcap_m \tilde{q} &= V_{\tilde{p}} \cup E_{\tilde{p}} \cup V_{\tilde{q}} \cup E_{\tilde{q}}, \\
        f(\tilde{p}\, \sqcap_m \tilde{q}) &= \omega_1(|\tilde{p} \sqcap_m \tilde{q}| - 2) + 2\omega_0,\\
        \text{splits}(m) &= \{(u_2, \{v_2, v_4\})\},\\
        g(\text{splits}(m)) &= 2,\\
        f(\tilde{p} \cup \tilde{q}) &= \omega_1(|\tilde{p} \sqcap_m \tilde{q}| - 2) + 2\omega_0
    \end{align*}
    where the second and the fourth equation are based on the fact that the intersection (via $m$) covers the entire union of the two patterns and all the nodes, except two, are assigned a specific chemical identity (when needed). All nodes and edges therefore receive the weight $\omega_1$, except for the two unspecified nodes (the ligands) which receive the weight $\omega_0$.
    For $\omega_0 = 0.5, \omega_1 = 1$, we get the candidate score
    \[\frac{\omega_1(|\tilde{p}\, \sqcap_m \tilde{q}| - 2) + 2\omega_0 - 2}{\omega_1(|\tilde{p}\, \sqcap_m \tilde{q}| - 2) + 2\omega_0} = \frac{1\cdot (21 - 2) + 2\cdot0.5 - 2}{1\cdot (21 - 2) + 2\cdot0.5} = \frac{18}{20} = 0.9.\]
    A computational examination of the problem reveals also that the mapping hereby proposed is locally optimal.\\

    Consider now a slightly different case, namely the case in which $\tilde{p}$ is compared first against the graph $\tilde{q}'$, represented in \figref{ax_example_pattern_c3}, and then against $\tilde{q}''$, \figref{ax_example_pattern_c4}. Both graphs are almost identical to $\tilde{q}$, except for the label of the nodes that correspond to $u_4$ in $\tilde{q}$, namely $w_4$ in $\tilde{q}'$ and $x_4$ in $\tilde{q}''$.

    First, $w_4$ is labelled as \texttt{tmAr:C}. Under the same mapping used before (with $u_i$ replaced by $w_i$), now it holds that
    \[\tilde{p}\, \sqcap_m \tilde{q}' = V_{\tilde{p}} \cup E_{\tilde{p}} \cup V_{\tilde{q}'} \cup E_{\tilde{q}'} \setminus\{v_5, u_4\}\]
    and the candidate score becomes
     \[\frac{\omega_1(|\tilde{p}\, \sqcap_m \tilde{q}'| - 4) + 2\omega_0 - 2}{\omega_1(|\tilde{p}\, \sqcap_m \tilde{q}'| - 2) + 2\omega_0} = \frac{1\cdot (21 - 4) + 2\cdot0.5 - 2}{1\cdot (21 - 2) + 2\cdot0.5} = \frac{16}{20} = 0.8.\]
    Notice that the denominator did not change with respect to the previous case, as the label change did not alter the number of specified chemical identities.

    Now, in $\tilde{q}''$, the node $x_4$ is labelled \texttt{tmA:Atom}. The intersection and the associated score are the same as before, but this time the denominator changes to
    \[f(\tilde{p} \cup \tilde{q}'') = \omega_1(|\tilde{p} \sqcap_m \tilde{q}''| - 3) + 3\omega_0, \]
    as now there are three unspecified chemical identities. The similarity score in this case becomes
     \[\frac{\omega_1(|\tilde{p} \sqcap_m \tilde{q}''| - 4) + 2\omega_0 - 2}{\omega_1(|\tilde{p} \sqcap_m \tilde{q}''| - 3) + 3\omega_0} = \frac{1\cdot (21 - 4) + 2\cdot0.5 - 2}{1\cdot (21 - 3) + 3\cdot0.5} = \frac{16}{19.5} = 0.82.\]

     This example then showcases the behaviour of the chosen weight functions, i.e., they penalise more two patterns which specify wrong identities with respect to two patterns in which one expresses uncertainty in terms of the identity,
\end{example}

\subsubsection{Learned semantic weights}
An alternative and more refined weighting scheme can however be devised. In particular, the weights should take two different principles into account. There is prior ontological information ready to be exploited: specific RDFS classes are more informative than general ones. In addition, there is empirical information that can be learned from the pattern dataset, in the sense that the frequency of a label can be an indicator of its relevance (for instance, a rare ligand class may be more informative than a frequent one). These two principles can be accounted for by means of two separate mechanisms.

The ontological hierarchy of the labels can be defined by defining the weight of an identity-specifying label as the sum of two components:
\[\omega^\mathrm{learned}(\ell) = \begin{cases}
    \omega_\mathrm{gen}(\text{general}(\ell))&\qquad\text{if general}(\ell) = \ell\\
    \omega_\mathrm{gen} + \omega_0(\ell)&\qquad\text{otherwise}
\end{cases},\]
where $\omega_0$ is an arbitrary (non negative) weighting function and $\text{general}(\ell)$ is the corresponding general-purpose RDFS class (e.g. if $\ell =\,$\texttt{tmAr:C}, then $\text{general}(\ell)=\,$\texttt{tmAr:Atom}). By doing so, we can guarantee that specific labels will always have more influence than generic labels (recall the mechanism showcased in Example \ref{ex:s_delg}).

Now, empirical information can be incorporated by defining the weighting scheme $\omega_0$ as a data-driven inverse document frequency weight:
\[\omega_\mathrm{gen}(\ell) = 1 + \log\frac{|\mathcal{P}|}{1 + |\mathcal{P}_{\ell}|},\]
where, once again, $\mathcal{P}$ is the set of frequent patterns and $\mathcal{P}_\ell$ is the set of (compressed) patterns in which $\ell$ appears at least once. It is easy to see that $\omega_0(\ell)$ decreases as $|\mathcal{P}_\ell|$ increases.\\

Using these two weighting functions, we can define node and edge weighting functions
\begin{align*}
    \omega_{f,V}^\mathrm{learned}(v, \ell) &= \omega^\mathrm{learned}(\ell)\\
    \omega_{f,E}^\mathrm{learned}(v, w ,\ell) &= \omega^\mathrm{learned}(\ell),
\end{align*}

which finally allows us to define
\[s_{\mathrm{DLG};\,l}(p,q\mid\omega^\mathrm{learned}) := s_\mathrm{DLG}(p,q\mid\omega_{f,V}^\mathrm{learned},\omega_{f,E}^\mathrm{learned}, \omega_{g})\]
where $\omega_g(v, s_v) = |s_v|$, as in \apxref{ax_similarity_s_delg_naive}.

\section{Computational parameters}\apxlabel{parameters}
We report in \tblref{parameters} the list of the values of all the parameters used in this work during the experiments described in \sctref{experiments} and \apxref{experiments_methods}.

{
\newcommand{\locspace}{0.1cm}
\begin{table}[]
    \makebox[\textwidth][c]{%
    \centering
    
    \begin{tabular}{l l l l}
        \toprule
         Parameter & Description & Reference & Value \\
         \midrule
         \midrule
         \multicolumn{4}{c}{\emph{Dataset selections}}\\
         \midrule
         $N_\mathrm{seed}$ & Number of ligands in the seed & \apxref{selection} & \makecell[tl]{\emph{earlyTM:} $1350$\\\emph{lateTM:} $350$\\[\locspace]}\\
         \multicolumn{4}{c}{\emph{Frequent pattern mining and filtering}}\\
         \midrule
         $\alpha$ & Minimum number of matched graphs & \apxref{experiments_pattern_mining} &  \makecell[tl]{\emph{earlyTM:} $10$\\\emph{lateTM:} $20$\\[\locspace]}\\
         $p_0$ & Seed pattern & \apxref{experiments_pattern_mining} & \makecell[tl]{$(?v_0,\, t_p,\, ?v_1)$, with\\$t_p = \texttt{cmT:hasMetalCentre}$ \\[\locspace]}\\
         $M_\mathrm{max}$ & Maximum pattern size& \apxref{pattern_mining} &  $12$\\[\locspace]
         $(\pi_i)_{i=2}^{M_\mathrm{max}}$ & Rejection probabilities & \apxref{pattern_mining} &  \makecell[tl]{\emph{earlyTM:} $(0.1,\, 0.15,\,0.2,\,0.25$\\$\qquad\qquad\quad0.3,\,0.4,\,0.65,\,0.72,$\\$\qquad\qquad\quad0.895,\,0.9,\,0.9)$\\\emph{lateTM:} $(0.2,\, 0.25,\,0.25,\,0.42$\\$\qquad\qquad\quad0.45,\,0.45,\,0.53,\,0.8,$\\$\qquad\qquad\quad0.8,\,0.798,\,0.8)$\\[\locspace]}\\
         $M_\mathrm{min}$ & Minimum pattern size& \apxref{pattern_filtering} &  $10$\\[\locspace]
         \multicolumn{4}{c}{\emph{Agglomerative clustering}}\\
         \midrule
         $\Delta$ & Set of candidate similarity thresholds& \apxref{experiments_clustering} &  $\left\{\frac{i}{200}\mid i = 0,\dots,200\right\}$\\[\locspace]
         $M^{\mathcal{C}}_\mathrm{min}$ & Minimum number of clusters & \apxref{experiments_clustering} &  $300$\\[\locspace]
         $(\omega_0,\, \omega_1, \omega_2)$ & Weights of $s_{\mathrm{DLG};\,n}$ & \apxref{ax_similarity_s_delg_naive} & $(0.5,\, 1,\,1)$ \\[\locspace]
         \multicolumn{4}{c}{\emph{Bayesian Network}}\\
         \midrule
         $\lambda_0$ & Dirichlet prior hyperparameters & \apxref{experiments_bn_parameters} & $10$ \\[\locspace]
         \bottomrule
    \end{tabular}
    
    }
    \caption{The list of all the numerical parameters employed in this work, organised by topic, together with the value used. Unless otherwise specified, the same value of each parameter has been used across the experiments involving the two selection \emph{earlyTM} and \emph{lateTM}.}
    \tbllabel{parameters}
\end{table}
}

\section{Extensive experimental results}\apxlabel{extensive_results}
We have performed the experiment from \sctref{experiments} using all the $8$ possible combinations of similarity metric and aggregation function to compute the feature matrix $Y$, as described in \apxref{experiments_methods}.

We briefly recall that the goal of the experiment was to learn a score function $f: \mathcal{G} \to [0,\,1]$ which could be used as a means of evaluating the outcomes of elementary TMC manipulations. In particular, during the test phase we removed one random ligand from each TMC in the test population $\mathcal{G}_\mathrm{test}$ and we then tried to reconstruct the original compound by computing all the possible reconstructions that can be formed using ligands from the training population. We then used $f$ to rank all the possible reconstructions of each incomplete complex. The quality of the learned function $f$ is assessed by measuring the top$-k$ accuracy, for $k \in \{1,\, 5,\, 10\}$, intended as the fraction of test TMCs that can be found among the top$-k$ highest scoring reconstructions that can be formed after one of their ligands is removed.

The full set of results is shown in \tblref{extensive_results}.

\subsection{Discussion}\apxlabel{extensive_discussion}
In the case of the \emph{earlyTM} dataset, it emerged that the proposed method of learning a score function $f$ is capable of achieving promising results already for $k = 1$. As reported in \tblref{extensive_results}, the average top$-1$ accuracy across all configurations amounts to $0.649$, $0.7245$ and $0.672$ depending on the filter applied to the set of ligands available for the reconstruction (respectively: no filter, same denticity/hapticity order as the missing ligand, same charge as the missing ligand). These values increase remarkably at $k = 10$, with the average accuracy reporting at $0.874$, $0.975$, and $0.905$ respectively.

For the \emph{lateTM} dataset, on the other hand, the situation appears to be different. The top$-1$ accuracy is particularly low, with average values of $0.280$, $0.324$ and $0.294$. Still, at $k = 10$ we can observe similar performances to the \emph{lateTM} dataset, with average accuracy values of $0.806$, $0.889$ and $0.910$.

Overall, given that each molecular scaffold admits hundreds of possible reconstructions within this framework, the high accuracy values found at the relatively restrictive threshold of $k = 10$ advocate in favour of the soundness of our approach to structure assessment, as well as the usefulness of the unified view, at different resolutions, proposed in tmQM-RDF.\\

We also observe that, while it is possible to appreciate a similar trend across all configurations and commensurable performances, there is not a single configuration that stands out as being particularly superior across all scenarios. Nevertheless, the pair $(s_{\mathrm{cos};\,p},\,a_\mathrm{median})$ achieves excellent performances, relative to the other settings, within the \emph{lateTM} dataset. In \emph{earlyTM}, this pair is often outperformed by other configurations, but the improvements are contained in magnitude (the largest difference in accuracy amounts to $0.051$). Given the particularly restrictive nature of the $a_{median}$ function, as discussed in \apxref{experiments_graph_features}, this is perhaps the strongest indicator that the assumption of the existence of few fundamental structural modalities is deserving of further investigations.

The pair $(s_{\mathrm{cos};\,p},\,a_\mathrm{max})$ also displays a near-optimal behaviour in \emph{lateTM}, although not in \emph{earlyTM}. This reason, paired with the intuitive interpretation of the $a_\mathrm{max}$ aggregation function, still prompted us to choose this last pair as the representative configuration whose results have been discussed in \sctref{experiments}.

\begin{table}[t]
    \centering
    \makebox[\textwidth][c]{%
    
    \begin{tabular}{c c c c c c c c c c c c c c c}
         \toprule
Metric & Variant & Aggregation & & \multicolumn{11}{c}{Filter}\\
\cmidrule(lr){5-15}
 &  &  & & \multicolumn{3}{c}{\emph{No filter}} & & \multicolumn{3}{c}{\emph{Hapticity/denticity order}} & & \multicolumn{3}{c}{\emph{Charge}}\\
\cmidrule(lr){5-7}\cmidrule(lr){9-11}\cmidrule(lr){13-15}
 &  &  & & \multicolumn{3}{c}{Top$-k$ accuracy} & & \multicolumn{3}{c}{Top$-k$ accuracy} & & \multicolumn{3}{c}{Top$-k$ accuracy}\\
\cmidrule(lr){5-7}\cmidrule(lr){9-11}\cmidrule(lr){13-15}
 &  &  &  & $\emph{k = 1}$ & $\emph{k = 5}$ & $\emph{k = 10}$ &  & $\emph{k = 1}$ & $\emph{k = 5}$ & $\emph{k = 10}$ &  & $\emph{k = 1}$ & $\emph{k = 5}$ & $\emph{k = 10}$\\
\midrule
\midrule
\multicolumn{15}{c}{\emph{earlyTM}}\\
\midrule
$s_{cos}$ & $s_{cos;\,p}$ & $a_{max}$ &  & 0.619 & 0.826 & 0.866 &  & 0.702 & 0.926 & 0.967 &  & 0.649 & 0.846 & 0.896\\
	 & 	 & $a_{median}$ &  & 0.662 & 0.826 & 0.863 &  & 0.732 & $\mathbf{0.946}$ & $\mathbf{0.983}$ &  & 0.696 & 0.846 & 0.896\\
	 & $s_{cos;\,s}$ & $a_{max}$ &  & 0.649 & 0.816 & 0.866 &  & 0.722 & 0.933 & 0.973 &  & 0.656 & 0.833 & $\emph{0.906}$\\
	 & 	 & $a_{median}$ &  & $\mathbf{0.692}$ & 0.833 & 0.860 &  & $\emph{0.753}$ & 0.920 & $\emph{0.980}$ &  & $\emph{0.699}$ & 0.839 & $\emph{0.906}$\\
$s_{DELG}$ & $s_{DELG;\,n}$ & $a_{max}$ &  & $\emph{0.689}$ & $\emph{0.849}$ & $\mathbf{0.890}$ &  & $\mathbf{0.783}$ & 0.936 & 0.977 &  & $\mathbf{0.726}$ & $\emph{0.860}$ & $\mathbf{0.913}$\\
	 & 	 & $a_{median}$ &  & 0.629 & 0.839 & $\emph{0.886}$ &  & 0.686 & 0.936 & 0.973 &  & 0.649 & $\emph{0.860}$ & $\mathbf{0.913}$\\
	 & $s_{DELG;\,l}$ & $a_{max}$ &  & 0.632 & $\mathbf{0.853}$ & $\emph{0.886}$ &  & 0.736 & $\emph{0.940}$ & 0.973 &  & 0.659 & $\mathbf{0.870}$ & $\emph{0.906}$\\
	 & 	 & $a_{median}$ &  & 0.622 & 0.826 & 0.873 &  & 0.689 & $\emph{0.940}$ & 0.977 &  & 0.645 & 0.843 & 0.900\\[0.1cm]
\multicolumn{15}{c}{\emph{lateTM}}\\
\midrule
$s_{cos}$ & $s_{cos;\,p}$ & $a_{max}$ &  & $\emph{0.357}$ & $\mathbf{0.667}$ & $\emph{0.833}$ &  & $\emph{0.403}$ & $\mathbf{0.767}$ & $\emph{0.903}$ &  & $\mathbf{0.377}$ & $\emph{0.770}$ & $\emph{0.920}$\\
	 & 	 & $a_{median}$ &  & $\mathbf{0.363}$ & $\mathbf{0.667}$ & $\mathbf{0.853}$ &  & $\mathbf{0.407}$ & 0.743 & $\mathbf{0.923}$ &  & $\emph{0.367}$ & $\mathbf{0.797}$ & $\mathbf{0.927}$\\
	 & $s_{cos;\,s}$ & $a_{max}$ &  & 0.193 & 0.633 & 0.810 &  & 0.227 & 0.730 & $\emph{0.903}$ &  & 0.200 & 0.723 & 0.893\\
	 & 	 & $a_{median}$ &  & 0.213 & 0.640 & 0.807 &  & 0.250 & 0.733 & 0.883 &  & 0.237 & 0.737 & 0.903\\
$s_{DELG}$ & $s_{DELG;\,n}$ & $a_{max}$ &  & 0.313 & 0.660 & 0.797 &  & 0.373 & 0.733 & 0.873 &  & 0.330 & 0.740 & 0.917\\
	 & 	 & $a_{median}$ &  & 0.287 & $\emph{0.663}$ & 0.790 &  & 0.333 & $\emph{0.747}$ & 0.880 &  & 0.303 & 0.733 & 0.910\\
	 & $s_{DELG;\,l}$ & $a_{max}$ &  & 0.260 & 0.633 & 0.767 &  & 0.307 & 0.713 & 0.860 &  & 0.263 & 0.707 & 0.897\\
	 & 	 & $a_{median}$ &  & 0.253 & 0.617 & 0.793 &  & 0.293 & 0.713 & 0.887 &  & 0.277 & 0.707 & 0.910\\
\bottomrule
    \end{tabular}
    }
    \caption{Top$-k$ accuracy measured during the TMC completion task for the two datasets employed and each of the $8$ possible combinations of similarity metric and aggregation function. In each column, the highest and second-highest value are highlighted in boldface and italic respectively.}
    \tbllabel{extensive_results}
\end{table}

\end{artappendix}
\end{document}

\typeout{get arXiv to do 4 passes: Label(s) may have changed. Rerun}